\documentclass[11pt,a4paper]{amsart}
\usepackage[foot]{amsaddr}
\usepackage{ifxetex}
\ifxetex
  \usepackage[no-math]{fontspec}
\else
\fi
\usepackage{amsmath}
\usepackage{amsfonts}
\usepackage{amssymb}
\usepackage{amsthm}
\usepackage{fullpage}
\usepackage{microtype}
\usepackage[libertine]{newtxmath}
\usepackage[tt=false]{libertine} 
\usepackage{caption}
\usepackage{bbm}
\usepackage{hyperref, color}
\hypersetup{colorlinks=true,citecolor=blue, linkcolor=blue, urlcolor=blue}
\usepackage[linesnumbered,boxed,ruled,vlined]{algorithm2e}
\usepackage{bm}
\usepackage{bbm}
\usepackage[numbers]{natbib}
\usepackage{xcolor}
\usepackage{enumerate} 
\usepackage{enumitem}
\usepackage{tabularx}
\usepackage{array}
\usepackage{cleveref}
\newcolumntype{L}[1]{>{\raggedright\arraybackslash}p{#1}}
\newcolumntype{C}[1]{>{\centering\arraybackslash}m{#1}}
\newcolumntype{R}[1]{>{\raggedleft\arraybackslash}p{#1}}

\usepackage{makecell}
\usepackage{footnote}
\makesavenoteenv{tabular}

\newcommand{\Sat}{\#\textsc{Sat}}

\newcommand{\DTV}[2]{d_{\mathrm{TV}}\left({#1},{#2}\right)}
\newcommand{\dist}{\mathrm{dist}}

\renewcommand{\epsilon}{\varepsilon}

\newcommand{\Lin}{\mathrm{Lin}}
\newcommand{\Vcol}{V_{\mathrm{set}}}

\newcommand{\pup}{p_{\mathrm{up}}}
\newcommand{\plow}{p_{\mathrm{low}}}

\newcommand{\Xalg}{X_{\mathsf{alg}}}

\newcommand{\vbl}[1]{\mathsf{vbl}\left(#1\right)}

\newcommand{\pfailed}{p_{\mathsf{failed}}}

\newcommand{\sample}{\textnormal{\textsf{Sample}}}

\newcommand{\edge}{\mathcal{E}}

\newcommand{\Tmix}{\left\lceil 2n \log \frac{4n}{\epsilon} \right\rceil}

\newcommand{\tmix}{T_{\textsf{mix}}}
\newcommand{\Glauber}{P_{\textsf{Glauber}}}

\newtheorem{theorem}{Theorem}[section]
\newtheorem{observation}[theorem]{Observation}
\newtheorem{claim}[theorem]{Claim}
\newtheorem*{claim*}{Claim}
\newtheorem{condition}[theorem]{Condition}

\newtheorem{lemma}[theorem]{Lemma}
\newtheorem{proposition}[theorem]{Proposition}
\newtheorem{corollary}[theorem]{Corollary}
\theoremstyle{definition}

\newtheorem{definition}[theorem]{Definition}

\newtheorem*{remark*}{Remark}

\newcommand{\toolarge}{dk\log\frac{n}{\delta}}
\newcommand{\repeattime}{\left\lceil\left( \frac{n}{\delta} \right)^{\frac{\eta}{10}} \log \frac{n}{\delta} \right\rceil }
\newcommand{\Rejection}{\textsf{RejectionSampling}}

\def\Pr{\mathop{\mathbf{Pr}}\nolimits}

\renewcommand{\emptyset}{\varnothing}

\newcommand{\abs}[1]{\left\vert#1\right\vert}
\newcommand{\set}[1]{\left\{#1\right\}}
 \newcommand{\tuple}[1]{\left(#1\right)} \newcommand{\eps}{\varepsilon}
 
 \newcommand{\tp}{\tuple}

\newcommand{\numP}{\#\mathbf{P}} 
\newcommand{\defeq}{\triangleq}

\def\*#1{\mathbf{#1}} 
\def\+#1{\mathcal{#1}} 
\def\-#1{\mathrm{#1}} 

\usepackage{todonotes}

\usepackage{xifthen}

\renewcommand{\Pr}[2][]{ \ifthenelse{\isempty{#1}}
  {\mathbf{Pr}\left[#2\right]} {\mathbf{Pr}_{#1}\left[#2\right]} } 
\newcommand{\E}[2][]{ \ifthenelse{\isempty{#1}}
  {\mathbf{\mathbf{E}}\left[#2\right]}
  {\mathbf{\mathbf{E}}_{#1}\left[#2\right]} }
  \newcommand{\Var}[2][]{ \ifthenelse{\isempty{#1}}
  {\mathbf{\mathbf{Var}}\left[#2\right]}
  {\mathbf{\mathbf{Var}}_{#1}\left[#2\right]} }

\newcommand{\one}[1]{\mathbf{1}\left[#1\right]}
\title{Fast sampling and counting $k$-Sat solutions in the local lemma regime}
\author{Weiming Feng}
\author{Heng Guo}
\author{Yitong Yin}
\author{Chihao Zhang}

\address[Weiming Feng, Yitong Yin]{State Key Laboratory for Novel Software Technology, Nanjing University, 163 Xianlin Avenue, Nanjing, Jiangsu Province, China. \textnormal{E-mail: \url{fengwm@smail.nju.edu.cn, yinyt@nju.edu.cn}}}

\address[Heng Guo]{School of Informatics, University of Edinburgh, Informatics Forum, Edinburgh, EH8 9AB, United Kingdom. \textnormal{E-mail: \url{hguo@inf.ed.ac.uk}}}

\address[Chihao Zhang]{John Hopcroft Center for Computer Science, Shanghai Jiao Tong University, 800
  Dongchuan Road, Minhang District, Shanghai, China. \textnormal{E-mail: \url{chihao@sjtu.edu.cn}}}

\begin{document}
\begin{abstract}
We give new algorithms based on Markov chains to sample and approximately count satisfying assignments to $k$-uniform CNF formulas where each variable appears at most $d$ times.
For any $k$ and $d$ satisfying $kd<n^{o(1)}$ and $k\ge 20\log k + 20\log d + 60$, the new sampling algorithm runs in close to linear time, and the counting algorithm runs in close to quadratic time.

Our approach is inspired by Moitra (JACM, 2019) which remarkably utilizes the Lov\'{a}sz local lemma in approximate counting.
Our main technical contribution is to use the local lemma to bypass the connectivity barrier in traditional Markov chain approaches,
which makes the well developed MCMC method applicable on disconnected state spaces such as SAT solutions.
The benefit of our approach is to avoid the enumeration of local structures and obtain fixed polynomial running times, even if $k=\omega(1)$ or $d=\omega(1)$.
\end{abstract}

\thanks{This research was supported by the National Key R\&D Program of China 2018YFB1003202 and the NSFC Nos. 61722207, 61672275 and 61902241.}
\thanks{Part of the research was done when Weiming Feng was visiting the University of Edinburgh.}
\maketitle


\section{Introduction}

Sampling from an exponential-sized solution space and estimating the number of feasible solutions are two very related fundamental computation problems.
The Markov chain Monte Carlo (MCMC) method is the most successful technique due to its generic nature and the fast running time,
with many famous applications such as~\cite{DFK91,jerrum2004polynomial}.
A basic requirement for the method to apply is that the state space has to be connected via moves of the Markov chain to let the chain converge to the desired distribution.
This requirement prevents us from applying the method to the problems where the solution space is not connected via local moves.
Unfortunately, this barrier holds for perhaps the most important solution space in Computer Science: the satisfying assignments of \emph{conjunctive normal form} (CNF) formulas~\cite{wigderson2019book}.

Recently, a number of new methods based on the variable framework of the Lov\'asz local lemma were proposed to tackle the problem~\cite{Moi19,GJL19}. 
Most notably, the breakthrough of~\cite{Moi19} introduced a novel approach for estimating the number of solutions of $k$-SAT in a local lemma regime.
By far, it is still the only tractable result for sampling and approximately counting $k$-SAT solutions in the local lemma regime without additional structural assumptions on the formulas. 
However, since this new algorithm relies on local enumeration, its time cost is in the form of $n^{O(d^2k^2)}$, where $d$ is the variable degree in the local lemma.
Although a polynomial time for constant $d$ and $k$, this time cost is not fixed-parameter tractable with parameters $d$ and $k$.
Indeed, for $d=\omega(1)$ or $k=\omega(1)$, the running time becomes super-polynomial.

In this paper, we develop a new approach to overcome the connectivity barrier for Markov chain methods.
The main idea is to sample from the marginal probability of an algorithmically chosen subset of variables,
so that the standard Glauber dynamics is now ergodic.
However, this distribution is not a Gibbs distribution nor satisfies any kind of conditional independence properties.
New challenges arise as both analyzing and implementing the Glauber dynamics require new ideas.
We give a high-level overview of the techniques in \Cref{sec:tech}.


To illustrate the new technique, we choose a canonical $\numP$-complete problem,
namely counting the number of satisfying assignments of CNF formulas ($\Sat$) as our main application.
We call a CNF formula $\Phi$ a $(k,d)$-formula if all of its clauses have size $k$ and each variable appears in at most $d$ clauses.

%

\begin{theorem}[simplified]
  \label{theorem-sample-simplified}
  The following holds for any sufficiently small $\zeta>0$.

  There is an algorithm such that given any $0<\eps<1$ and $(k,d)$-formula $\Phi$ with $n$ variables where  $k \geq 20\log k + 20\log d + 3\log\tp{\frac{1}{\zeta}}$,
  it terminates in time $\widetilde{O}\left(d^2k^3 n \left(\frac{n}{\epsilon}\right)^{\zeta} \right)$ 
  and outputs a random assignment $X$ that is $\eps$ close in total variation distance to the uniform distribution over satisfying assignments of $\Phi$. 
  
  Moreover, there is an algorithm that given any $0<\eps<1$ and $(k,d)$-formula $\Phi$, under the same assumption, 
  terminates in time $\widetilde{O}\left( d^{3}k^3\left(\frac{n}{\epsilon}\right)^{2+\zeta}\right)$ 
  and outputs an $\textrm{e}^{\eps}$-approximation to the number of satisfying assignments of $\Phi$.
  In the above,  $\widetilde{O}(\cdot)$ hides a factor of $\mathrm{polylog}(n, d, \frac{1}{\epsilon})$.
\end{theorem}

The formal statements, with explicit running time bounds, are given in \Cref{theorem-sampling-main} (for sampling) and in \Cref{theorem-counting-main} (for counting).

\begin{remark*}
  Our algorithms in \Cref{theorem-sample-simplified} have unusual running time bounds that are controlled by a parameter $\zeta$.
  The parameter $\zeta$ cannot be too large.
  In fact, it must be no greater than $2^{-20}$, which implies that $k$ is at least $60$.
  As $\zeta$ gets smaller, the condition we require becomes stronger,
  but the sampling and counting algorithms run closer to linear and quadratic time, respectively.
  This is somewhat similar to algorithms for the Lov\'asz local lemma, 
  where the running time increases as the slack of the condition goes to $0$.
  
  In particular, if we set $\zeta = 2^{-20}$,
  the condition becomes $k \geq 20\log k + 20\log d + 60$. 
  The running time of our algorithm is a fixed polynomial in $n$, $\frac{1}{\epsilon}$, $d$, and $k$.
  Besides, for example, the exponent of $n$ is $1 + \zeta$ for sampling, which is very close to 1.
  In contrast,
  Moitra's algorithms \cite{Moi19} for both counting and sampling require a stronger condition $k \geq 60\log d + 60\log k + 300$,
  and run in time $\left(\frac{n}{\epsilon} \right)^{O(d^2k^2)}$.
  Our algorithms are much faster and remain in polynomial-time even if $k$ or $d$ is as large as $\Omega(n)$.
  Nonetheless, for approximate counting, Moitra's algorithm remains the only efficient \emph{deterministic} algorithm for $\Sat$ under conditions of this type.

  \Cref{theorem-sampling-main} and \Cref{theorem-counting-main} are in fact slightly stronger than \Cref{theorem-sample-simplified},
  because in \Cref{theorem-sample-simplified} we have simplified the condition between the exponent $\zeta$ and $(k,d)$.
  For example, 
  for $\epsilon = 1/{\mathrm{poly}(n)}$, and for $\omega(1)<kd<n^{o(1)}$ in the regime above,
  our algorithms run in $n^{1+o(1)}$ time for sampling, and $(n/\eps)^{2+o(1)}$ time for $\mathrm{e}^\epsilon$-approximate counting.
\end{remark*}

\subsection{Algorithm overview}
\label{sec:tech}

The first step of our algorithm is to mark variables.
We ensure that every clause has a certain amount of marked and unmarked variables.
Because every clause has sufficiently many unmarked variables,
using the local lemma,
we show that each individual marked variable is close to the uniform distribution.
We call this local uniformity.
This step so far is very similar to~\cite{Moi19}.

Our goal is to sample from the marginal distribution on the marked variables.
To do this, we simulate an idealized Glauber dynamics $\Glauber$ which converges to this distribution.
However, this distribution is not a Gibbs distribution,
and to calculate the transition probabilities becomes $\numP$-hard.
Our main effort is to show the following two things:
\begin{enumerate}
  \item\label{item:rapid} $\Glauber$ mixes in $O(n\log n)$ time (\Cref{section-proof-mixing});
  \item\label{item:implement} $\Glauber$ can be approximately efficiently implemented (\Cref{section-proof-main}).
\end{enumerate}

To show \Cref{item:rapid},
we use the path coupling technique by Bubley and Dyer \cite{bubley1997path},
which requires that for two assignments $X_t$ and $Y_t$ that differ on only one variable $v_0$,
the expected difference of $X_{t+1}$ and $Y_{t+1}$ after one step of $\Glauber$ is less than $1$.
For a marked variable $v\neq v_0$, let $\mu^X_v$ be the Gibbs distribution conditioned on $X_{t}$ minus the assignment of $v$.
In other words, $\mu^X_v$ is defined over assignments to all unmarked variables and $v$.
Define $\mu^Y_v$ similarly.
Consider a disagreement coupling $\+C_v$ between $\mu^X_v$ and $\mu^Y_v$,
constructed greedily starting from $v_0$.
The crucial observation is that, 
the probability that $v$ cannot be coupled is upper bounded by the probability that $v$ is in the discrepancy set of $\+C_v$.
Similar couplings have been defined by Moitra \cite{Moi19}.
(To get a better condition on our parameters, we actually follow the adaptive version in \cite{guo2019counting}.)
We then define a different disagreement coupling $\+C$ over all variables other than $v_0$, marked and unmarked alike, so that the expected difference of $X_{t+1}$ and $Y_{t+1}$ is upper bounded by the expected size of the discrepancy set of $\+C$.
This upper bound is shown by yet another coupling between the two couplings $\+C_v$ and $\+C$.

Finally, we show that the expected size of the discrepancy set of $\+C$ (not including $v_0$) is less than $1$.
Here we need a new argument based on counting induced paths to analyze these greedy disagreement couplings.
This is because the old analysis based on the so-called $\{2,3\}$-trees \cite{Moi19,guo2019counting},
which was used to show these couplings terminate in $O(\log n)$ steps
with high probability, can only get a constant bound in the form of $O(dk)$ on this expectation,
and thus is no longer strong enough.

To show \Cref{item:implement},
we first observe that due to local uniformity,
at any step of $\Glauber$,
unmarked variables are scattered into small connected components.
This has been observed before \cite{Moi19,guo2019counting}.
However, these components can have size as large as $\Omega(dk\log n)$.
Thus, a brute force enumeration would take time $n^{\Omega(dk)}$, which is too slow to our need.
Instead, we employ the local lemma again to show that a random assignment on these components satisfy all relevant clauses with probability roughly $\Omega(n^{-\zeta})$.
Thus, a naive rejection sampling has expected running time $O(n^{\zeta})$,
which results in the small overhead in \Cref{theorem-sample-simplified}.
Moreover, at the end of the algorithm,
we need to sample all unmarked variables,
this is done by the same rejection sampling method.

So far we have explained our sampling algorithm.
For counting, we use the simulated annealing method \cite{bezakova2008accelerating,vstefankovivc2009adaptive,huber2015approximation,kolmogorov18faster}.
First we define a suitable Gibbs distribution, which can be viewed as a product distribution conditioned on a new formula $\Phi'$ being satisfied.
Then our sampling algorithm can be adapted with minimal changes.
With the Gibbs distribution and its sampling algorithm, adaptive annealing can be applied to yield fast algorithms already.
Instead, we show that a simpler non-adaptive annealing procedure provides similar time bounds. 
Note that in general non-adaptive annealing is provably slower than the adaptive version \cite{vstefankovivc2009adaptive}.
The local lemma once again plays an important role to obtain necessary properties for a fast non-adaptive annealing procedure.

In \cite{guo2019counting}, a notion called ``pre-Gibbs distribution'' was introduced.
Its samples are pairs $(S,\sigma_S)$ where $S$ is a random subset of variables and $\sigma_S$ is an assignment of $S$.
The main requirement is that if we sample from the Gibbs distribution conditioned on $\sigma_S$,
the resulting sample follows the desired Gibbs distribution.
Our algorithm here is a realization of sampling from the pre-Gibbs distribution, where $S$ is fixed a priori.
It remains interesting to explore this idea of ``pre-Gibbs sampling'',
where we should allow a dynamic $S$.
With a dynamic $S$, we may get rid of the marking process by incorporating the adaptive coupling idea of \cite{guo2019counting},
which can greatly improve our assumption in
\Cref{theorem-sample-simplified}. 


\subsection{Related work}

The most relevant work is the algorithm by Moitra \cite{Moi19},
which we have discussed and compared with in detail above.
Moitra's work is subsequently refined and adapted to hypergraph colorings \cite{guo2019counting},
but it still suffers from the same slow running time.
The partial rejection sampling (PRS) method \cite{GJL19} also works in the local lemma setting.
However, for CNF formulas, PRS requires more complicated structural conditions in addition to the ones relating $k$ and $d$.

Prior to our work, no Markov chain algorithm is known to work in the local lemma parameter regimes for $\Sat$,
mainly because of the connectivity barrier.
For monotone $k$-CNF formulas,
where connectivity is not an issue,
Hermon et al.\ \cite{HSZ19} showed that the (straightforward) Glauber dynamics mixes in $O(n\log n)$ time if $k\ge 2\log d+C$ for some constant $C$,
which is tight up to the constant $C$ due to complementing hardness results \cite{BGGGS19}.
For proper colorings over simple hypergraphs, 
Frieze and Anastos \cite{FA17} showed that a slight variant of the straightforward Glauber dynamics mixes rapidly under conditions almost match the local lemma.
However, their work also requires that the vertex degrees are at least $\Omega(\log n)$ to ensure that the giant connected component occupies a $1-n^{-c}$ fraction of the whole state space.
In comparison, although our algorithm is also based on Markov chains,
we completely bypassed the connectivity issue.

Deterministic approximate counting algorithms often run in time $n^{f(\Delta)}$ where $\Delta$ is some parameter, 
such as the maximum degree of vertices in a graph, and $f(\Delta)\rightarrow\infty$ as $\Delta\rightarrow\infty$.
This is not desirable and is not polynomial-time if $\Delta=\omega(1)$.
Recently, there has been some effort to bring down such running times (often using randomized techniques like Markov chains) to achieve polynomial running time with fixed exponents for all $\Delta$.
Examples include the work of Efthymiou et al.\ \cite{EHSVY19} for counting independent sets \cite{Wei06},
and the work of Chen et al.\ \cite{CGGPSV19} for the algorithmic Pirogov-Sinai theory \cite{JKP19,HPR19}.

\section{Preliminaries}	
\subsection{Notations}

Let $\Phi =(V, C)$ be a CNF formula, where $V$ is the set of Boolean variables and $C$ is the set of clauses.
For each clause $c \in C$, we use
\begin{align}
\vbl{c} \triangleq \{y \in V \mid y \text{ or } \neg y \text{ appears in } c\}
\end{align}
to denote the set of variables that appear in $c$.
We say a CNF formula $\Phi$ is $k$-uniform if each clause contains exactly $k$ literals on distinct variables, i.e.\ $\abs{\vbl{c}} = k$ for all $c \in C$.
For any $c \in C$ and $x \in \vbl{c}$, we assume only one of the literal in $\{x,\neg x\}$ appears in $c$. Otherwise, the clause $c$ can always be satisfied.
We also assume that each variable belongs to at most $d$ distinct clauses.
Let $\mu$ denote the uniform distribution over all satisfying assignments for $\Phi$.
Our goal is to draw from a distribution close enough to $\mu$.
%

We often model the CNF formula $\Phi = (V, C)$ as a hypergraph 
\begin{align}
\label{eq-def-Hphi}
H_{\Phi}\triangleq(V, \edge),
\end{align}
where the vertices in $H_{\Phi}$ are variables in $\Phi$ and the hyperedges are defined as
$\edge \triangleq \{\vbl{c} \mid c \in C \}.$

We write $\log$ to denote $\log_2$ and $\ln$ to denote $\log_\mathrm{e}$.
We also write $\exp(s)$ to denote $\mathrm{e}^s$, especially when $s$ is a complicated expression.
We use $\mathbf{Pr}$ without subscript to denote the probability space generated by the algorithm in the context,
and use subscript to clarify other probability spaces.

\subsection{Lov\'asz local lemma}
Let $\mathcal{R}=\{R_1,R_2,\ldots,R_n\}$ be a collection of mutually independent random variables.
For any event $E$, denote by $\vbl{E}\subseteq \+R$ the set of variables determining $E$.
In other words, changing the values of variables outside of $\vbl{E}$ does not change the truth value of $E$.
Let $\mathcal{B} = \{B_1,B_2,\ldots, B_n\}$ be a collection of ``bad'' events.
%
%
For each event $B \in \mathcal{B}$, we define $\Gamma(B) \defeq \set{B'\in \mathcal{B} \mid B' \neq B \text{ and } \vbl{B'} \cap \vbl{B} \ne \emptyset }$.
For any event $A\notin \mathcal{B}$ and its determining variables $\vbl{A} \subseteq \mathcal{R}$,
we define $\Gamma(A) \triangleq \{B \in \mathcal{B} \mid \vbl{A} \cap \vbl{B} \ne \emptyset \}$.
%
%
Let $\Pr[\+P]{\cdot}$ denote the product distribution of variables in
$\mathcal{R}$. The following version of the Lov\'asz local lemma is from ~\cite{haeupler2011new}.
\begin{theorem}
\label{theorem-LLL}
If there is a function $x: \mathcal{B} \rightarrow (0,1)$ such that for any $B \in \mathcal{B}$,
\begin{align}\label{eqn:LLL}
  \Pr[\+P]{B} \leq x(B) \prod_{B' \in \Gamma(B)}(1-x(B')),	
\end{align}
then it holds that
\begin{align*}
  \Pr[\+P]{\bigwedge_{B \in \mathcal{B}} \overline{B} } \geq \prod_{B \in \mathcal{B}}(1 - x(B)) > 0. 	
\end{align*}
Thus, there exists an assignment of all variables that avoids all the bad events.

Moreover, for any event $A$, it holds that
\begin{align*}
\Pr[\+P]{A \,\big|\, \bigwedge_{B \in \mathcal{B}} \overline{B} } \leq \Pr[\+P]{A} \prod_{B \in \Gamma(A)}(1- x(B))^{-1}.	
\end{align*}
\end{theorem}

The next corollary follows from \Cref{theorem-LLL}. 
\begin{corollary}
\label{corollary-local-uniform}
Let $\Phi = (V, C)$ be a CNF formula.
Assume each clause contains at least $k_1$ variables and at most $k_2$ variables, and each variable belongs to at most $d$ clauses. 
For any $s \geq k_2$, if $2^{k_1} \geq 2\mathrm{e}ds$,  then there exists a satisfying assignment for $\Phi$ and for any $v \in V$,
\begin{align*}
 \max\left\{\Pr[X \sim \mu]{X(v) = 0}, \Pr[X \sim \mu]{X(v) = 1} \right\}\leq  \frac{1}{2}\exp\left(\frac{1}{s}\right),
\end{align*}
where $\mu$ is the uniform distribution of all satisfying assignments for $\Phi$.
\end{corollary}
\begin{proof}
Let $\Pr[\+P]{\cdot}$ denote the product distribution that every variable in $V$ takes a value from $\{0,1\}$ uniformly and independently. 
We define a collection of bad events $B_c$ for each $c \in C$, where $B_c$ represents the clause $c$ is not satisfied.
For each $c \in C$, we take $x(B_c) = \frac{1}{2ds}$. Thus, for any clause $c \in C$, we have
\begin{align*}
\Pr[\+P]{B_c} \leq \left(\frac{1}{2}\right)^{k_1}\leq \frac{1}{2\mathrm{e}ds}.
\end{align*}
To verify \eqref{eqn:LLL},
note that for any $y > 1$, it holds that $\left( 1 - \frac{1}{y} \right)^{y-1} \geq \frac{1}{\mathrm{e}}$. 
Since $s \geq k_2$ and $|\Gamma(B_c)| \leq (d-1)k_2 \leq 2ds-1$ for all $c \in C$, We have
\begin{align*}
  \Pr[\+P]{B_c} \leq \frac{1}{2ds}\left( 1 - \frac{1}{2ds} \right)^{2ds-1}
  \leq \frac{1}{2ds}\left( 1 - \frac{1}{2ds} \right)^{|\Gamma(B_c)|} = x(B_c) \prod_{b \in \Gamma(B_c)}(1 - x(B_b)).
\end{align*}
Hence, there exists a satisfying assignment for CNF formula $\Phi$.
For any variable $v \in V$, let $B_v$ denote the event that $v$ takes the value $0$. Note that $|\Gamma(B_v)| = d$.
By Theorem~\ref{theorem-LLL}, we have
\begin{align*}
\Pr[X \sim \mu]{X(v) = 1} \leq \frac{1}{2}\left( 1 - \frac{1}{2ds} \right)^{-d} \leq \frac{1}{2}\exp\left(\frac{1}{s} \right).
\end{align*}
Similarly, we have $\Pr[X \sim \mu]{X(v) = 0} \leq \frac{1}{2}\exp\left(\frac{1}{s} \right)$.
\end{proof}

The  Moser-Tardos algorithm \cite{moser2010constructive} constructs an assignment of all random variables in $\mathcal{P}$ that avoids all the bad events in $\mathcal{B}$. The Moser-Tardos algorithm is given in Algorithm~\ref{alg-MT}.
\begin{algorithm}[ht]
\SetKwInOut{Input}{Input}
\SetKwInOut{Output}{Output}
for each $R \in \mathcal{R}$, sample $v_R$ independently according to the distribution of $R$\;
\While{there exists a bad event $B \in \mathcal{B}$ s.t. $B$ occurs}{
pick an arbitrary $B \in \mathcal{B}$ s.t. $B$ occurs\;
resample the value of $v_R$ for all variables $R \in \vbl{B}$\;
 }
\Return{$(v_R)_{R \in \mathcal{R}}$}
\caption{The Moser-Tardos algorithm}\label{alg-MT}
\end{algorithm}
\begin{proposition}[Moser and Tardos~\cite{moser2010constructive}]
\label{proposition-MT}
Suppose the asymmetric local lemma condition \eqref{eqn:LLL} in Theorem~\ref{theorem-LLL} holds with the function $x: \mathcal{B}\rightarrow (0,1)$.
Upon termination, the Moser-Tardos algorithm returns an assignment that avoids all the bad events.
The expected total resampling steps for Moser-Tardos algorithm is at most $\sum_{B \in \mathcal{B}}\frac{x(B)}{1-x(B)}$.
\end{proposition}

\subsection{Coupling and mixing times for Markov chains}

Let $\mu$ and $\nu$ be two probability distributions over the same space $\Omega$.
The total variation distance is defined by 
\begin{align*}
  \DTV{\mu}{\nu}\defeq\frac{1}{2}\sum_{x\in\Omega}\abs{\mu(x)-\nu(x)}.
\end{align*}
If we have a random variable $X$ whose law is $\nu$,
we may write $\DTV{\mu}{X}$ instead of $\DTV{\mu}{\nu}$ to simplify the notation.

A coupling $\+C$ of $\mu$ and $\nu$ is a joint distribution over $\Omega\times\Omega$ such that projecting on the first (or second) coordinate is $\mu$ (or $\nu$).
A well known inequality regarding coupling is the following.
\begin{proposition}  \label{prop:coupling}
  Let $\+C$ be an arbitrary coupling of $\mu$ and $\nu$.
  Then
  \begin{align*}
    \DTV{\mu}{\nu}\le\Pr[(x,y)\sim\+C]{x\neq y}.
  \end{align*}
  Moreover, there exists an optimal coupling that achieves equality.
\end{proposition}

A Markov chain $(X_t)_{t \geq 0}$ over a state space $\Omega$ is given by its transition matrix $P:\Omega\times\Omega\rightarrow \mathbb{R}_{\ge 0}$.
A Markov chain $P$ is called \emph{irreducible} if for any $X,Y \in \Omega$, there exists an integer $t$ such that $P^t(X,Y) > 0$. 
A Markov chain $P$ is called \emph{aperiodic} if for any $X \in \Omega$, it holds that $\gcd\{t\mid P^t(X,X)>0 \}=1$.
We say the distribution $\mu$ over $\Omega$ is the \emph{stationary distribution} of a Markov chain $P$  if $\mu = \mu P$.
A Markov chain $P$ is \emph{reversible} with respect to $\mu$ if it satisfies the detailed balance condition
\begin{align*}
  \mu(X)P(X,Y)=\mu(Y)P(Y,X),
\end{align*}
which implies that $\mu$ is a stationary distribution of $P$.
If a Markov chain $P$ is irreducible and aperiodic, then it converges to the unique stationary distribution $\mu$.
The \emph{mixing time} of a Markov chain $P$ with stationary distribution $\mu$ is defined by 
\begin{align}
  \tmix(P,\delta) \triangleq \max_{X_0\in\Omega} \min \left\{t: \DTV{P^t(X_0,\cdot)}{\mu}\leq \delta\right\}.
  \label{eqn:mixing-time}
\end{align}
See the textbook~\cite{levin2017markov} for more details and backgrounds on Markov chains and mixing times.

Consider an irreducible and aperiodic Markov chain specified by the transition matrix $P$.
A coupling of the Markov chain is a joint process $(X_t,Y_t)_{t \geq 0}$ such that both $(X_t)_{t\geq 0}$ and $(Y_t)_{t \geq 0}$ individually follow the transition rule of $P$, and  if $X_t = Y_t$ then $X_s = Y_s$ for all $s \geq t$. The total variation distance between $P^t(X_0,\cdot)$ and $\mu$ can be bounded by
$\max_{X_0 \in \Omega}\DTV{P^t(X_0,\cdot)}{\mu} \leq \max_{X_0,Y_0 \in \Omega^2}\Pr{X_t \neq Y_t}$.

\emph{Path coupling}~\cite{bubley1997path} is a powerful technique to construct couplings of Markov chains. 
In this paper, we use the following path coupling lemma,
which is simplified for the Boolean hypercube. 
Let the state space $\Omega = \{0,1\}^N$ for some integer $N \geq 1$. For any $X, Y \in \Omega$, define the \emph{Hamming distance} between $X, Y$ as
\begin{align*}
  d_{\textrm{Ham}}(X, Y) \triangleq \abs{\set{1 \leq i \leq N \mid X(i) \neq Y(i) }}.	
\end{align*}
\begin{proposition}[\cite{bubley1997path}]
\label{proposition-path-coupling}
Let $\Omega = \{0,1\}^N$ some integer $N \geq 1$. Let $P:\Omega\times\Omega\rightarrow \mathbb{R}_{\ge 0}$ be the transition matrix of an irreducible and aperiodic Markov chain.
Suppose there is a coupling  $(X,Y) \rightarrow (X',Y')$	 of the Markov chain defined for all $X, Y \in \Omega$ with $ d_{\textrm{Ham}}(X, Y) = 1$, which satisfies 
\begin{align*}
\E{ d_{\textrm{Ham}}(X', Y') \mid X, Y } \leq 1 - \lambda,	
\end{align*}
for some $0 < \lambda < 1$. Then the mixing time of the Markov chain is bounded by
\begin{align*}
\tmix(P,\delta) \leq \frac{1}{\lambda}\log \left( \frac{N}{\delta} \right).	
\end{align*}
\end{proposition}

%

\section{The sampling algorithm}
\label{section-alg}
Let $\Phi = (V, C)$ be a $k$-uniform CNF formula, in which each variable belongs to at most $d$ clauses. 
In this section we give our Markov chain based algorithm to sample satisfying assignments almost uniformly at random.
%
%

\subsection{Marking variables}
\label{section-mark}
Our algorithm first marks a set of marked variables ${\+M} \subseteq V$. 
We say a variable $v \in V$ is \emph{marked} if $v \in {\+M}$, 
or $v$ is \emph{unmarked} if $v \not\in {\+M}$. 
We will ensure the following condition for the set of marked variables ${\+M}$,
where $k_\alpha \geq 1$ and  $k_\beta \geq 1$ are two integer parameters to be specified later satisfying $k_\alpha + k_\beta \leq k$.
\begin{condition}
\label{condition-marked-variables}
Each clause has at least $k_\alpha$ marked variables and at least $k_\beta$ unmarked variables.
\end{condition}

We use the Moser-Tardos algorithm, \Cref{alg-MT}, to find $\+M$. Define $0\leq \alpha,\beta\leq 1$ as
\begin{align*}
\alpha \triangleq \frac{k_{\alpha}}{k},\qquad \beta \triangleq \frac{k_{\beta}}{k}. 
\end{align*}
Suppose we mark each variable independently with probability $\frac{1+\alpha -\beta}{2}$. 
For each clause $c \in C$, 
let $M_c$ be the bad event that ``$c$ has less than $k_\alpha$ marked variables or less than $k_\beta$ unmarked variables''. 
The lemma below follows from \Cref{proposition-MT} and verifying \eqref{eqn:LLL}.

\begin{lemma}
  \label{lemma-MT}
  Assume $2^k \geq (2\mathrm{e}dk)^{\frac{6 \ln 2 \cdot (1+\alpha-\beta)}{(1-\alpha-\beta)^2}}$.
  There is an algorithm such that for any $\delta > 0$, with probability at least $1 - \delta$, 
  it returns a set of marked variables satisfying Condition~\ref{condition-marked-variables} 
  with time complexity $O\left(dkn \log \frac{1}{\delta}\right)$, where $n =|V|$ is the number of variables.
\end{lemma}
\begin{proof}
  To apply \Cref{alg-MT},
  Let $\Pr[\+P]{\cdot}$ denote the product distribution that every variable is marked independently with probability $\frac{1+\alpha-\beta}{2}$. 
  By Chernoff bound~\cite[Corollary~4.6]{mitzenmacher2017probability}, we have
  \begin{align*}
    \forall c \in C:\quad \Pr[\+P]{M_c} \leq 	 
    2 \exp \left( -\frac{(1-\alpha-\beta)^2}{6(1+\alpha-\beta)}\cdot k \right) = 2 \left( \frac{1}{2} \right)^ {\frac{(1-\alpha-\beta)^2}{6 \ln 2 \cdot (1+\alpha-\beta)}\cdot k }.
  \end{align*}
  We define a function $x$ as $x(M_c) \triangleq \frac{1}{dk}$ for all $c \in C$. We have for all $c \in C$,
  \begin{align*}
    \Pr[\+P]{M_c} \leq \frac{1}{\mathrm{e}dk} \leq  x(M_c) \prod_{M_b \in \Gamma(M_c)} (1 - x(M_b)).
  \end{align*}
  Since the total number of clauses is at most $dn$, by Proposition~\ref{proposition-MT},
  the expected number of resampling steps is at most 
  \begin{align}
    \sum_{c \in C} \frac{x(M_c)}{1-x(M_c)} \leq \frac{2n}{k}.
  \end{align}
  By Markov inequality, if we run \Cref{alg-MT} for at most $\frac{4n}{k}$ resampling steps, 
  the algorithm returns the set $\+M$ with probability at least $\frac{1}{2}$. 
  If we run $\left\lceil\log\frac{1}{\delta}\right\rceil$ Moser-Tardos algorithms independently, 
  then with probability at least $1 - \delta$, one of them finds the set $\+M$ within  $\frac{4n}{k}$ resampling steps.

  Note that in each resampling step, we need to resample $k$ variables and check whether $dk$ bad event occurs, 
  and the cost of checking one event is at most $k$. Hence, the total time complexity is $O\left(ndk\log\frac{1}{\delta}\right)$.
\end{proof}

We note that much better concentration bound is known to the Moser-Tardos algorithm \cite{HH17}.
However, \Cref{lemma-MT} is sufficient to our need.

We use the algorithm in Lemma~\ref{lemma-MT} with $\delta = \frac{\epsilon}{4}$ to construct the set of marked variables $\+M \subseteq V$.
If the algorithm fails to construct ${\+M}$, then our algorithm terminates immediately and outputs an arbitrary assignment $X \in \{0,1\}^V$.
This bad event occurs with probability at most $\frac{\epsilon}{4}$.
In the rest of this section, we assume that the set of marked variables ${\+M} \subseteq V$ is already found.

\subsection{The main algorithm}

In this section we present our algorithm for sampling satisfying assignments of CNFs. 
We will need some notations first.
For an arbitrary set of variables $S\subseteq V$, let $\mu_S$ be the marginal distribution on $S$ induced from $\mu$.
Formally,
\begin{align*}
  \forall \sigma \in \{0,1\}^{S}:\quad \mu_{S}(\sigma) = \sum_{\substack {\tau \in \{0,1\}^V,\ \tau(S) = \sigma}} \mu(\tau).
\end{align*}
When $S=\{v\}$ for some $v\in V$,
we also write $\mu_v$ instead of $\mu_{\{v\}}$.
Moreover, for a partial assignment $X\in\{0,1\}^{\Lambda}$ where $\Lambda\subset V$ and $S\cap\Lambda=\emptyset$,
let $\mu_S^X(\cdot):=\mu_S(\cdot\mid X)$ be the marginal distribution on $S$ conditioned on the partial assignment on $\Lambda$ is $X$.

The main idea of our sampling algorithm is to simulate a Markov chain whose stationary distribution is the marginal distribution $\mu_{\+M}$ on $\+M$.
Let $\Glauber$ be the idealized Glauber dynamics for the marked variables.
Namely, we start with an initial assignment $X_0\in\{0,1\}^{\+M}$ where $X_0(v)$ is uniformly at random for all $v\in\+M$.
In the $t$-th step, the chain evolves as follows:
\begin{itemize}
\item pick $v \in \+M$ uniformly at random and set $X_t(u) \gets X_{t-1}(u)$ for all
  $u \in \+M \setminus \{v\}$;
\item sample $X_t(v) \in \{0, 1\}$ from the distribution
  $\mu_{v}(\cdot \mid X_{t-1}(\+M \setminus \{v\} ))$.
\end{itemize}
This chain is reversible with respect to $\mu_{\+M}$, as for any $X,Y\in\{0,1\}^{\+M}$ that differ on only $v$,
\begin{align}\label{eqn:db}
  \begin{split}
    \mu_\+M(X)\Glauber(X,Y)&=\frac{1}{\abs{M}}\cdot\mu_\+M(X)\mu_{v}(Y(v)\mid X(\+M \setminus \{v\} )) = \frac{1}{\abs{M}}\cdot\frac{\mu_\+M(X)\mu_\+M(Y)}{\mu_{\+M\setminus\{v\}}(X(\+M \setminus \{v\}))}\\
    & = \frac{1}{\abs{M}}\cdot \mu_\+M(Y)\mu_{v}(X(v)\mid Y(\+M \setminus \{v\} )) = \mu_\+M(Y)\Glauber(Y,X).
  \end{split}
\end{align}
We will show that $\Glauber$ is both irreducible and aperiodic in our parameter regimes.
We simulate this chain to obtain a random assignment $X_{\+M}\in\set{0,1}^{\+M}$ whose distribution is close enough to $\mu_{\+M}$.
Then the algorithm samples a random assignment $X_{V\setminus \+M}\in \{0,1\}^{V \setminus \+M}$ for
unmarked variables from the distribution $\mu_{V \setminus \+M }(\cdot \mid X_{\+M})$.
The final sample is $\Xalg\defeq X_{\+M}\cup X_{V\setminus\+M}$.


This chain $\Glauber$ is an idealized process because the transitions of the chain rely on evaluating some nontrivial marginal probabilities, which in general can be as hard as the problem of counting the number of satisfying assignments itself.
To efficiently simulate one step of the Markov chain and to complete the random assignments for unmarked variables, we need to sample from the marginal distributions $\mu_{v}(\cdot \mid X_{t-1}(\+M \setminus \{v\} ))$ and $\mu_{V \setminus \+M }(\cdot \mid X_{T})$,
where $t\le T$ and $T$ is an upper bound of the mixing time of $\Glauber$.
We will use a subroutine $\sample(\cdot)$ for this. 
Given an assignment $X \in \{0,1\}^\Lambda$ on the subset $\Lambda \subseteq {\+M}$ and a subset $S \subseteq V \setminus \Lambda$ of variables, the subroutine $\sample(\Phi, \delta, X, S)$ returns a random assignment $Y\in\{0,1\}^S$ from the distribution $\mu_{S}(\cdot\mid X)$ upon success.
We will ensure that $\sample(\Phi, \delta, X, S)$ is efficient and
when we call it in \Cref{alg-mcmc}, it returns a sample within total variation distance $\delta$ to the desired distribution with probability at least $1-\delta$ for a small $\delta>0$.
This is because due to \Cref{corollary-local-uniform} and its variants, the marked variables are almost uniform, and conditioned on any almost uniform assignment of (almost all) marked variables, the remaining formula splits into many disjoint small connected components.

The whole sampling algorithm is formally described in \Cref{alg-mcmc}.
\begin{algorithm}[ht]
  \SetKwInOut{Input}{input} \SetKwInOut{Output}{output} 
  \Input{a CNF formula $\Phi=(V,C)$, a parameter $\epsilon > 0$, and a set of marked variables ${\+M}$.}  
  \Output{a random assignment $\Xalg \in \{0,1\}^{V}$.}  
  sample $X_0(v) \in \{0,1\}$ uniformly and independently for each $v \in {\+M}$\; 
  \For{each $t$ from $1$ to $T := \Tmix$} {
    choose variable $v \in {\+M}$ uniformly at random\; 
    \tcc{resample $X_{\+M}(v)$ from the distribution $\mu_v(\cdot \mid X_{t-1}({\+M} \setminus \{v\} ) )$}
    $X_{t}(v) \gets \sample(\Phi, \frac{\epsilon}{4(T+1)}, X_{t-1}({\+M} \setminus \{v\}), \{v\})$\label{line-sample-1}\; 
    $\forall u\in\+M$ and $u\neq v$, $X_{t}(u)\gets X_{t-1}(u)$\;}
  \tcc{sample $X_{V\setminus\+M}$ from the distribution $\mu_{V \setminus {\+M}}(\cdot\mid X_{T})$} 
  $X_{V\setminus\+M} \gets \sample(\Phi, \frac{\epsilon}{4(T+1)}, X_{T}, V \setminus {\+M})$\label{line-sample-2}\;
  \Return{$\Xalg = X_{T}\cup X_{V\setminus\+M}$\;}
  \caption{The sampling algorithm}\label{alg-mcmc}
\end{algorithm}

In \Cref{alg-mcmc}, $\sample(\cdot)$ appears in Line~\ref{line-sample-1} and Line~\ref{line-sample-2} and
returns random assignments on $\set{v}$ and $V\setminus\+M$ respectively. 
In our implementation, we allow their distributions to be slightly biased (controlled by the parameter $\delta=\frac{\eps}{4(T+1)}$). 

The correctness and the efficiency of \Cref{alg-mcmc} rely on three facts:
\begin{enumerate}
  \item the Glauber dynamics for marked vertices is rapidly mixing;
  \item the $\sample(\cdot)$ subroutine for unmarked vertices is efficient;
  \item the small bias in the distribution caused by $\sample(\cdot)$ does not affect the final distribution much.
\end{enumerate}

The rapid mixing property of the Glauber dynamics is analyzed in Section~\ref{section-proof-mixing}.
Details of $\sample(\cdot)$ will be given in \Cref{section-sampleunmark} and its analysis in \Cref{section-proof-main}. 

We will ensure that, with high probability,
$\sample(\cdot)$ returns samples whose distributions are close to the desired ones in both Line~\ref{line-sample-1} and Line~\ref{line-sample-2}.
Using this, we will show that \Cref{alg-mcmc} couples with high probability with the idealized chain $\Glauber$.
As a result, the distribution of the random assignment $\Xalg$ returned by \Cref{alg-mcmc} is close to~$\mu(\cdot)$.

\begin{lemma}\label{lemma-sample}
  Suppose $2^{k_\alpha} \geq 4\mathrm{e}^2d^2k^2$, $2^{k_{\beta}} \geq 2^{16}d^9k^9$,
  and $\+M$ satisfying \Cref{condition-marked-variables} has been found.
  The random assignment $\Xalg \in \{0,1\}^V$ returned by \Cref{alg-mcmc} satisfies
  \begin{align}
    \label{eq-DTV-Xid-mu}
    \DTV{\Xalg}{\mu} \leq \frac{3\epsilon}{4}.	
  \end{align}
\end{lemma}

\Cref{lemma-sample} is proved in \Cref{section-sampling-analyze}.

\subsection{The \texorpdfstring{$\sample$}{Sample} subroutine}
\label{section-sampleunmark}

Here we give the subroutine $\sample(\Phi, \delta, X, S)$, where $X \in \{0,1\}^\Lambda$ is an assignment on subset $\Lambda \subseteq {\+M}$ and $S \subseteq V \setminus \Lambda$ is a subset of variables. 
The output of the subroutine is a random assignment $Y \in \{0,1\}^{S}$, which ideally should follow the conditional marginal distribution $\mu_{S}(\cdot\mid X)$. 
However, in order for the efficiency of the subroutine, some small error is tolerated.

Our basic idea is to find all connected components of a new formula $\Phi^X$.
We will show that in the execution of \Cref{alg-mcmc}, these components are sufficiently small.
Then we will use rejection sampling on them independently for each component.

Let us first define $\Phi^X$ and its connected components.
Given a CNF formula $\Phi=(V,C)$ and a partial assignment $X\in\set{0,1}^\Lambda$ for some $\Lambda\subseteq V$, we simplify $\Phi$ under $X$ to obtain $\Phi^X=(V^X,C^X)$. 
Formally, we have
\begin{itemize}
\item $V^X=V\setminus \Lambda$, and
\item $C^X$ is obtained from $C$ by removing all clauses that has been satisfied under $X$\footnote{Let $c \in C$ be a clause in $\Phi$. We say $c$ is satisfied under the (partial) assignment $X$ if any literal of $c$ is already assigned true.} 
  and removing the appearance of $x$ or $\neg x$ from the remaining unsatisfied clauses
  for every $x\in \Lambda$.
\end{itemize}

Recall that
\begin{align}
  \label{eq-same-distribution}
  \forall \sigma \in \{0,1\}^{V^X}=\{0,1\}^{V\setminus\Lambda}:\quad \mu^X_{V \setminus \Lambda }(\sigma) = \mu_{V \setminus \Lambda }(\sigma \mid X).
\end{align}
It is straightforward to check that $\mu^X_{V \setminus \Lambda }$ is the uniform distribution over all satisfying assignments of $\Phi^X$. 
Let $H_{\Phi^X} = (V^X, \+E^X)$ be the hypergraph defined in~\eqref{eq-def-Hphi} for the CNF formula
$\Phi^X$.
Let $H^X_i = (V^X_i,\+E^X_i)$ for $1\leq i \leq \ell$ denote all the connected components in the hypergraph $H_{\Phi^X}$, where $\ell$ is the number of connected components. 
Each $H^X_i = (V^X_i,\+E^X_i)$ represents a CNF formula $\Phi^X_i = (V^X_i, C^X_i)$, where
\begin{align*}
  C^X_i \defeq \set{c \in C^X \mid  \text{clause $c$ is represented by a hyperedge in $\+E^X_i$ }  }.	
\end{align*}
We have $ \Phi^X = \Phi^X_1 \land \Phi^X_2 \land \dots \land \Phi^X_{\ell}$, and all the $V^X_i$ are
disjoint. Let $\mu^X_i$ be the uniform distribution on all satisfying assignments of $\Phi^X_i$ for
every $i=1,\dots,\ell$, then $\mu^X_{V \setminus \Lambda }(\cdot)$ is the product distribution of all~$\mu^X_i$.

%

Obviously, the distribution $\mu_{S}(\cdot \mid X)$ is determined by only those connected components intersecting $S$.
Without loss of generality, we assume that $S \cap V^X_i \neq \emptyset$ for $1\leq i \leq m$ and $S \cap V^X_i = \emptyset$ for $ m< i \leq \ell$.
To draw a random assignment $Y \in \{0,1\}^{S}$ from the distribution $\mu_{S}(\cdot \mid X)$, 
we independently draw a random assignment $Y_i$ from $\mu^X_{i}(\cdot)$ for each $1\leq i \leq m$. 
Let
\begin{align*}
  Y' \defeq \bigcup_{i = 1}^m Y_i.	
\end{align*}
Note that $S\subseteq \bigcup_{i=1}^m V_i$. 
Our sample $Y$ is the projection of $Y'$ on $S$, namely $Y=Y'(S)$.
It is easy to verify that $Y$ follows the marginal distribution on $S$ induced by $\mu^X_{V\setminus\Lambda}$. 
By~\eqref{eq-same-distribution}, the random assignment $Y$ follows the distribution $\mu_{S}(\cdot \mid X)$.

To draw from individual $\mu^X_{i}(\cdot)$ for each $1\le i\le m$,
we can simply use the naive rejection sampling: repeatedly draw uniform assignments on $\set{0,1}^{V_i^X}$ and return the first one that satisfies $\Phi^X_i$.
This should terminate fast if the connected component $\mathcal{E}_i^X$ is small.

Our implementation of $\sample(\Phi,\delta, X, S)$ is then clear: it tries for each $\Phi^X_i$, $1\le i\le m$, to repeatedly draw uniform assignments for at most $R$ (to be suitably fixed) times and return the first satisfying one.
Bad events happen if for one of the components, say $\Phi^X_i$, the size of $\Phi^X_i$ is too large or all $R$ trials fail to satisfy $\Phi^X_i$, in which case an arbitrary assignment on $S$ is returned.

%
%


Formally,  
let $0<\eta <1$ satisfy
\begin{align}
\label{eq-def-eta}
2^{k_{\beta}} \geq \frac{20}{\eta}\mathrm{e}dk,	
\end{align}
and define
\begin{align*}
R \triangleq \repeattime.
\end{align*}
In the subroutine $\sample(\Phi,\delta,X,S)$, we
\begin{itemize}
\item check the size $|\mathcal{E}_i^X|$ for all $1\leq i \leq m$ , if there exists
  $|\mathcal{E}_i^X| > \toolarge$, then the subroutine terminates and returns a
  $Y \in \{0,1\}^{S}$ uniformly at random;
\item for each $1\leq i \leq m$, use the naive rejection sampling for at most
  $R$ times to draw a random assignment $Y^X_i$ from the distribution
  $\mu^X_i$; if there exists $1\leq i \leq m$ such that the subroutine fails to draw a $Y^X_i$ from
  $\mu^X_i$ after $R$ rejection sampling trials, then the subroutine terminates and returns a
  $Y \in \{0,1\}^{S}$ uniformly at random.
\end{itemize}
The subroutine $\sample(\Phi, \delta, X, S)$ is described in Algorithm~\ref{alg-sample}.

\begin{algorithm}[ht]
  \SetKwInOut{Input}{Input}%
  \SetKwInOut{Output}{Output}%
  \Input{a CNF formula $\Phi=(V,C)$, a
    parameter $0<\delta,\eta<1$, an assignment $X \in \{0,1\}^{\Lambda}$ for some
    $\Lambda \subseteq V$, a set of variables $S \subseteq V \setminus \Lambda$, and
    $n = |V|$.}%
  \Output{a random assignment $Y \in \{0,1\}^{S}$.}%
  simplify $\Phi$ under $X$ and obtain a new formula $\Phi^X$\;%
  find all the connected components $\set{H^X_i = (V_i^X,\mathcal{E}_i^X)\mid 1\leq i \leq m}$ in
  $H_{\Phi^X}$ s.t. each $V^X_i \cap S \neq \emptyset$\;%
  \If{there exists $ 1 \leq i \leq m$
    s.t. $|\mathcal{E}^X_i| > \toolarge$}
  { \Return{an assignment $Y \in \{0,1\}^{S}$ uniformly at random\;\label{line-bad-return-1}} }
  \For{each $i$ from 1 to $m$ }{ let
    $\Phi^X_i =(V_i^X, C_i^X)$ be the CNF formula represented by $H^X_i=(V_i^X, \mathcal{E}_i^X)$\;%
    $Y^X_i \gets \Rejection\left(\Phi^X_i, R \right)$, where $R = \repeattime$\label{line-sample-rejection-sampling}\;
    \If{$Y^X_i = \perp$}{
      \Return{an assignment $Y \in \{0,1\}^{S}$ uniformly at random\;\label{line-bad-return-2}}
    } }
  \Return{$Y = Y'(S)$, where $Y' = \bigcup_{i=1}^m Y^X_i$ \;\label{line-good}}
  \caption{$\sample(\Phi, \delta, X, S)$}\label{alg-sample}
\end{algorithm}

\begin{algorithm}[ht]
  \SetKwInOut{Input}{Input} \SetKwInOut{Output}{Output} \Input{a CNF formula $\Phi=(V,C)$, a
    parameter $R > 0$.}  \Output{a random assignment $Y \in \{0,1\}^{V}$ or a special symbol
    $\perp$.}  \For{each $i$ from 1 to $R$}{ sample $Y \in \{0,1\}^V$ uniformly and independently\;
    \If{all the clauses in $C$ are satisfied by $Y$}{ \Return{Y\;} } } \Return{$\perp$\;}
  \caption{$\Rejection(\Phi,R)$}\label{alg-rejection}
\end{algorithm}

The following proposition is a basic property of rejection sampling.
\begin{proposition}
  \label{proposition-good-event}
  In the subroutine $\sample(\Phi, \delta, X, S)$, conditioned on that the random assignment
  $Y \in \{0,1\}^{S}$ is returned in Line~\ref{line-good}, $Y$ follows the law
  $\mu_{S}(\cdot\mid X)$.
\end{proposition}

%

With the CNF formula represented by a standard data structure, the running time of $\sample(\Phi, \delta, X, S)$ is easily bounded by $\tilde{O}(|S|\cdot R\cdot\mathrm{poly}(d,k))$.
This is rigorously analyzed in Lemma~\ref{lemma-sample-correctness} in~\Cref{section-proof-main}.
In the same lemma we also prove that 
conditioning on that every component is small (i.e.~Line~\ref{line-bad-return-1} in \Cref{alg-sample} is not executed),
the $\sample$ subroutine fails (i.e.~Line~\ref{line-bad-return-2} in \Cref{alg-sample} happens) with probability at most $\delta$.
Such failure is due to the randomness of the rejection sampling.
In another key lemma, Lemma~\ref{lemma-too-large} in~\Cref{section-proof-main}, we prove that for any call of $\sample$ in \Cref{alg-mcmc},  Line~\ref{line-bad-return-1} in \Cref{alg-sample} is indeed executed with probability at most $\delta$.
Such failure is due to the randomness of the input $X$ to $\sample$.
%
Overall, 
with probability at least $1 - \delta$,  
the distribution of the assignments returned by $\sample(\Phi,\delta,X,S)$ is within total variation distance at most $\delta$ from $\mu_S(\cdot\mid
X)$.
%

\section{Rapid mixing of the idealized dynamics}
\label{section-proof-mixing}
\noindent 
Let $\Phi=(V, C)$ be a CNF formula. Let ${\+M}\subseteq V$ be the set of marked variables satisfying Condition~\ref{condition-marked-variables} and $\Omega \triangleq \{0,1\}^{{\+M}}$.
Let $\Glauber$ be the Glauber dynamics for marked variables,
and use $(X_t)_{t\ge 0}$ to denote the state at time $t$ where $X_t\in\{0,1\}^{\+M}$.
In this section, we show that the idealized Glauber dynamics $\Glauber$ is rapidly mixing.
\begin{lemma}
  \label{lemma-mixing}
  Let $\Phi = (V, C)$ be a $k$-uniform CNF formula such that each variable belongs to at most $d$
  clauses.  Suppose ${\+M} \subseteq V$ satisfies Condition~\ref{condition-marked-variables} with
  parameters $k_\alpha$ and $k_\beta$.  Let $\Glauber$ be the Glauber dynamics for marked
  variables.  If $2^{k_{\beta}}\geq 2^{16}d^9k^9$, then for any $\delta > 0$, it holds that
  \begin{align*}
    \tmix(\Glauber,\delta) \leq  \left\lceil 2n \log \frac{n}{\delta} \right\rceil,
  \end{align*}
  where $n = |V|$ and the mixing time $\tmix$ is defined in \eqref{eqn:mixing-time}.
\end{lemma}

\subsection{The stationary distribution}
We first prove that the Glauber dynamics $\Glauber$ has the unique stationary distribution $\mu_{{\+M}}$. 
\begin{lemma}
  \label{lemma-stationary}
  If $2^{k_\beta} \geq 4\mathrm{e}dk$, then the support of $\mu_{\+M}$ is all of $\Omega=\{0,1\}^{\+M}$,
  and the Glauber dynamics $\Glauber$ for marked variables has the unique stationary distribution $\mu_{{\+M}}$.
\end{lemma}
\begin{proof}
  For any $v \in {\+M}$ and any assignment $X' \in \{0,1\}^{{\+M} \setminus \{v\} }$, we claim that 
  \begin{align}
  \label{eq-irreducible}
  \forall c \in \{0,1\}:\quad \mu_v(c \mid X') > 0.
  \end{align}
  This implies that for any $X, Y \in \Omega$ with Hamming distance $d_{\textrm{Ham}}(X, Y)  = |\{ v \in \mathcal{M} \mid X(v) \neq Y(v) \}|$, 
  it is possible to transform $X$ to $Y$ in $d_{\textrm{Ham}}(X, Y)$ steps. 
  Hence, $\Glauber$ is irreducible.
  It also implies that the support of $\mu_{\+M}$ is $\Omega$.
  Besides, for any $X \in \Omega$, we have $\Glauber(X, X) > 0$. Hence, this chain is aperiodic.
  
  We now prove~\eqref{eq-irreducible}. Let $\Phi^{X'}$ be the CNF formula obtained from $\Phi$ by deleting all the clauses that are satisfied by $X'$ and all the variables in ${\+M} \setminus \{v\}$. Let $\mu'$ denote the uniform distribution of all solutions of $\Phi'$.
  Then we have
  \begin{align*}
  \forall c \in \{0,1\}: \quad \mu_v(c \mid X') = \mu'_v(c).	
  \end{align*}
  In CNF formula $\Phi'$, each clause has at least $k_\beta$ variables and at most $k$ variables and each variable belongs to at most $d$ clauses. Since $2^{k_\beta} \geq 4\mathrm{e}dk$, by Corollary~\ref{corollary-local-uniform}, we have
  \begin{align*}
  \forall c \in \{0,1\}: \quad \mu_v(c \mid X') = \mu'_v(c) \leq \frac{1}{2}\exp\left(\frac{1}{2k} \right) \leq \frac{\sqrt{\mathrm{e}}}{2} < 1.	
  \end{align*}
  This implies $\mu_v(c \mid X') > 0$ for all $c \in \{0,1\}$.
  
  By the update rule of the Glauber dynamics chain, it is easy to verify the following detailed balance condition as in \eqref{eqn:db}:
  \begin{align*}
    \forall X,Y \in \Omega:\quad	\mu_{{\+M}}(X) \Glauber(X, Y) = \mu_{{\+M}}(Y)\Glauber(Y, X).
  \end{align*}
  Since the Markov chain is irreducible and aperiodic, this proves that the Markov chain $(X_t)_{t \geq 0}$ has the unique stationary distribution $\mu_{{\+M}}$.
\end{proof}
Hence, under the condition in \Cref{lemma-mixing}, $\Glauber$ has unique stationary distribution $\mu_{\+M}$.

\subsection{The mixing time}
We next prove that $\Glauber$ is rapidly mixing provided that $2^{k_{\beta}}\geq 2^{16}d^9k^9$.
The mixing time in Lemma~\ref{lemma-mixing} is proved by the path coupling argument~\cite{bubley1997path}.
For any $X, Y \in \Omega$, recall their Hamming distance as
\begin{align*}
  d_{\textrm{Ham}}(X, Y) \triangleq \abs{\set{v \in {\+M} \mid X(v) \neq Y(v) }}.	
\end{align*}
Let $X, Y \in \Omega$ be two assignments that disagree only on a single variable,
namely, $d_{\textrm{Ham}}(X, Y)=1$.
We construct a coupling of Markov chains $(X, Y) \rightarrow (X' ,Y')$ satisfying
\begin{align}
  \label{eq-path-coupling-target}
  \E{d_{\textrm{Ham}}(X',Y') \mid X, Y} \leq 1 - \frac{1}{2n}.	
\end{align}
Note that $d_{\textrm{Ham}}(X, Y) \leq n$ for all $X,Y \in \Omega$. Then Lemma~\ref{lemma-mixing}
is proved by the path coupling lemma~(Proposition~\ref{proposition-path-coupling}) together with Lemma~\ref{lemma-stationary}.

The coupling $(X, Y) \rightarrow (X' ,Y')$ is defined as follows.
\begin{definition}
  Let $X, Y \in \Omega$ be two assignments that disagree only on a single variable, say $X(v_0) = 0$ and $Y(v_0) = 1$ where $v_0 \in \+M$. 
  Let $\+M_v\defeq\+M\setminus\{v\}$ for any $v\in \+M$.
  The coupling $(X, Y) \rightarrow (X' ,Y')$ is defined as:
  \begin{itemize}
  \item pick the same variable $v \in {\+M}$ uniformly at random, and set $X'(u) = X(u)$ and
    $Y'(u) = Y(u)$ for all variables $u \in \+M_v$;
  \item sample $(X'(v),Y'(v))$ jointly from the optimal coupling of two conditional marginal
    distributions $\mu_v(\cdot \mid X(\+M_v))$ and
    $\mu_v(\cdot \mid Y(\+M_v ))$.
  \end{itemize}
\end{definition}
\noindent It is easy to verify that this is a valid coupling of two Markov
chains. Two transitions $X \rightarrow X'$ and $Y \rightarrow Y'$ are both faithful copies of the Glauber dynamics chain.
We remark that none of the couplings in this section is efficiently computable.
They only serve as tools for the analysis of the Markov chain.

For each marked variable $v \in {\+M}$, we define $D_v$ as
\begin{align}
  \label{eq-def-Dv}
  D_v \triangleq 	\DTV{\mu_v(\cdot \mid X({\+M}_v ))}{\mu_v(\cdot \mid Y({\+M}_v ))}.
\end{align}
which is the total variation distance between $\mu_v(\cdot \mid X({\+M}_v ))$ and $\mu_v(\cdot \mid Y({\+M}_v ))$.
Moreover,
since $X(\+M_{v_0})=Y(\+M_{v_0})$, by~\eqref{eq-def-Dv},
\begin{align*}
  D_{v_0} = 0.	
\end{align*}
By \Cref{prop:coupling}, under our coupling,
\begin{align*}
  \Pr{X'(v) \neq Y'(v) \mid v \in \+M \text{ is picked}} = D_v.	
\end{align*}
Hence, the expected Hamming distance between $X'$ and $Y'$ is at most
\begin{align}
  \E{d_{\textrm{Ham}}(X',Y') \mid X, Y} &= 1+ \frac{1}{|\+M|}\sum_{v \in \+M}	D_v - \frac{1}{|\+M|}\notag\\
  \label{eq-bound-E}                     &= 1 - \frac{1}{|\+M|}\left( 1 - \sum_{v \in \+M}D_v \right).
\end{align}
To prove the inequality in~\eqref{eq-path-coupling-target}, 
it is sufficient to prove the following lemma and notice that $\abs{\+M}\le n$.
\begin{lemma}
  \label{lemma-bound-dv-sum}
  Given two assignments $X, Y \in \Omega$ such that $X$ and $Y$ disagree only on a single variable
  $v_0 \in {\+M}$, if $2^{k_{\beta}}\geq 2^{16}d^9k^9$, it holds that
  \begin{align*}
    \sum_{v \in {\+M}}D_v \leq \frac{1}{2},
  \end{align*}
  where $D_v$ is the total variation distance defined in~\eqref{eq-def-Dv}.
\end{lemma}
\noindent
Combining inequality~\eqref{eq-bound-E} and \Cref{lemma-bound-dv-sum} proves
inequality~\eqref{eq-path-coupling-target}.  This proves \Cref{lemma-mixing}.
\Cref{lemma-bound-dv-sum} is shown in the next subsection.

\subsection{Analysis of the path coupling}

Let us first sketch the proof idea of Lemma~\ref{lemma-bound-dv-sum}. 
Recall that we have two assignments $X$ and $Y$ which differ on only $v_0$.
In order to bound $D_v$ for any $v\in\+M$ and $v\neq v_0$, we construct a coupling $\+C_v$ of two distributions
$\mu(\cdot \mid X(\+M_v))$ and $\mu(\cdot\mid Y(\+M_v))$, where $\+M_v=\+M\setminus\{v\}$. 
Since $\+C_v$ projected on $v$ is a coupling
between $\mu_v(\cdot\mid X(\+M_v))$ and $\mu_v(\cdot\mid Y(\+M_v))$,
by \Cref{prop:coupling}, we have
\[
  D_v\le \Pr[(\sigma_X,\sigma_Y)\sim \+C_v]{\sigma_X(v)\ne\sigma_Y(v)}.
\]
A high-level description of our construction of $\+C_v$ is as follows: we start from two partial
assignments $X$ and $Y$ such that initially only the value on $v_0$ is set, say $X(v_0)=0$ and $Y(v_0)=1$. 
In each step, in a Breadth-First Search way, 
we extend the partial assignments using the optimal coupling between two marginal distributions to a new variable. 
At last, we obtain a set of variables $V_1^{\+C_v}$ which is a superset of all variables on which $X$ and $Y$ disagree. 
Therefore,
\[
  \Pr[(\sigma_X,\sigma_Y)\sim \+C_v]{\sigma_X(v)\ne\sigma_Y(v)}\le \Pr[\+C_v]{v\in V_1^{\+C_v}}.
\]

We then construct another coupling $\+C$ of distributions $\mu(\cdot\mid X(v_0))$ and
$\mu(\cdot\mid Y(v_0))$ in a similar way, where $v_0\in\+M$ is the unique vertex on which $X$ and
$Y$ differ. The coupling also produces a set $V_1$ which is a superset of all variables with different values. 
We carefully define the coupling $\+C$ so that for every $v\in \+M\setminus\set{v_0}$, it
holds that
\begin{equation}\label{eqn-couple-coupling}
  \Pr[\+C_v]{v\in V_1^{\+C_v}}=\Pr[\+C]{v\in V_1^{\+C}}.
\end{equation}
Recall that $D_{v_0}=0$.
Therefore, we only need to bound $D_v$ for those $v \in {\+M} \setminus \{v_0\} $.
Hence,
\begin{align*}
  \sum_{v\in\+M}D_v& =\sum_{v\in\+M\setminus\set{v_0}}D_v 
   \le \sum_{v\in\+M\setminus\set{v_0}}\Pr[\+C_v]{v\in V_1^{\+C_v}} \\
   \tag{by \eqref{eqn-couple-coupling}}& = \sum_{v\in\+M\setminus\set{v_0}}\Pr[\+C]{v\in V_1^{\+C}} = \E[\+C]{\abs{V_1^{\+C}}}-1.  
\end{align*}
Finally, we bound $\E[\+C]{\abs{V_1^{\+C}}}$ by enumerating all induced paths in the square of the line graph of $H_\Phi$ (\Cref{definition-line-graph}) rooted at $v_0$.


%
%

We describe the coupling $\+C_v$ in Section~\ref{section-coupling-Cv} and the coupling $\+C$ in
Section~\ref{section-coupling-C}. 
And finally in Section~\ref{section-proof-of-dv-sum}, Lemma~\ref{lemma-bound-dv-sum} is proved by a coupling between the two couplings $\+C_v$ and $\+C$.

\subsubsection{The coupling $\+C_v$}\label{section-coupling-Cv}

First we define the following distribution $\nu=\nu^{(v)}$ over all assignments in $\{0,1\}^V$.
\begin{definition}
  \label{definition-nu}
  Fix a variable $v \in \+M \setminus \{v_0\}$. Let $\nu=\nu^{(v)}$ be the distribution $\mu$
  conditional on the assignment of the set $\Lambda = \+M \setminus \{v_0,v\} $ is specified as
  $X(\Lambda) = Y(\Lambda)$, where $X, Y \in \{0,1\}^{\+M}$ differ at only $v_0$. Formally,
  \begin{align}
  \label{eq-def-nu-proof}
    \forall \sigma \in \{0,1\}^{V}:\quad \nu(\sigma) =\frac{\one{\sigma(\Lambda)= X(\Lambda)} \cdot \mu(\sigma) }{ \sum_{\tau \in \{0,1\}^V}\one{\tau(\Lambda)= X(\Lambda)}\cdot\mu(\tau)}.
  \end{align}
\end{definition}
Note that if $2^{k_{\beta}} \geq 2\mathrm{e}dk$, then by \Cref{lemma-stationary}, the distribution
$\nu$ is well-defined.

For every $v\in\+M\setminus\set{v_0}$, the coupling $\+C_v$ generates a pair of random assignments $X^{\+C_v}, Y^{\+C_v} \in \{0,1\}^V$.  
The projection $X^{\+C_v}$ (or $Y^{\+C_v}$) has the law $\nu$ conditioned on $X^{\+C_v}(v_0) = X(v_0) = 0$ (or on $Y^{\+C_v}(v_0) = Y(v_0) = 1$). 
%
Let $k_\gamma \geq 1$ be an integer parameter to be specified later satisfying
$k_\gamma < k_{\beta}$ and
\begin{align}
  \label{eq-assume-s}
  2^{k_\beta-k_\gamma} \geq 2\mathrm{e}ds, \text{ where } s \triangleq 36d^4k^5.	
\end{align}
We then define two parameters $\plow$ and $\pup$ as follows:
\begin{equation}
  \label{eq-def-plow-pup}
  \begin{split}
    \plow &\triangleq \frac{1}{2} - \frac{1}{s},\\
    \pup &\triangleq \frac{1}{2} + \frac{1}{s}.
  \end{split}
\end{equation}
We will see later that $[\plow,\pup]$ is the interval in which the marginal probability on a single
variable can locate during the process of the coupling.

Recall that $H_{\Phi}=(V, \mathcal{E})$ is the hypergraph for $\Phi$ defined in~\eqref{eq-def-Hphi}.
The coupling procedure $\+C_v$ is similar to the one used in~\cite{guo2019counting}, which is an
\emph{adaptive} version of the coupling appeared in~\cite{Moi19}.
%

The coupling procedure $\+C_v$ is described in \Cref{alg-coupling-v},
where we fix an arbitrary ordering of all clauses and all variables. 
The meanings of some variables appear in the algorithm are
\begin{itemize}
\item $V_1$ - a superset of all \emph{discrepancy} variables. 
  It contains all variables on which $X^{\+C_v}$ and $Y^{\+C_v}$ disagree. 
  It may contain some additional variables to ease our analysis later.
\item $\Vcol$ - the variables whose values have been determined in the BFS process. $X^{\+C_v}$ and
  $Y^{\+C_v}$ can either agree or disagree on them.
\item $\+S$ - a subset of $\Vcol$ on which $X^{\+C_v}$ and $Y^{\+C_v}$ agree. 
  The coupling guarantees that $\+S\cap\+M=\varnothing$. 
  Intuitively $\+S$ together with $\+M$ separates discrepancy variables from the rest.
\end{itemize}
The algorithm keeps growing the set $V_1$ in a BFS manner until there is no unassigned variable on the boundary of $V_1$. 
We remark that some of the choices in \Cref{alg-coupling-v} may seem confusing at first.
They are because we need to later compare it with $\+C$ to show \eqref{eqn-couple-coupling}.
For example,
we may choose $u\in\+M_{v_0}$ in Line~\ref{line-find-u-Cv}.
Since we are coupling $\nu$ conditioned on $v_0$ being $0$ and $1$ respectively,
any $u\in\+M_{v_0}$ is guaranteed to be coupled successfully according to $X(u)=Y(u)$.
However, we may still put $u$ into $V_1$.
This is a vacuous step that merely serves the purpose of comparing with $\+C$ later,
because we want to guarantee that under a suitable coupling,
the set $V_1$ generated by $\+C_v$ is the same as $\+C$.


\begin{algorithm}[h]
  \SetKwInOut{Input}{Input} \SetKwInOut{Output}{Output} 
  \Input{a CNF formula $\Phi$, a hypergraph
    $H_\Phi = (V,\mathcal{E})$, a set of marked variables $\+M$, a
     variable $v_0 \in \+M$, the distribution $\nu$ in~\eqref{eq-def-nu-proof}, the parameters $\plow,\pup$ in~\eqref{eq-def-plow-pup}, a parameter $k_\gamma > 0$ such that
    $k_\gamma < k_\beta$;}  
  \Output{a pair of assignments $X^{\+C_v}, Y^{\+C_v} \in \{0,1\}^{V}$.}
  $X^{\+C_v}(v_0) = 0$ and $Y^{\+C_v}(v_0) = 1$\; 
  $V_1 \gets \{v_0\}$, $V_2 \gets V \setminus V_1$, $\Vcol \gets \{v_0\}$ and $\mathcal{S} \gets \emptyset$\; 
  \While{$\exists e \in \mathcal{E}$ s.t.\ $e \cap V_1 \neq \emptyset, (e \cap V_2) \setminus \Vcol \neq \emptyset$\label{line-while-condition} }
    {let $e$ be the first such hyperedge and $u$ be the first variable in $(e \cap V_2) \setminus \Vcol$\label{line-find-u-Cv}\; 
    sample a real number $r_u \in [0,1]$ uniformly at random\label{line-draw-u-Cv}\; 
    let $p^{X}_u = \nu_u(0\mid X^{\+C_v})$ and $p^{Y}_u = \nu_u(0\mid Y^{\+C_v})$\label{line-pX-pY-Cv}\; 
    extend  $X^{\+C_v}$ to variable $u$ s.t.\ $X^{\+C_v}(u) = 0$ if $r_u \leq p^X_u$, o.w.\ $X^{\+C_v}(u) = 1$\label{line-set-X-Cv}\; 
    extend $Y^{\+C_v}$ to variable $u$ s.t.\ $Y^{\+C_v}(u) = 0$ if $r_u \leq p^Y_u$, o.w.\ $Y^{\+C_v}(u) = 1$\label{line-set-Y-Cv}\;
    $\Vcol \gets \Vcol \cup \{u\}$\label{line-add-vset-Cv}\; 
    \If{$\plow < r_u \leq \pup$\label{line-condition-V1-Cv}}
      { $V_1 \gets V_1 \cup \{u\}$, $V_2 \gets V \setminus V_1$\label{line-add-u-1-Cv}\;} 
    \If{$(u \not\in \+M) \land (r_u \leq \plow \lor r_u > \pup)$\label{line-condition-S-Cv}}
      {$\mathcal{S} \gets \mathcal{S} \cup \{u\}$\;} 
    \For{$e \in \mathcal{E}$ s.t.\ $e$ is satisfied by both $X^{\+C_v}(\mathcal{S})$ and $Y^{\+C_v}(\mathcal{S})$}
      {$\mathcal{E} \gets\mathcal{E} \setminus \{e\}$\label{line-delete-1-Cv}\; } 
    \For{$e \in \mathcal{E}$ s.t.\ $|e \cap (\Vcol \setminus \+M) | = k_\gamma $\label{line-adaptive-mark-Cv}}
      {$V_1 \gets V_1 \cup (e \setminus \Vcol )$,
       $V_2 \gets V \setminus V_1$\label{line-add-u-2-Cv}\;
       } 
    } 
    extend $X^{\+C_v}$ and $Y^{\+C_v}$ further on the set $V_2 \setminus \Vcol$ 
    using the optimal coupling between $\nu_{V_2 \setminus \Vcol}(\cdot\mid X^{\+C_v}(\Vcol))$ and $\nu_{V_2 \setminus \Vcol}(\cdot\mid Y^{\+C_v}(\Vcol))$\label{line-sample-V2-Cv}\;
    extend $X^{\+C_v}$ and $Y^{\+C_v}$ further on the set $V_1 \setminus \Vcol$ 
    using the optimal coupling between $\nu_{V_1 \setminus \Vcol}(\cdot\mid X^{\+C_v}(\Vcol\cup V_2))$ and $\nu_{V_1 \setminus \Vcol}(\cdot\mid Y^{\+C_v}(\Vcol \cup V_2))$\label{line-extend-Cv}\;
    \Return{$(X^{\+C_v}, Y^{\+C_v})$\;}
  \caption{The coupling procedure $\+C_v$}\label{alg-coupling-v}
\end{algorithm}

\begin{lemma}
  \label{lemma-marginal-Cv}
  The following properties hold for the coupling procedure $\+C_v$ in
  Algorithm~\ref{alg-coupling-v}.
  \begin{itemize}
  \item The coupling procedure $\+C_v$ terminates eventually and returns a pair $X^{\+C_v},Y^{\+C_v} \in \{0,1\}^V$ such that 
    $X^{\+C_v}$ and $Y^{\+C_v}$ have the law $\nu$ conditioned on $X^{\+C_v}(v_0)  = 0$ and on $Y^{\+C_v}(v_0)  = 1$, respectively. 
  \item If $2^{k_\beta-k_\gamma} \geq 2\mathrm{e}ds \text{ where } s = 36d^4k^5$, 
    then $X^{\+C_v}(V_2) = Y^{\+C_v}(V_2 )$.
  \end{itemize}
\end{lemma}

We need the following lemma to prove Lemma~\ref{lemma-marginal-Cv}.
\begin{lemma}
  \label{lemma-px-py-Cv}
  In the coupling procedure $\+C_v$, if
  $2^{k_\beta-k_\gamma} \geq 2\mathrm{e}ds \text{ where } s = 36d^4k^5$, then for each $p^X_u$ and
  $p^Y_u$ computed Line~\ref{line-pX-pY-Cv}, if $u \in \+M \setminus \{v_0, v\}$, then $p^X_u = p^Y_u$;
  if $u \not\in \+M \setminus \{v\}$, then it holds that
  \begin{align*}
    \plow \leq p^X_u, p^Y_u \leq \pup.	
  \end{align*}
\end{lemma}
\begin{proof}
  We prove the lemma by considering the two cases.

  \textbf{Case 1:} $u \in \+M \setminus \{v_0,v\}$. Due to the definition of $\nu$
  (Definition~\ref{definition-nu}), it must hold that $p^X_u = p^Y_u =0$ or $p^X_u = p^Y_u = 1$,
  which implies $p^X_u = p^Y_u$.

  \textbf{Case 2:} $u \not\in \+M \setminus \{v\}$. 
  We prove the lemma for $p^X_u$.  For $p^Y_u$ it holds similarly.
  In each step, we have
  $X^{\+C_v} \in \{0,1\}^{\Vcol}$.
  Due to \Cref{definition-nu}, the distributions $\nu_u(\cdot \mid X^{\+C_v})$ is the distribution $\mu$ conditional on the values of
  variables in $\+M_v \cup \Vcol$ are fixed. 
  We use $\+E_{H}$ to denote the set of all hyperedges in hypergraph $H_{\Phi}$.
  We claim that for each execution of Line~\ref{line-pX-pY-Cv}, the following property holds
  \begin{align}
  \label{eq-claim-unfixed-variable}
  \forall e \in \+E_H:\, |e \cap (\Vcol \setminus \+M) | \leq k_{\gamma}\lor\text{the clause represented by $e$ is satisfied by $X^{\+C_v}$}.	
  \end{align}
  By Condition~\ref{condition-marked-variables}, each hyperedge contains at least $k_{\beta}$ unmarked variables. By~\eqref{eq-claim-unfixed-variable}, for each hyperedge that is not satisfied by the current $X^{\+C_v}$, it contains at least $k_{\beta}-k_{\gamma}$ unmarked variables whose value are not fixed by the current $X^{\+C_v}$.
   By the definition of the distribution $\nu$ and~\Cref{corollary-local-uniform}, if $2^{k_{\beta}-k_{\gamma}} \geq 2eds$ where $s = 36d^4k^5$, then
  \begin{align*}
    p^X_u &= \nu_u(0\mid X^{\+C_v}) \leq \frac{1}{2}\exp\left(\frac{1}{s} \right) \leq \frac{1}{2}\left(1 + \frac{2}{s}\right)\leq \frac{1}{2} + \frac{1}{s},\\
    1 - p^X_u &= \nu_u(1\mid X^{\+C_v}) \leq \frac{1}{2}\exp\left(\frac{1}{s} \right) \leq \frac{1}{2}\left(1 + \frac{2}{s}\right)\leq \frac{1}{2} + \frac{1}{s}.   
  \end{align*}
  
 We now prove~\eqref{eq-claim-unfixed-variable}. 
 Note that at the beginning of the coupling procedure $\+C_v$, the set $\Vcol = \{v_0\} \subseteq \+M$, and thus for all hyperedges $e \in \+E$, it holds that $|e \cap (\Vcol \setminus \+M) | = 0$. Hence, the property in~\eqref{eq-claim-unfixed-variable} holds at the beginning. 
 
 Suppose in some execution of Line~\ref{line-pX-pY-Cv}, there is a hyperedge $e$ that violates the property in~\eqref{eq-claim-unfixed-variable}. 
 Formally, the clause represented by $e$ is not satisfied by $X^{\+C_v}$ and $|e \cap (\Vcol \setminus \+M) | > k_{\gamma}$.
 Then we can find the first round of the while-loop after which the clause represented by $e$ is not satisfied by $X^{\+C_v}$ and $|e \cap (\Vcol \setminus \+M) | = k_{\gamma}$.
 Denote this round by $R$.
 In round $R$ and any previous round of $R$, the clause represented by $e$ cannot be satisfied by $X^{\+C_v}$. Hence $e$ cannot be deleted in Line~\ref{line-delete-1-Cv} up to round $R$. 
 Since $|e \cap (\Vcol \setminus \+M) | = k_{\gamma}$, 
 $e$ satisfies the condition in Line~\ref{line-adaptive-mark-Cv}. 
 After Line~\ref{line-add-u-2-Cv}, we have $e \subseteq V_1 \cup \Vcol$,
 which means that, after the round $R$, any vertex $u \in e$ cannot be pick in Line~\ref{line-find-u-Cv}. 
 Hence, it holds that  $|e \cap (\Vcol \setminus \+M) | = k_{\gamma}$ after the round $R$,
 which contradicts to the assumption that $|e \cap (\Vcol \setminus \+M) | > k_{\gamma}$.
\end{proof}

\begin{proof}[Proof of Lemma~\ref{lemma-marginal-Cv}]
  Firstly, we prove that the coupling procedure must terminate. 
  This is because the size of the set $\Vcol$ is increased by one in each while-loop.

  Secondly, we prove that the final $X^{\+C_v}$ follows the distribution $\nu$ conditional on
  $X^{\+C_v}(v_0) = X(v_0) = 0$. The same argument applies to $Y^{\+C_v}$. At the
  beginning, we set $X^{\+C_v}(v_0) = 0$. Note that, in each step, it holds that
  $X^{\+C_v} \in \{0,1\}^{\Vcol}$ and the algorithm always extends $X^{\+C_v}$ according to the
  distribution $\nu$ conditional on the current assignment on $\Vcol$. By the chain rule, it is
  easy to verify the final $X^{\+C_v}$ follows the distribution $\nu$ conditional on
  $X^{\+C_v}(v_0) = X(v_0) = 1$.

  Finally, consider the final sets $V_1, V_2 , \mathcal{S},\Vcol$ and the final assignments
  $X^{\+C_v}$ and $Y^{\+C_v}$.
  We prove the following two properties.
  \begin{enumerate}[label={(\roman*)}]
  \item \label{fact-1} The two distributions $\nu_{V_2\setminus \Vcol}(\cdot \mid X^{\+C_v}(\Vcol)) $ and $\nu_{V_2\setminus \Vcol}(\cdot \mid X^{\+C_v}(\Vcol \cap V_2))$ are identical; 
    and the two distributions $\nu_{V_2\setminus \Vcol}(\cdot \mid Y^{\+C_v}(\Vcol)) $ and $\nu_{V_2\setminus \Vcol}(\cdot \mid Y^{\+C_v}(\Vcol \cap V_2))$ are identical. 
    \item \label{fact-2}$X^{\+C_v}(\Vcol \cap V_2) = Y^{\+C_v}(\Vcol \cap V_2)$.
  \end{enumerate}
  If the above two properties \ref{fact-1} and \ref{fact-2} hold, then
  $\nu_{V_2 \setminus \Vcol}(\cdot\mid X^{\+C_v}(\Vcol))$ and
  $\nu_{V_2 \setminus \Vcol}(\cdot\mid Y^{\+C_v}(\Vcol))$ can be perfectly coupled, which implies
  $X^{\+C_v}(V_2 \setminus \Vcol ) = Y^{\+C_v}(V_2 \setminus \Vcol)$. 
  Combining with Property~\ref{fact-2}, it proves that $X^{\+C_v}(V_2) = Y^{\+C_v}(V_2)$.

  We now prove Property \ref{fact-1}.
We show that two distributions $\nu_{V_2\setminus \Vcol}(\cdot \mid X^{\+C_v}(\Vcol)) $ and $\nu_{V_2\setminus \Vcol}(\cdot \mid X^{\+C_v}(\Vcol \cap V_2))$ are identical.
  For $Y^{\+C_v}$ it holds similarly.
  First observe that $\+S\subseteq V_2$.
  This is because a variable $u$ is added to $V_1$ either because the condition in Line~\ref{line-condition-V1-Cv} holds or because of Line~\ref{line-add-u-2-Cv}.
  In the first case, the condition in Line~\ref{line-condition-S-Cv} does not hold and $u$ will never be added to $\+S$.
  In the second case, $u\not\in\Vcol$ and thus $u\not\in\+S$ as well.
  Once a variable $u$ is added into $V_1$, $u$ cannot be picked in Line~\ref{line-find-u-Cv}, and thus $u$ cannot be added in $\mathcal{S}$ for the rest of the coupling.
  
  For any clause $c$ in the original CNF formula $\Phi$ such that $\vbl{c} \cap V_1 \neq \emptyset$ and $\vbl{c} \cap V_2 \neq \emptyset$, 
  we claim that one of the following properties must hold:
  \begin{itemize}
    \item The clause $c$ is satisfied by the assignment $X^{\+C_v}(\mathcal{S})$;
    \item The clause $c$ satisfies $\vbl{c} \cap V_2 \subseteq \Vcol$.
  \end{itemize}
  All clauses spanning both $V_1$ and $V_2\setminus \Vcol$ are in the first case,
  and they are satisfied by $X^{\+C_v}(\Vcol\cap V_2)$ as $\mathcal{S} \subseteq \Vcol\cap V_2$.
  This implies Property~\ref{fact-1}.

  We show the claim next. 
  Suppose there exists a clause $c$ with $\vbl{c} \cap V_1 \neq \emptyset$ and $\vbl{c} \cap V_2 \neq \emptyset$ 
  such that $c$ is not satisfied by $X^{\+C_v}(\mathcal{S})$ and $\vbl{c} \cap V_2 \not\subseteq \Vcol$. 
  Let $e$ denote the hyperedge that represents $c$ in $H_{\Phi}$.  
  Since the coupling procedure terminates, 
  the hyperedge $e$ must be deleted in Line~\ref{line-delete-1-Cv} during the coupling procedure $\+C_v$. 
  Otherwise, $e$ satisfies the condition in Line~\ref{line-while-condition}, and the coupling  procedure cannot terminate.
  However, since $c$ is not satisfied by $X^{\+C_v}(\mathcal{S})$ after the whole coupling procedure, 
  $c$ cannot be satisfied by $X^{\+C_v}(\mathcal{S})$ during the coupling procedure. 
  This implies that $e$ cannot be deleted in Lines~\ref{line-delete-1-Cv}.

  We then prove Property~\ref{fact-2}. 
  Suppose $X^{\+C_v}(\Vcol \cap V_2) \neq Y^{\+C_v}(\Vcol \cap V_2)$. 
  Let $u \in \Vcol \cap V_2$ be a variable such that $X^{\+C_v}(u) \neq Y^{\+C_v}(u)$. 
  Since $u \in \Vcol$ and $u \neq v_0$, 
  the coupling have computed $p^X_u,p^Y_u$ in Line~\ref{line-pX-pY-Cv}.  
  Since $X^{\+C_v}(u) \neq Y^{\+C_v}(u)$, it must be that $p^X_u \neq p^Y_u$. 
  By Lemma~\ref{lemma-px-py-Cv}, we know that $\plow \leq p^X_u, p^Y_u \leq \pup$. 
  By Lines~\ref{line-set-X-Cv} and~\ref{line-set-Y-Cv}, since $X^{\+C_v}(u) \neq Y^{\+C_v}(u)$, 
  \begin{align*}
    \plow < r_u\leq \pup,		
  \end{align*}
  where $r_u \in [0,1]$ is drawn in Line~\ref{line-draw-u-Cv}. 
  In this case, the variable $u$ must be added into $V_1$ in Line~\ref{line-add-u-1-Cv}
  and $u$ stays in $V_1$ for the rest of the coupling. 
  However, by assumption, $u \in V_2=V\setminus V_1$.
  Contradiction.
%
%
\end{proof}

By Lemma~\ref{lemma-marginal-Cv}, we know that the marginal distribution of $X^{\+C_v}(v)$ is
identical to $\nu_v(\cdot \mid X^{\+C_v}(v_0) = 0)$. By Definition~\ref{definition-nu}, we know that
$X^{\+C_v}(v)$ follows the law $\mu_v(\cdot \mid X({\+M}_v ))$. Similarly, we know that
$Y^{\+C_v}(v)$ follows the law $\mu_v(\cdot \mid Y({\+M}_v ))$. By \Cref{prop:coupling}, 
we have that $\forall v \in \+M \setminus \{v_0\}$, 
\begin{align*}
  D_v &= \DTV{\mu_v(\cdot \mid X({\+M}_v ))}{\mu_v(\cdot \mid Y({\+M}_v ))}\\
      &\leq \Pr[\+C_v]{X^{\+C_v}(v) \neq Y^{\+C_v}(v)}\\
      &\leq \Pr[\+C_v]{v \in V_1}.
\end{align*}
The last inequality holds because by Lemma~\ref{lemma-marginal-Cv}, 
if $X^{\+C_v}(v) \neq Y^{\+C_v}(v)$, then $v \not\in V_2$ and thus $v \in V_1$.  
Note that $D_{v_0} = 0$. 
The sum of all $D_v$ can be bounded as follows
\begin{align}
  \label{eq-bound-sum-Pr}
  \sum_{v \in {\+M}}D_v = \sum_{v \in {\+M} \setminus\{v_0\} }D_v \leq  \sum_{v \in {\+M} \setminus\{v_0\} }\Pr[\+C_v]{v \in V_1^{\+C_v}},
\end{align}
where we use $V_1^{\+C_v}$ to denote the set $V_1$ generated by the coupling procedure $\+C_v$.

\subsubsection{The coupling $\+C$}\label{section-coupling-C}

To bound the sum of all $\Pr[\+C_v]{v \in V_1^{\+C_v}}$, we introduce the coupling procedure $\+C$ in Algorithm~\ref{alg-coupling}. 
The coupling $\+C$ is basically the same as $\+C_v$ except that \emph{it treats all variables in $\+M_{v_0}$ as free variables}.
This difference is reflected in Line~\ref{line-pX-pY-C} of \Cref{alg-coupling},
where we use conditional distribution of $\mu$ instead of $\nu$ in Line~\ref{line-pX-pY-Cv} of \Cref{alg-coupling-v}.
However, as $\plow$ and $\pup$ stay the same,
we can construct a coupling of two couplings $\+C$ and $\+C_v$ such that the final set $V_1$ does not change.
In this way, we obtain a uniform treatment for $\Pr[\+C_v]{v \in V_1^{\+C_v}}$ for all $v$,
which leads to a better bound comparing to analysing $\Pr[\+C_v]{v \in V_1^{\+C_v}}$ individually. 

\begin{algorithm}[ht]
  \SetKwInOut{Input}{Input} \SetKwInOut{Output}{Output} \Input{a CNF formula $\Phi$, a hypergraph
    $H_\Phi = (V,\mathcal{E})$, a set of marked variables $\+M$, a
     variable $v_0 \in \+M$, the parameters $\plow,\pup$ in~\eqref{eq-def-plow-pup}, a parameter $k_\gamma > 0$ such that
    $k_\gamma < k_\beta$;} \Output{a pair of assignments $X^{\+C}, Y^{\+C} \in \{0,1\}^{\Vcol}$ for some
    random set $\Vcol \subseteq V$.}  $X^{\+C}(v_0) \gets 0$ and $Y^{\+C}(v_0) \gets 1$\;
  $V_1 \gets \{v_0\}$, $V_2 \gets V \setminus V_1$, $\Vcol \gets \{v_0\}$ and
  $\mathcal{S} \gets \emptyset$\; \While{$\exists e \in \mathcal{E}$ s.t.\ $e \cap V_1 \neq \emptyset, (e \cap V_2) \setminus \Vcol \neq \emptyset$ }{ let $e$ be the first such hyperedge and $u$
    be the first variable in $(e \cap V_2) \setminus \Vcol$\label{line-find-u-C}\; sample a random
    real number $r_u \in [0,1]$ uniformly at random\label{line-sample-r}\;
    let $p^{X}_u = \mu_u(0\mid X^{\+C})$ and $p^{Y}_u = \mu_u(0\mid Y^{\+C})$\label{line-pX-pY-C}\;
    extend $X^{\+C}$ further on variable $u$ s.t. $X^{\+C}(u) = 0$ if $r_u \leq p^X_u$, o.w.\ $X^{\+C}(u) = 1$\; 
    extend $Y^{\+C}$ further on variable $u$ s.t. $Y^{\+C}(u) = 0$ if $r_u \leq p^Y_u$, o.w.\ $Y^{\+C}(u) = 1$\; 
    $\Vcol \gets \Vcol \cup \{u\}$\;
    \If{$\plow < r_u \leq \pup$}{ $V_1 \gets V_1 \cup \{u\}$,
      $V_2 \gets V \setminus V_1$\label{line-add-u-1-C}\; } \If{
      $(u \not\in \+M) \land (r_u \leq \plow \lor r_u > \pup)$}{
      $\mathcal{S} \gets \mathcal{S} \cup \{u\}$ \label{line-set-S}\; }
    \For{$e \in \mathcal{E}$ s.t.\ $e$ is satisfied by both $X^{\+C}(\mathcal{S})$ and
      $Y^{\+C}(\mathcal{S})$\label{line-delete-1-condition}}{
      $\mathcal{E} \gets\mathcal{E} \setminus \{e\}$\label{line-delete-1}\; }
    \For{$e \in \mathcal{E}$ s.t.\ $|e \cap (\Vcol \setminus \+M) | = k_\gamma$ \label{line-delete-2-condition} }
      {$V_1 \gets V_1 \cup (e \setminus \Vcol )$; $V_2 \gets V \setminus V_1$\label{line-add-u-2-C}\;
      } 
     }
  \Return{$(X^{\+C}, Y^{\+C})$\;}
  \caption{The coupling procedure $\+C$}\label{alg-coupling}
\end{algorithm}

To be more precise, we have the following lemma.
\begin{lemma}
  \label{lemma-marginal-C}
  The following properties hold for the coupling procedure $\+C$ in Algorithm~\ref{alg-coupling}.
  \begin{itemize}
  \item The coupling procedure $\+C$ terminates eventually and returns a pair
    $X^{\+C},Y^{\+C} \in \{0,1\}^{\Vcol}$ for a random set $\Vcol\subseteq V$ such that $v_0\in\Vcol$.
  \item If $2^{k_\beta-k_\gamma} \geq 2\mathrm{e}ds \text{ where } s = 36d^4k^5$, then for any
    variable $v \in \+M \setminus\{v_0\}$,
    \begin{align*}
      \Pr[\+C_v]{v \in V_1^{\+C_v}} = \Pr[\+C]{v \in V^{\+C}_1},	
    \end{align*}
    where $V_1^{\+C_v}$ is the set $V_1$ generated by the coupling procedure $\+C_v$ and $V^{\+C}_1$
    is the set $V_1$ generated by the coupling procedure $\+C$.
  \end{itemize}
\end{lemma}

We need the following lemma,
which is the analogue of \Cref{lemma-px-py-Cv}.
It follows from the same proof of the second case of \Cref{lemma-px-py-Cv}.
\begin{lemma}
  \label{lemma-px-py-C}
  In the coupling procedure $\+C$, if
  $2^{k_\beta-k_\gamma} \geq 2\mathrm{e}ds \text{ where } s = 36d^4k^5$, then for each $p^X_u$ and
  $p^Y_u$ computed Line~\ref{line-pX-pY-C}, it holds that
  \begin{align*}
    \plow \leq p^X_u, p^Y_u \leq \pup.	
  \end{align*}
\end{lemma}

\begin{proof}[Proof of Lemma~\ref{lemma-marginal-C}]
  We first show that the coupling procedure must terminate. This is because 
  the size of the set $\Vcol$ is increased by one in each while-loop.

  Fix a variable $v \in \+M \setminus \{v_0\}$. Consider the coupling procedure $\+C_v$
  (Algorithm~\ref{alg-coupling-v}) and the coupling procedure $\+C$ (Algorithm~\ref{alg-coupling}). 
  We couple the two procedures by sampling the same random real number $r_u \in [0,1]$ for each variable $u$.
  We claim that the following invariant holds for the two coupling procedures:
  \begin{equation}
    \label{eq-identical}
    \begin{split}
      V_1^{\+C_v} = V^{\+C}_1&, \quad  V_2^{\+C_v} = V^{\+C}_2, \quad \Vcol^{\+C_v} = \Vcol^{\+C}, \quad \+E^{\+C_v} = \+E^{\+C},\\
      \+S^{\+C_v} = \+S^{\+C}&,\quad X^{\+C_v}(\+S^{\+C_v}) = X^{\+C}(\+S^{\+C}), \quad
      Y^{\+C_v}(\+S^{\+C_v}) = Y^{\+C}(\+S^{\+C})
    \end{split}
  \end{equation}
  This implies that $ V_1^{\+C_v} = V^{\+C}_1$ in the end, which is the second item of the lemma. 
  We show~\eqref{eq-identical} by induction.

  Initially, it holds that $ V_1^{\+C_v} = V^{\+C}_1 = \{v_0\}$,
  $V_2^{\+C_v} = V^{\+C}_2 = V \setminus \{v_0\}$, $\Vcol^{\+C_v} = \Vcol^{\+C} = \{v_0\}$ and
  $\+S^{\+C_v} = \+S^{\+C} = \emptyset$.

  For each step of the while-loop, suppose~\eqref{eq-identical} holds, then two coupling procedure
  pick the same hyperedge $e$ and the same vertex $u \in e$. 
  The two coupling procedures sample the same random number $r_u$
  and use the same parameters $\plow$ and $\pup$ in~\eqref{eq-def-plow-pup}.
  Hence, after the Line~\ref{line-set-S} of either coupling,
  $V_1^{\+C_v} = V^{\+C}_1$, $V_2^{\+C_v} = V^{\+C}_2$, $\Vcol^{\+C_v} = \Vcol^{\+C}$, $\+E^{\+C_v} = \+E^{\+C}$, and $\+S^{\+C_v} = \+S^{\+C}$. 
  Note that if the variable $u$ is added into $\+S$ in Line~\ref{line-set-S}, 
  then it must be that $(u \not\in \+M) \land (r_u \leq \plow \lor r_u > \pup)$. 
  If $r_u \leq \plow$, 
  then by Lemma~\ref{lemma-px-py-Cv} and Lemma~\ref{lemma-px-py-C}, 
  in both coupling procedures $r_u \leq p^X_u$ and $r_u \leq p^Y_u$, which implies
  \begin{align*}
       X^{\+C_v}(u) = X^{\+C}(u) = Y^{\+C_v}(u) = Y^{\+C}(u) = 0. 	
  \end{align*}
  Similarly, if $r_u > \pup$, then
  \begin{align*}
    X^{\+C_v}(u) = X^{\+C}(u) = Y^{\+C_v}(u) = Y^{\+C}(u) = 1. 	
  \end{align*}
  Hence, the invariant in~\eqref{eq-identical} holds after the Line~\ref{line-set-S}.  
  It is easy to verify that after the rest of the while-loop, 
  the invariants in~\eqref{eq-identical} still hold.
\end{proof}

By Lemma~\ref{lemma-marginal-C} and inequality~\eqref{eq-bound-sum-Pr}, we have
\begin{align*}
  \sum_{v \in {\+M}}D_v &\leq 	 \sum_{v \in {\+M} \setminus\{v_0\} }\Pr[\+C_v]{v \in V_1^{\+C_v}}\\
                        &\leq \sum_{v \in V \setminus \{v_0\} }\Pr[\+C]{v \in V_1^{\+C}}\\
                        &= \E[\+C]{|V_1^{\+C}|} - \Pr[\+C]{v_0 \in V_1^{\+C}}\\
                        & = \E[\+C]{|V_1^{\+C}|} - 1,
\end{align*}
where the last equation holds because $v_0$ must be in the set $V_1^{\+C}$.
Our next step is to bound $\E[\+C]{|V_1^{\+C}|}$.

\subsubsection{The proof of Lemma~\ref{lemma-bound-dv-sum}}\label{section-proof-of-dv-sum}

Finally, we finish the proof of Lemma~\ref{lemma-bound-dv-sum} by proving the following lemma.

\begin{lemma}
  \label{lemma-Ex-V1}
  In the coupling procedure $\+C$ (Algorithm~\ref{alg-coupling}), if $2^{k_{\gamma}} \geq 36d^4k^4$
  and $2^{k_{\beta}-k_{\gamma}} \geq 2eds$ where $s = 36d^4k^5$, it holds that
  \begin{align*}
    \E[\+C]{|V_1|}\leq \frac{3}{2}.
  \end{align*}
\end{lemma}

In Lemma~\ref{lemma-Ex-V1}, we can take
\begin{align*}
  k_{\gamma} = \left \lceil \frac{4}{9}k_{\beta} \right\rceil.	
\end{align*}
Then, the following condition is sufficient to imply the condition of Lemma~\ref{lemma-Ex-V1}:
\begin{align}
  \label{eq-condition-beta}
  2^{k_{\beta}} \geq (36)^{\frac{9}{4}}d^9k^9, \quad 	2^{k_{\beta}}\geq   (144\mathrm{e})^{\frac{9}{5}}d^{9}k^{9}.
\end{align}
Note that $2^{k_{\beta}}\geq 2^{16}d^9k^9$ is a sufficient condition for~\eqref{eq-condition-beta}. 

Consider the coupling procedure $\+C$ defined in Algorithm~\ref{alg-coupling}.  Upon termination,
the coupling procedure generates assignments $X^{\+C}$ and $Y^{\+C}$, and the sets of variables
$V_1, V_2, \Vcol,\mathcal{S} \subseteq V$.  We define the failed hyperedge as follows.
\begin{definition}[failed hyperedge]
  \label{definition-failed-edge}
  We say a hyperedge $e \in \mathcal{E}$ is failed if one of the following events occurs after the coupling procedure $\+C$:
  \begin{enumerate}[label={(\roman*)}]
    \item \label{reason-1} there exists $v \in (e \cap \Vcol) \setminus \{v_0\} $ such that $\plow < r_v \leq \pup$;
    \item \label{reason-2} $|e \cap (\Vcol \setminus \+M) | = k_{\gamma}$ and $e$ is not satisfied by both $X(\mathcal{S})$ and $Y(\mathcal{S})$.
  \end{enumerate}
\end{definition}
In the following, we will use Reason~\ref{reason-1} and Reason~\ref{reason-2} to denote the above two reasons of failure.

\begin{definition}[line graph]
  \label{definition-line-graph}
  Let $H=(V, \mathcal{E})$ be a hypergraph.  
  The line graph $\Lin(H)=(V_L, E_L)$ has hyperedges in $\+E$ as its vertices and two hyperedges are adjacent if they intersect, 
  i.e.\ $V_L = \mathcal{E}$ and $\{e_1,e_2\}\in E_L $ iff $e_1 \cap e_2 \neq \emptyset$.
\end{definition}

Let $\Lin^2(H)$ denote the power graph of $\Lin(H)$.  
Two vertices in $\Lin^2(H)$ are adjacent if and only if their distance in $\Lin(H)$ is at most $2$.
%
For any vertex $v \in V$, we define the sets $N_v, N_v^2$ of hyperedges as
\begin{align*}
N_v &\triangleq \{e \in \mathcal{E} \mid v \in e \};\\
  N_v^2 &\triangleq \{e \in \+E \mid (v \in e) \lor (\exists e' \in \mathcal{E} \text{ s.t. } e \cap e' \neq \emptyset \land v \in e') \}.	
\end{align*}
The set $N_v$ is the set of all hyperedges that contains $v$.
The set $N_v^2$ is the set of all hyperedges that either contains $v$ or intersects with some hyperedges containing $v$. 
The following lemma asserts that for any $v\in V_1$,
there are a path in $\Lin^2(H)$ that leads to $v$.
\begin{lemma}
  \label{lemma-exist-path}
  For any variable $v \in V \setminus \{v_0\}$, if $v \in V_1$, then there must exist a sequence of
  hyperedges $e_1,e_2,\ldots, e_{\ell}$ for some $\ell \geq 1$ such that the following properties hold:
  \begin{itemize}
    \item $e_1 \in N^2_{v_0}$ and $v \in e_{\ell}$;
    \item for all $1 \leq i \leq \ell $, the hyperedge $e_i$ is failed;
    \item for all $1 \leq i < \ell $, $e_{i}$ and $e_{i+1}$ are adjacent in $\Lin^2(H)$.
  \end{itemize}
\end{lemma}
\begin{proof}
%
%
  We first show that each variable $u \in V_1 \setminus \{v_0\}$ must be incident to a failed hyperedge.  
  For $u \in V_1 \setminus\{v_0\}$, $u$ is either added into $V_1$ in Line~\ref{line-add-u-1-C} or in Line~\ref{line-add-u-2-C}. 
  Suppose $u$ is added into $V_1$ in Line~\ref{line-add-u-1-C}. 
  In this case, the variable $u$ is picked in Line~\ref{line-find-u-C} due to some hyperedge $e$. 
  Then, it must be that $\plow \leq r_u \leq \pup$. 
  This implies that $u$ is incident to the failed hyperedge $e$ (for Reason~\ref{reason-1}). 
  Next suppose $u$ is added into $V_1$ in Line~\ref{line-add-u-2-C}. 
  In this case, $u \in e$ for some $e \in \mathcal{E}$ satisfying the condition in Line~\ref{line-delete-2-condition}. 
  Hence, the hyperedge $e$ is failed for Reason~\ref{reason-2} and $u$ is incident to $e$.
  If $e$ satisfies the condition in Line~\ref{line-delete-2-condition}, then $e \subseteq V_1 \cup \Vcol$ after Line~\ref{line-add-u-2-C}.
  Hence, the condition in Reason~\ref{reason-2} holds for $e$ for the rest of the coupling.

  Thus we only need to show the following claim: for each failed $e \in \mathcal{E}$, 
  there must exist a sequence of hyperedges $e_1,e_2,\ldots, e_{\ell}$ for some $\ell \geq 1$ such that the following properties hold:
  \begin{itemize}
    \item $e_1 \in N^2_{v_0}$ and $e = e_{\ell}$;
    \item for all $1 \leq i \leq \ell $, $e_i$ is failed;
    \item for all $1 \leq i < \ell $, $e_{i}$ and $e_{i+1}$ are adjacent in $\Lin^2(H)$.
  \end{itemize}
  
  Consider the execution of the coupling procedure $\+C$.
  We say a hyperedge $e \in \+E$ becomes failed once $e$ satisfies one of the reasons in \Cref{definition-failed-edge}. 
  Note that once a hyperedge becomes failed, it will stay failed for the rest of the coupling. 
  Moreover, the failed hyperedge must intersect the hyperedge satisfying the condition of the round of the while-loop in which it becomes failed.
  We list all failed hyperedges $e_{i_1},e_{i_2},\ldots,e_{i_r}$ such that $e_{i_j}$ is the $j$-th hyperedge that becomes failed. 
  Ties are broken arbitrarily.
  We prove the claim above by induction on the index $j$ from $1$ to $r$.

  For the base case, we only need to show that $e_{i_1}\in N^2_{v_0}$.
  Notice that $v\neq v_0$ and $v\in V_1$.
  If some hyperedge containing $v_0$ is failed,
  then $e_{i_1}\in N_{v_0}$.
  Otherwise, the only possibility that $V_1\neq\{v_0\}$ is that after setting a number of successfully coupled variables,
  there is a failed hyperedge satisfying Reason~\ref{reason-2}.
  In the round when this happens, the current hyperedge chosen in Line~\ref{line-find-u-C} must contain $v_0$ (otherwise $\+C$ terminates with $V_1=\{v_0\}$).
  The first such hyperedge is $e_{i_1}$ and thus $e_{i_1}\in N^2_{v_0}$.


  Suppose the claim holds for $e_{i_1},e_{i_2},\ldots,e_{i_{k-1}}$. We show the claim for $e_{i_k}$.
  Consider the round of the while-loop when $e_{i_k}$ becomes failed. 
  In Line~\ref{line-find-u-C} of this round, the coupling procedure picks a hyperedge $e$ and a variable $u\in e$ such that $e \cap V_1 \neq \emptyset$.
  As $e_{i_k}$ went failed in this round,
  either $e_k = e$ (due to Reason~\ref{reason-1}),
  or $e_{i_k} \in N_u$ (due to Reason~\ref{reason-2}).
  In both cases, $e\cap e_{i_k}\neq\emptyset$. 
  If $v_0 \in e$, then $e_{i_k}\in N^2_{v_0}$ and the claim holds by letting $e_1=e_{i_k}$. 
  Otherwise, since $e$ is picked in this round,
  there must exist a variable $u' \in V_1 \cap e$ and $u' \neq v_0$. 
  Thus $u'$ is incident to a failed hyperedge $e_{i_j}$ for some $1\leq j \leq k -1$. 
  Since $u'\in e\cap e_{i_j}$ and $e\cap e_{i_k}\neq\emptyset$,
  $e_{i_j}$ and $e_{i_k}$ are adjacent in $\Lin^2(H)$. 
  By the induction hypothesis, there exists a failed hyperedge path in $\Lin^2(H)$ that ends with $e_{i_j}$. 
  This proves the claim for $e_{i_k}$.
\end{proof}

An \emph{induced path} is a path that is also an induced subgraph.
In particular, if we have an induced path $e_1,e_2,\dots,e_\ell$,
then for any $i<j$ such that $\abs{i-j}\ge 2$, $e_i$ and $e_j$ are not adjacent.
The following lemma follows from taking the shortest path among all paths guaranteed in \Cref{lemma-exist-path}.

\begin{corollary}
  \label{corollary-exist-shortest-path}
  For any variable $v \in V \setminus \{v_0\}$, if $v \in V_1$, then there must exist a sequence of
  hyperedges $e_1,e_2,\ldots, e_{\ell}$ for some $\ell \geq 1$ such that the following properties
  hold
  \begin{itemize}
    \item $e_1 \in N^2_{v_0}$ and $v \in e_{\ell}$;
    \item for all $1 \leq i \leq \ell $, $e_i$ is failed;
    \item $e_1,e_2,\ldots, e_{\ell}$ is an induced path in $\Lin^2(H)$.
  \end{itemize}
\end{corollary}
%

We are now ready to prove Lemma~\ref{lemma-Ex-V1}, namely
\begin{align*}
  \E[\+C]{|V_1|} \leq \frac{3}{2}.	
\end{align*}

Fix any induced path (IP)  $e_1,e_2,\ldots, e_{\ell}$ in $\Lin^2(H)$. 
We bound the probability that all hyperedges in this path are failed hyperedges. Obliviously,
\begin{align}
  \label{eq-disjoint-edge}
  \Pr[\+C]{\forall 1\leq i \leq \ell, e_i \text{ is failed}} \leq \Pr[\+C]{\forall 1\leq j \leq \lceil\ell / 2\rceil, e_{2j-1} \text{ is failed}}. 
\end{align}
To bound the RHS of~\eqref{eq-disjoint-edge}, we define the set of disjoint hyperedges 
\begin{align}
  \label{eq-def-set-D}
  {\+D} \triangleq \{e_{2j-1} \mid 1\leq j \leq \lceil\ell / 2\rceil \}.
\end{align}
Because this is an induced path in $\Lin^2(H)$,
for any $e, e' \in \+D$, it holds that $e \cap e' = \emptyset$.
However, because of the subtlety of the adaptive coupling procedure $\+C$,
we cannot claim that the events of $e$ being failed are independent from each other for $e\in\+D$ based on this disjointness alone.
Instead, we will implement the coupling procedure $\+C$ in a slightly different way.

For each hyperedge $e \in \+D$, 
we define two sequences of random numbers: $\+R_{e,1}$ of length $k - k_{\beta}$ and $\+R_{e,2}$ of length $k_{\gamma}$,
where 
\begin{itemize}
  \item for each $1 \leq i \leq k - k_{\beta} $, $\+R_{e,1}(i) \in [0,1]$ is a uniform and independent real number;
  \item for each $1 \leq i \leq k _{\gamma}$, $\+R_{e,2}(i) \in [0,1]$ is a uniform and independent real number;
\end{itemize}
Suppose each hyperedge $e \in \+D$ maintains two indices $i_{e,1}$ and $i_{e,2}$. 
Initially, $i_{e,1}= i_{e,2} = 1$. 
We run the coupling procedure $\+C$ with the following modification. 
For each round of the  while-loop in $\+C$, 
if the vertex $u$ picked in Line~\ref{line-find-u-C} satisfies $u \in e$ for some $e \in \+D$ (such $e$ is unique because all hyperedges in $\+D$ are disjoint), 
then we modify Line~\ref{line-sample-r} as follows:
\begin{itemize}
  \item if $u \in \+M$, let $r_u = \+R_{e,1}(i_{e,1})$, and let $i_{e,1}\gets i_{e,1}+1$;
  \item if $u \not\in \+M$, let $r_u = \+R_{e,2}(i_{e,2})$ if the literal $u$ appears in the clause represented by $e$; 
    let $r_u = 1 - \+R_{e,2}(i_{e,2})$ if the literal $\neg u$ appears in the clause represented by $e$, and let $i_{e,2} \gets i_{e,2} +1$.
\end{itemize}
Note that all numbers in $\+R_{e, 1}$ and $\+R_{e,2}$ are uniformly distributed over $[0,1]$.
In the modification above, each $r_u$ is either $r$ or $1 - r$ for some $r \in \+R_{e, 1} \cup \+R_{e,2}$.
Hence, each $r_u$ is uniformly distributed over $[0,1]$. 
For any $e \in \+D$,
it contains at most $k - k_{\beta}$ marked variables, 
and there are at most $k_{\gamma}$ unmarked variables $u \in e$ that need to sample $r_u$ in $\+C$. 
Hence, the two sequences $\+R_{e,1}$ and $\+R_{e,2}$ will not exhaust during the coupling procedure $\+C$.
As a result, the modification above will not affect the execution and the outcome of $\+C$.

For each $e \in \+D$, we say the  event $\+A_e$ occurs if one of the following two events occurs:
\begin{itemize}
\item there exists a random number $r$ in $\+R_{e,1}\cup \+R_{e,2}$ such that $\plow< r \leq \pup$;
\item for all $1 \leq i \leq k _{\gamma}$, $0 \leq \+R_{e,2}(i) \leq \pup $. 
\end{itemize}

Then, we have the following claim.
\begin{claim}
\label{claim-upper-bound}
For each hyperedge $e \in \+D$, if $e$ is a failed hyperedge after the coupling procedure $\+C$, then the event $\+A_e$ must occur.
\end{claim}
\begin{proof}
Fix a hyperedge $e \in \+D$. After the coupling procedure $\+C$, for all $u \in {e \cap \Vcol} \setminus \{v_0\}$, the random number $r_u$ comes from 	$\+R_{e,1}\cup\+R_{e,2}$. 
Suppose $e$ is a failed hyperedge after the coupling procedure $\+C$, by Definition~\ref{definition-failed-edge}, here are two cases.

Reason~\ref{reason-1}: there exists $v \in (e \cap \Vcol) \setminus \{v_0\} $ such that $\plow < r_v \leq \pup$, then there must exist a random number $r$ in $\+R_{e,1}\cup \+R_{e,2}$ such that $\plow< r \leq \pup$;

Reason~\ref{reason-2}: $|e \cap (\Vcol \setminus \+M) | = k_{\gamma}$ and $e$ is not satisfied by both $X(\mathcal{S})$ and $Y(\mathcal{S})$. Let $c_e$ denote the clause represented by $e$. We list all variables in $u_1,u_2,\ldots,u_{k_{\gamma}}$ in $|e \cap (\Vcol \setminus \+M) |$ such that $u_i$ is the $i$-th variable processed by the while-loop in $\+C$. Fix $1\leq i \leq k_{\gamma} $. 
\begin{itemize}
\item Suppose the literal $u_i$ appears in $c_e$. If $r_{u_i} > \pup$, then by Lemma~\ref{lemma-px-py-C} and Line~\ref{line-set-S}, we have $X^{\+C}(u_i) = Y^{\+C}(u_i) = 1$ and $u_i \in \+S$. In this case $c_e$ is satisfied by $X(\mathcal{S})$ and $Y(\mathcal{S})$, and the event in Reason~\ref{reason-2} cannot occur. So we must have $r_{u_i} \leq \pup$. Since $r_{u_i} = \+R_{e,2}(i)$, we have $\+R_{e,2}(i) \leq \pup$. 
\item Suppose the literal $\neg u_i$ appears in $c_e$. If $r_{u_i} \leq \plow$, then by Lemma~\ref{lemma-px-py-C} and Line~\ref{line-set-S}, we have $X^{\+C}(u_i) = Y^{\+C}(u_i) = 0$ and $u_i \in \+S$. In this case $c_e$ is satisfied by $X(\mathcal{S})$ and $Y(\mathcal{S})$, and the event in Reason~\ref{reason-2} cannot occur. So we must have $r_{u_i} > \plow$. Since $r_{u_i} = 1 - \+R_{e,2}(i)$, we have $\+R_{e,2}(i) < 1 - \plow = \pup$. 
\end{itemize}
This implies for all $1 \leq i \leq k _{\gamma}$, $0 \leq \+R_{e,2}(i) \leq \pup $.
\end{proof}

For each $e \in \+D$, all reals numbers in $\+R_{e,1}$ and $\+R_{e,2}$ are sampled uniformly and independently. We use $\+R_e$ to denote this product distribution. And we use $\+R$ to denote the product distribution of all $\+R_e$ for $e \in \+D$.
By the definition of $\+D$ in~\eqref{eq-def-set-D},  we can bound the RHS of~\eqref{eq-disjoint-edge} as 
\begin{align*}
\Pr[\+C]{\forall 1\leq i \leq \ell, e_i \text{ is failed}} &\leq \Pr[\+C]{\forall e \in \+D, e \text{ is failed} }\\
\tag{by \Cref{claim-upper-bound}}\quad&\leq  \Pr[\+R]{ \bigwedge_{e \in \+D}\+A_e}\\ 
 &= \prod_{e \in \+D}\Pr[\+R_e]{\+A_e}\\
\tag{by the definition of $\+A_e$}\quad &\leq\prod_{e \in \+D}\left( \frac{2k}{s} + \left( \frac{1}{2} + \frac{1}{s} \right)^{k_\gamma}\right).
\end{align*}
We define $\pfailed$ as
\begin{align}
  \label{eq-failed-edge}
 \pfailed \triangleq \frac{2k}{s} + \left(\frac{1}{2} + \frac{1}{s}\right)^{k_\gamma}.
\end{align}
Note that $|\+D| \geq \ell / 2$.
Thus, for any induced path (IP)  $e_1,e_2,\ldots, e_{\ell}$, we have
\begin{align*}
\Pr[\+C]{\forall 1\leq i \leq \ell, e_i \text{ is failed}} \leq \pfailed^{\ell / 2}.	
\end{align*}

By Corollary~\ref{corollary-exist-shortest-path}, we have for any vertex $v \neq v_0$,
\begin{align}
  \label{eq-bound-v-in-V1}
  \Pr[\+C]{v \in V_1} &\leq \sum_{\substack{\text{IP } e_1,e_2,\ldots,e_{\ell} \text{ in } L^2 \\ \text{satisfying } e_1 \in N^2_{v_0}, v \in e_{\ell} }}\Pr[\+C]{\forall 1\leq i \leq \ell, e_i \text{ is failed}}\notag\\
  &\leq \sum_{\substack{\text{IP } e_1,e_2,\ldots,e_{\ell} \text{ in } L^2 \\ \text{satisfying } e_1 \in N^2_{v_0}, v \in e_{\ell} }}\pfailed^{\ell / 2}.
\end{align}
Note that $v_0 \in V_1$, then we have
\begin{align*}
  \E[\+C]{|V_1|} - 1  &= \sum_{v \in V \setminus \{v_0\} }\Pr[\+C]{v \in V_1}\\
  \tag*{(by~\eqref{eq-bound-v-in-V1})}\quad &\leq \sum_{v \in V \setminus \{v_0\} }\sum_{\substack{\text{IP } e_1,e_2,\ldots,e_{\ell} \text{ in } L^2 \\ \text{satisfying } e_1 \in N^2_{v_0}, v \in e_{\ell} }}\pfailed^{\ell / 2}\\
  \quad&\leq \sum_{\substack{\text{IP } e_1,e_2,\ldots,e_{\ell} \text{ in } L^2 \\ \text{satisfying } e_1 \in N^2_{v_0} }}k\cdot\pfailed^{\ell / 2},
\end{align*}
where in the last inequality, we enumerate all the IPs starting from $N^2_{v_0}$ and use the fact that each hyperedge contains $k$ vertices.
Note that the maximum degree of $\Lin^2(H)$ is at most $d^2k^2$ and there are at most $d^2k$ hyperedges in
set $N^2_{v_0}$.
Thus, we have
\begin{align*}
  \E[\+C]{|V_1|} - 1 \leq \sum_{\ell = 1}^{\infty} d^2k \cdot (d^2k^2)^{\ell - 1} \cdot k \cdot  \pfailed^{\ell / 2} = \sum_{\ell = 1}^{\infty} (d^2k^2)^{\ell} \cdot  \pfailed^{\ell / 2}=  \sum_{\ell = 1}^{\infty} c^\ell = \frac{c}{1-c},
\end{align*}
where $c \defeq d^2k^2\sqrt{\pfailed}$.	
Hence, to prove $\E{|V_1|} \leq \frac{3}{2}$, it is sufficient to prove that
\begin{align*}
  c = d^2k^2\sqrt{\pfailed} \leq \frac{1}{3},
\end{align*}
which, in turn, is implied by
\begin{align*}
  \pfailed \leq \frac{1}{9d^4k^4}.
\end{align*}
Recall that $\pfailed$ is defined in~\eqref{eq-failed-edge} and $s = 36d^4k^5$. We have
\begin{align*}
  \pfailed &= \frac{2k}{s} + \left(\frac{1}{2} + \frac{1}{s}\right)^{k_\gamma}\\
        &\leq \frac{1}{18d^4k^4} + \left( \frac{1}{2} \right)^{k_{\gamma}}\exp\left(\frac{2k_{\gamma}}{s} \right)\\
        \tag{by $k_\gamma \leq k$}\quad&\leq \frac{1}{18d^4k^4} + \left( \frac{1}{2} \right)^{k_{\gamma} - 1}.
\end{align*}
Since $2^{k_{\gamma}} \geq 36d^4k^4$, 
we have that $\pfailed \leq \frac{1}{9d^4k^4}$.

\section{Analyze the rejection sampling subroutine}\label{section-proof-main}
In this section, we will analyze the $\sample$ subroutine (Algorithm~\ref{alg-sample}).

Let $\Lambda\subseteq\+M$ be a subset of marked variables, $\eps>0$, $X\in \set{0,1}^{\Lambda}$ and $S\subseteq V\setminus\Lambda$. 
We continue to use the same notations as in Section~\ref{section-sampleunmark}. 
Let $\Phi^X=(V^X,C^X)$ be the formula obtained from $\Phi$ simplified under $X$,
and $\Phi^X=\Phi_1^X\land\Phi_2^X\land\dots\land\Phi_\ell^X$ where all $\Phi^X_i=(V_i^X,C_i^X)$ are disjoint.
For every $i\in[m]$, $H_i^X=(V_i^X,\+E_i^X)$, the hypergraph representation of $\Phi_i^X$, is connected. 
Assume without loss of generality that $V_i\cap S\ne\varnothing$ for $1\le i\le m$ and $V_i\cap S=\varnothing$ for $m<i\le\ell$.

\begin{lemma}\label{lemma-sample-correctness}
  For any $0< \eta < 1$, the time complexity of $\sample(\Phi,\delta,X,S)$ is  $O\left(|S| \left(\frac{n}{\delta}\right)^{\frac{\eta}{10}} d^2k^3 \log^2 \frac{n}{\delta}\right)$.
  Furthermore, if  $2^{k_\beta} \geq \frac{20}{\eta}\mathrm{e}dk$ and $\abs{\+E_i^X}\le \toolarge$ for every $1\le i\le m$, then
  $\sample(\Phi,\delta,X,S)$ returns a random assignment $Y\in\set{0,1}^S$ satisfying
  $d_{\-{TV}}\tp{Y,\mu_S\tp{\cdot\mid X}}\le \delta$.
\end{lemma}
\begin{proof}
We first analyze the running time of $\sample(\Phi,\delta,X,S)$. We need to find all the connected components $\set{H^X_i = (V_i^X,\mathcal{E}_i^X)\mid 1\leq i \leq m}$ in
  $H_{\Phi^X}$ such that each $V^X_i \cap S \neq \emptyset$ and check whether there exists $ 1 \leq i \leq m$ such that $|\mathcal{E}^X_i| > \toolarge$. Suppose we store the hypergraph $H_{\Phi}$ as an adjacent list. For each vertex $v \in S$, we apply the deep first search starting from $v$ in $H_{\Phi}$. When visiting each hyperedge $e$, we can check whether $e$ is in $H_{\Phi^X}$. Once we find that one connected component in $H_{\Phi^X}$ contains more than $\toolarge$ hyperedges,  we stop this process immediately. The time complexity of the deep first search step is at most
  \begin{align*}
  T_{\mathsf{DFS}} = O\left(|S| d^2k^3 \log \frac{n}{\delta}	\right).
  \end{align*}
  If $|\mathcal{E}^X_i| \leq \toolarge$ for all $1\leq i\leq m$,
  then we apply the rejection sampling for each $\Phi_i^X$. Note that $m \leq |S|$. The time complexity of the rejection sampling step is at most
  \begin{align*}
    T_{\mathsf{RS}} = O\left(|S|R dk^2 \log \frac{n}{\delta}	\right) = O\left(|S| \left(\frac{n}{\delta}\right)^{\frac{\eta}{10}} dk^2 \log^2 \frac{n}{\delta}\right).	
  \end{align*}
The overall time complexity for the subroutine $\sample(\Phi,\delta,X,S)$ is at most
\begin{align*}
T_{\mathsf{S}} =  T_{\mathsf{DFS}} + T_{\mathsf{RS}} = O\left(|S| \left(\frac{n}{\delta}\right)^{\frac{\eta}{10}} d^2k^3 \log^2 \frac{n}{\delta}\right).
\end{align*}

We next analyze the total variation distance between $Y$ and $\mu_S(\cdot \mid X)$.
  Since $\abs{\+E_i^X}\le \toolarge$ for every $1\le i\le m$, the random
  assignment $Y$ is returned in either Line~\ref{line-bad-return-2} or Line~\ref{line-good}. It
  follows from Proposition~\ref{proposition-good-event} that we only need to show the probability
  that $Y$ is returned in Line~\ref{line-bad-return-2} is at most $\delta$, which is equivalent to
  that one of the $\Rejection(\Phi_i^X,R)$ returns $\perp$ among all $1\le i\le m$.

%

%


  Fix $1\leq i \leq m$. Consider the rejection sampling for the instance $\Phi_i^X$. Let
  $\Pr[\+P]{\cdot}$ be the product distribution such that each variable in $C_i^X$ takes a value from
  $\{0,1\}$ uniformly and independently. For each clause $c \in C_i^X$, let $B_c$ denote the event
  that $c$ is not satisfied. Define
  \begin{align*}
    \Gamma(B_c) = \{B_{b} \mid b \in C_i^X \land b \neq c \land \vbl{c} \cap \vbl{b} \neq \emptyset \}.
  \end{align*}
  Suppose $2^{k_\beta} \geq \frac{20}{\eta}\mathrm{e}dk$ for some $0< \eta < 1$. For each $c \in C_i'$, let $x(B_c) \defeq \frac{\eta}{20dk}$. 
  Since every clause has at least $k_\beta$ unmarked vertices, we have that
  \begin{align*}
    \Pr[\+P]{B_c} \leq \left( \frac{1}{2} \right)^{k_\beta} \leq x(B_c) \prod_{B \in \Gamma(B_c)}\left(1 - x(B)\right).	
  \end{align*}
  By the Lov\'asz local lemma in Theorem~\ref{theorem-LLL}, we have
  \begin{align*}
    \Pr[\+P]{ \bigwedge_{c \in C^X_i}\overline{B_{c}} } \geq \prod_{c \in C^X_i}(1 - x(B_c)) = \prod_{c \in C^X_i}\left( 1 - \frac{\eta}{20dk} \right).
  \end{align*}
  Since $\abs{\+E^X_i}\le \toolarge$, we have
  \begin{align*}
    \Pr[\+P]{ \bigwedge_{c \in C^X_i}\overline{B_{c}} }
    \geq \left( 1 - \frac{\eta}{20dk} \right)^{\toolarge}
    \geq \left( 1 - \frac{1}{\frac{15}{\eta}dk + 1} \right)^{\toolarge}
    \geq \exp\left(-\frac{\eta}{15}\log \frac{n}{\delta} \right)
    >  \left(\frac{\delta}{n} \right)^{\frac{\eta}{10}}.
  \end{align*}
  For each $\Phi^X_i$, our algorithm repeats the rejection sampling for $\repeattime$ times. Hence,
  the probability that the rejection for $\Phi^X_i$ fails is at most
  \begin{align*}
    \left( 1 -  \left(\frac{\delta}{n} \right)^{\frac{\eta}{10}} \right)^{\repeattime} \leq \frac{\delta}{n}.
  \end{align*}
  Note that $m$ is at most $n$. Taking a union bound over all $\Phi^X_i$ for $1\leq i \leq m$, 
  we have that if the conditions of the lemma holds, then
  \begin{align*}
    \Pr{\exists i\in[m], \Rejection\tp{\Phi_i^X,R}=\perp} & \leq \delta. \qedhere
  \end{align*}
\end{proof}

We now proceed to show that, in all calls to $\sample(\Phi,\delta,X,S)$ during the execution of
Algorithm~\ref{alg-mcmc}, $\abs{\+E_i^X}\le \toolarge$ for every $i\in[\ell]$ with high
probability.

Algorithm~\ref{alg-mcmc} calls the subroutine $\sample$ for $T+1$ times ($T$ times in
Line~\ref{line-sample-1} and once in Line~\ref{line-sample-2}). For each $1\le t\le T+1$, we use the
$\+B_t$ to denote the event that $\abs{\+E_i^X}>\toolarge$ for some $1\le i\le \ell$ at the
$t$-th call to $\sample(\cdot)$.
Note that, in all calls to $\sample(\Phi,\delta,X,S)$ during the execution of
Algorithm~\ref{alg-mcmc}, the parameter $\delta$ is always set to $\frac{\eps}{4(T+1)}$. 
The following lemma bounds the probability of each $\+B_t$.

\begin{lemma}
  \label{lemma-too-large}
  Assume $2^{k_\alpha} \geq 4\mathrm{e}^2d^2k^2$ and $2^{k_\beta} \geq 2\mathrm{e}dk$.  For each
  $1\leq t \leq T+1$, it holds that in the execution of \Cref{alg-mcmc},
  $\Pr{\mathcal{B}_t} \leq \delta$, where $\delta = \frac{\epsilon}{4(T+1)}$ and $T = \Tmix$.
  \end{lemma}

The rest of this section is devoted to the proof of Lemma~\ref{lemma-too-large}.

%
Recall that $(X_t)_{t=0}^T$ is the random process defined by Algorithm~\ref{alg-mcmc}.
Fix $1\leq t \leq T+1$. Consider the $t$-th call of the subroutine
$\sample(\Phi,\delta, X, S)$ (Algorithm~\ref{alg-sample}).
If $1\leq t \leq T$, let $v \in \mathcal{M}$ denote the random vertex picked in the $t$-th step.
The random assignment $X$ and the subset $S$ in the subroutine $\sample(\Phi,\delta, X, S)$ are
defined as
\begin{align}
  X &= \begin{cases}
    X_{t-1}({\+M} \setminus \{v\} ) \text{\ \ (namely $\Lambda=\+M\setminus\{v\}$)} &\text{if } 1 \leq t \leq T, \\
    X_T  \text{\ \ (namely $\Lambda=\+M$)} &\text{if } t = T+1,
  \end{cases} \label{eq-def-X}\\
  S &= \begin{cases}
    \{v\} &\text{if } 1 \leq t \leq T, \\
    V \setminus \mathcal{M} &\text{if } t = T+1.
  \end{cases}\label{eq-def-S}
\end{align}

Consider the hypergraph $H_{\Phi} = (V, \mathcal{E})$ as defined in~\eqref{eq-def-Hphi}.
Given an assignment $X \in \{0,1\}^{{\+M}}$, we say a hyperedge $e \in \mathcal{E}$ in $H_{\Phi}$ is \emph{bad} if the clause represented by $e$ is not satisfied by $X$. 
Recall that we use $\Phi^X$ to denote the CNF formula obtained from $\Phi$ simplified under $X$ and use
$H_{\Phi^X}=(V,\+E^X)$ to denote its hypergraph representation.
Hence $\mathcal{E}^X \subseteq \mathcal{E}$ is the set of all bad hyperedges.
If the bad event $\mathcal{B}_t$ occurs, there must exist a connected component in $H_{\Phi^X}$
containing more than $\toolarge$ bad hyperedges.

Fix a hyperedge $e \in \mathcal{E}$, let $\mathcal{B}_e$ be the event that 
\begin{itemize}
\item the hyperedge $e$ is in $\+E^X$;
\item $|\mathcal{E}_e| \geq \toolarge$,
  where $H_e =(V_e, \mathcal{E}_e)$ is the connected component in $H_{{\Phi}^X}$ such that $e \in \mathcal{E}_e$.
\end{itemize}
By the definition of $\mathcal{B}_e$, if the event $\mathcal{B}_t$ occurs, then there
must exist $e \in \mathcal{E}$ such that the event $\mathcal{B}_e$ occurs. We have
\begin{align}
  \label{eq-prob-Be}
  \Pr{\mathcal{B}_t } \leq \Pr{\exists\,e \in \mathcal{E} \text{ s.t. } \mathcal{B}_e } \leq \sum_{e \in \mathcal{E}} \Pr{ \mathcal{B}_e }.	
\end{align}

Next we bound the probability of $\mathcal{B}_e$. 
We first establish local uniformity of any intermediate assignment $X_t$.
\begin{lemma}
  \label{lemma-local-uniform}
  Suppose the CNF formula $\Phi$ satisfies $2^{k_\beta} \geq 2\mathrm{e}ds$ for some $s \geq k$.  Let
  $X \subseteq \{0,1\}^{\Lambda}$ be the random assignment defined in~\eqref{eq-def-X}, where
  $\Lambda ={\+M}$ or $\+M\setminus \{v\}$ for some $v$. For any subset $S \subseteq \Lambda$ and any assignment
  $\sigma \in \{0,1\}^S$, it holds that
  \begin{align*}
    \Pr{ X(S) = \sigma } \leq 	\left( \frac{1}{2} \right)^{|S|}\exp\left( \frac{|S|}{s} \right).
  \end{align*}
\end{lemma}
\begin{proof}
  By the definition of the $X$ in~\eqref{eq-def-X}, we know that $X= X_t(\Lambda)$ for some
  $0\leq t \leq T$.
  For each vertex $v \in S$, we define $t_v \leq t$ as follows.
  If $v$ is chosen by the Algorithm~\ref{alg-mcmc} at least once, 
  then let $t_v$ be the largest $t' \leq t$ such that $v$ is chosen at the $t'$-th step.  
  Otherwise, let $t_v = 0$.

  We sort all the vertices in $S$ according to $t_v$. If two vertices $u,v \in S$ satisfy
  $t_u = t_v =0$, we break the tie arbitrarily.  
  Let $v_1,v_2,\ldots,v_{\ell}$ be the set of
  all vertices in $S$ such that
  \begin{align*}
    0\leq t_{v_1} \leq t_{v_1} \leq \ldots \leq t_{v_\ell} \leq T.
  \end{align*}
  Thus, we have
  \begin{align*}
    \forall 1\leq i \leq \ell:\quad X(v_i) = X_{t_{v_i}}(v_i).
  \end{align*}
  Consider the $t_{v_i}$-th step. The value $X_{t'}(v_i)$ is generated by
  $\sample(\Phi, \frac{\epsilon}{4(T+1)}, X_{t'-1}({\+M} \setminus \{v_i\} ), \{v_i\})$, where $t' = t_{v_i}$.
  Suppose $2^{k_\beta} \geq 2\mathrm{e}ds$ for some $s \geq k$.  
  We claim that for any $v \in {\+M}$, any $X' \in \{0,1\}^{{\+M} \setminus \{v\}} $ and any $0<\delta <1$, it holds that
  \begin{align}
    \label{eq-subroutine-uniform}
    \forall c\in\{0,1\}, \quad \Pr{ \sample\left(\Phi,\delta, X' , \{v\}\right) \text{ returns } c} \leq 	\frac{1}{2}\exp\left( \frac{1}{s} \right).
  \end{align}
  Assume~inequality~\eqref{eq-subroutine-uniform} holds. Note that $|S| = \ell$. By the chain rule,
  we have
  \begin{align*}
    \Pr{X(S) = \sigma }
    &= \prod_{i = 1}^{\ell} \Pr{X(v_i) = \sigma(v_i) \mid \forall 1\leq j < i, X(v_j) = \sigma(v_j)   }\\
    &=\prod_{i = 1}^{\ell} \Pr{X_{t_{v_i}}(v_i) = \sigma(v_i) \mid \forall 1\leq j < i, X_{t_{v_j}}(v_j) = \sigma(v_j)  }\\
    &\leq \left( \frac{1}{2} \right)^{|S|}\exp\left( \frac{|S|}{s} \right),
  \end{align*}
  where the last inequality holds due to~\eqref{eq-subroutine-uniform} and the fact that the initial random
  assignment $X_0$ is sampled from $\{0,1\}^{{\+M}}$ uniformly at random.
 
  We now prove the inequality~\eqref{eq-subroutine-uniform}. 
  By Algorithm~\ref{alg-sample} and
  Proposition~\ref{proposition-good-event}, we know that the random value $c$ returned by the
  subroutine $\sample(\Phi,\delta, X', \{v\})$ is either sampled from $\{0,1\}$ uniformly at
  random or sampled independently from the distribution $\mu_{v}(\cdot \mid X' )$. If $c$ is
  sampled from $\{0,1\}$ uniformly at random, then~\eqref{eq-subroutine-uniform} holds trivially. We
  now prove that
  \begin{align}
    \label{eq-upbound-conditional-marginal}
    \forall c \in \{0,1\}, \quad  \mu_v(c \mid X')	\leq \frac{1}{2} \exp\left( \frac{1}{s} \right).
  \end{align}
  Recall $X'  \in \{0,1\}^{{\+M} \setminus \{v\}}$.
  Let $\Phi'\defeq\Phi^{X'}$ be the CNF formula obtained from $\Phi$ by deleting all the clauses satisfied by $X'$ and all the variables in ${\+M} \setminus \{v\}$,
  and $\mu'\defeq\mu^{X'}$ be the uniform distribution of all solutions in $\Phi'$.
  Then the two distributions $\mu'_v(\cdot)$ and $\mu_v(\cdot \mid X')$ are identical.
  By Condition~\ref{condition-marked-variables}, we have each clause in $\Phi'$ contains at least
  $k_\beta$ variables and at most $k$ variables. Each variable belongs to at most $d$ clauses.
  Since $2^{k_{\beta}} \geq 2\mathrm{e}ds$ for some $s \geq k$,
  inequality~\eqref{eq-upbound-conditional-marginal} follows from Corollary~\ref{corollary-local-uniform}.
\end{proof}

To bound the size of connected components including a particular hyperedge $e$,
recall that $\Lin(H)$ is the line graph of $H$ defined in \Cref{definition-line-graph}. 
We also need the notion of $2$-trees.

\begin{definition}[2-tree]
  Let $G=(V, E)$ be a graph. A set of vertices $T \subseteq V$ is called a 2-tree if (1) for any
  $u,v \in T$, $\dist_G(u,v) \geq 2$; (2) if one adds an edge between every $u,v \in T$ such that
  $\dist_G(u,v) = 2$, then $T$ is connected.
\end{definition}

The following simple observation follows directly from the definition of $2$-trees.
\begin{observation}
  \label{observation-down}
  If a graph $G=(V, E)$ has a 2-tree of size $\ell > 1$ containing the vertex $v \in V$, then
  $G$ must have a 2-tree of size $\ell - 1$ containing the vertex $v$.
\end{observation}
\begin{proof}
  Let $T \subseteq V$ be a 2-tree in $G$. Let $G'=(T, E_T)$, where each $\{u, v\} \in E_T$ if and
  only if $u,v\in T$ and $\dist_G(u,v) = 2$.
  Then $G'$ is a connected graph.
  We can find an arbitrary spanning tree $T_{G'}$ of graph $G'$.
  Since the number of vertices in $T_{G'}$ is $\ell > 1$, then $T_{G'}$ contains at least two leaf
  vertices.
  Let $w$ be the leaf vertex in $T_{G'}$ such that $w \neq v$.
  It is easy to see $T \setminus \{w\}$ is a 2-tree of size $\ell - 1$ containing the vertex
  $v$.
\end{proof}

To bound the number of $2$-trees, we need the following lemma in~\cite{borgs2013left} to bound the
number of connected subgraphs.
\begin{lemma}
  \label{lemma-number-of-component}
  Let $G=(V, E)$ be a graph with maximum degree $\Delta$ and $v \in V$ be a vertex. Then the number
  of connected induced subgraphs of size $\ell$ containing $v$ is at most
  $\frac{(\mathrm{e}\Delta)^{\ell - 1}}{2}$.
\end{lemma}
\begin{corollary}
  \label{corollary-number-2-tree}
  Let $G=(V, E)$ be a graph with maximum degree $\Delta$ and $v \in V$ be a vertex. Then the number
  of 2-trees in $G$ of size $\ell$ containing $v$ is at most
  $\frac{(\mathrm{e}\Delta^2)^{\ell - 1}}{2}$.
\end{corollary}
\begin{proof}
  Consider the power graph $G^2$. The maximum degree of $G^2$ is at most $\Delta^2$. The number of
  connected induced subgraphs in $G^2$ of size $\ell$ containing vertex $v$ is at most
  $\frac{(e\Delta^2)^{\ell - 1}}{2}$. This is an upper bound of the number of 2-trees in $G$ of size
  $\ell$ containing $v$.
\end{proof}

\begin{lemma}
  \label{lemma-lower-bound-size}
  Let $H=(V, \mathcal{E})$ be a $k$-uniform hypergraph such that each vertex belongs to at most $d$
  hyperedges.  Let $B \subseteq \mathcal{E}$ be a subset of hyperedges which induces a connected
  subgraph in $\Lin(H)$, and $e \in B$ be an arbitrary hyperedge.
  Then, there must exist a $2$-tree $T \subseteq B$ in the graph $\Lin(H)$ such that $e \in T$ and
  $|T| =\left\lfloor \frac{|B|}{kd} \right\rfloor $.
\end{lemma}
\begin{proof}
  Consider the graph $\Lin(H)=(V_L, E_L)$. For any subset of vertices $S$ in $\Lin(H)$, 
  let the extended neighbourhood of $S$ be
  \begin{align*}
    \Gamma^+(S) \triangleq \{v \in V_L \mid v \in S \text{ or there exists } u \in S \text{ s.t. } \{u,v\} \in E_L \}.	
  \end{align*}
  We construct a $2$-tree greedily. Let $T_0 = \{e\}$. For the $i$-th step, we set
  $S \gets B \setminus \Gamma^+(T_{i-1})$, let $e_i$ be the first hyperedge in $S$ such that $\dist_{\Lin(H)}(T_{i-1}, e_i) = 2$,
  and set $T_{i} = T_{i-1} \cup \{e_i\}$. 
  The process ends when $B=\Gamma^+(T_j)$ for some $j$.

  We claim that the set $S$ will become empty eventually. 
  Suppose the current $2$-tree is $T$,
  and some non-empty $S = B\setminus \Gamma^+(T)$ remains.
  Thus, $\forall e' \in S$, $\dist_{\Lin(H)}(T, e') \neq 2$.
  Note that if $\dist_{\Lin(H)}(T, e') \le 1$, $e'\in \Gamma^+(T)$.
  Thus, $\forall e' \in S$, $\dist_{\Lin(H)}(T, e') \ge 3$.
  Note that $B \subseteq \Gamma^+(T) \cup S$, $B \cap \Gamma^+(T) \neq \emptyset$ and $B \cap S \neq \emptyset$.
  Hence $B$ is disconnected in $\Lin(H)$. Contradiction.

  In every step, at most $kd$ hyperedges are removed, so we have $|T| \geq \left\lfloor \frac{|B|}{kd} \right\rfloor$.  Then by
  Observation~\ref{observation-down}, there must exist a $2$-tree $T \subseteq B$ in graph $\Lin(H)$
  such that $e \in T$ and $|T| = \left\lfloor \frac{|B|}{kd} \right\rfloor$.
\end{proof}

We are now ready to prove Lemma~\ref{lemma-too-large}.
\begin{proof}[Proof of \Cref{lemma-too-large}]
  We bound the probability $\mathcal{B}_e$ in~\eqref{eq-prob-Be}. 
  If the event $\mathcal{B}_e$ occurs, there must exist a subset set $B \subseteq \mathcal{E}$ such that $e \in B$,
  $|B| =L \triangleq \left\lceil\toolarge \right\rceil$,
  $B$ is connected in $\Lin(H)$,
  and all hyperedges in $B$ are bad hyperedges, i.e.\ all hyperedges in $B$ are not satisfied by $X$. 
  Let $\ell\defeq  \lfloor\frac{L}{kd}\rfloor$.
  By \ref{lemma-lower-bound-size}, there must exists a 2-tree in $T \subseteq B$ such that $e\in T$ and $|T| = \ell$.
  
  By the definition of $X \in \{0,1\}^\Lambda$ in~\eqref{eq-def-X} and
  Condition~\ref{condition-marked-variables}, we have $|e \cap \Lambda |\geq k_\alpha - 1$ for all
  $e \in \mathcal{E}$.
  Note that all hyperedges in $T$ are disjoint.  %
  By assumption $2^{k_\beta} \geq 2\mathrm{e}dk$. 
  We then use Lemma~\ref{lemma-local-uniform} with $s = k$. 
  This gives us the following
  \begin{align*}
    \Pr{\text{all hyperedges in $T$ are bad}} \leq 	\left(\frac{1}{2} \right)^{(k_\alpha - 1)\ell} \exp \left( \frac{(k_\alpha - 1)\ell}{k} \right).
  \end{align*}
  Note that the maximum degree of the graph $\Lin(H)$ is at most $dk$. 
  By \Cref{corollary-number-2-tree} and a union bound over all
  2-trees of size $\ell$ containing the hyperedge $e$, we have
  \begin{align*}
    \Pr{ \mathcal{B}_e }
    &\leq 	\frac{(\mathrm{e}d^2k^2)^{\ell - 1}}{2} \cdot 	\left(\frac{1}{2} \right)^{(k_\alpha - 1)\ell} \cdot \exp \left( \frac{(k_\alpha - 1)\ell}{k} \right)\\
    &\leq  \frac{1}{2\mathrm{e}d^2k^2} \left( \frac{2 \mathrm{e}^2 d^2k^2 }{2^{k_\alpha}} \right)^\ell,
  \end{align*}
  where the last inequality holds because $k_\alpha -1\leq k$. 
  By assumption $2^{k_\alpha} \geq 4\mathrm{e}^2d^2k^2$,
  and thus for any $e\in\mathcal{E}$
  \begin{align*}
    \Pr{\mathcal{B}_e} \le d^{-1} 2^{-\ell-1}.
  \end{align*}
  By~\eqref{eq-prob-Be}, we have
  \begin{align*}
    \Pr{\mathcal{B}_t }  \leq \sum_{e \in \mathcal{E}} \Pr{ \mathcal{B}_e } \leq nd\cdot d^{-1} 2^{-\ell-1}= n2^{-\ell-1} \leq \delta,
  \end{align*}
  since $\ell = \lfloor L / (kd) \rfloor \geq \log \frac{n}{\delta} - 1$.
  \end{proof}

\section{Analyze the main sampling algorithm}
\label{section-sampling-analyze}
Now we can finish the analysis of the main sampling algorithm, \Cref{alg-mcmc}.
\begin{theorem}\label{theorem-sampling-main}
The following holds for all $\xi \geq 0$.
There is an algorithm such that given any $0<\eps<1$ and $(k,d)$-formula $\Phi$ with $n$ variables where $k \geq 20\log k + 20\log d + 60 + \xi$,
it outputs a random assignment $X$ of $\Phi$ satisfying
  $d_{\-{TV}}(X,\mu)\le \eps$, where $\mu$ is the uniform distribution of satisfying assignments of $\Phi$. 
  The algorithm terminates in time $O\left(n \left(\frac{n}{\epsilon}\right)^{\eta} d^2k^3 \log^3 \frac{n}{\epsilon} \right)$, where $\eta = \left(\frac{1}{2} \right)^{20+\xi/3}\left( \frac{1}{dk} \right)^9$.
\end{theorem}
The sampling result in Theorem~\ref{theorem-sample-simplified} is a corollary of Theorem~\ref{theorem-sampling-main}. We can set the parameter $\zeta$ in Theorem~\ref{theorem-sample-simplified} as $\zeta = \left(\frac{1}{2} \right)^{20+\xi/3}$. The running time of the sampling algorithm in Theorem~\ref{theorem-sampling-main} is
\begin{align*}
O\left(n \left(\frac{n}{\epsilon}\right)^{\zeta (dk)^{-9}} d^2k^3 \log^3 \frac{n}{\epsilon} \right) =\widetilde{O}\left(d^2k^3 n \left(\frac{n}{\epsilon}\right)^{\zeta}  \right).
\end{align*}

We first prove \Cref{lemma-sample}. Then we use Lemma~\ref{lemma-sample} to prove Theorem~\ref{theorem-sampling-main}.

\begin{proof}[Proof of \Cref{lemma-sample}]
%
%
  We first couple $X_T$ of \Cref{alg-mcmc} with the idealized Glauber dynamics $\Glauber$.
  At each step of the Markov chain,
  we couple the outcome of $\sample$ with the idealized chain optimally.
  Coupling errors comes from the event $\+B_t$ and the failure of rejection sampling.
  By \Cref{lemma-too-large},
  with probability at most $\delta=\frac{\eps}{4(T+1)}$,
  event $\+B_t$ happens.
  When $\+B_t$ does not happen, by \Cref{lemma-sample-correctness},
  the output of $\sample$ is within total variation distance $\delta$ from the desired output.
  By \Cref{prop:coupling},
  we can successfully couple it with the ideal output with probability at least $1-\delta$.
  Thus, $X_T$ of \Cref{alg-mcmc} can be coupled with the $T$-th step of $\Glauber$ with probability at least $1-2T\delta$.

  Consider a sample $X_{\mathsf{Glauber}}$ by first running $\Glauber$ for $T$ steps to get $X_T'\in\{0,1\}^\+M$,
  and then draw from $\mu_{V\setminus\+M}(\cdot\mid X_T')$.
  In Line~\ref{line-sample-2} of \Cref{alg-mcmc}, by \Cref{lemma-sample-correctness} and \Cref{lemma-too-large},
  $\sample$ returns a sample within TV distance $\delta$ from $\mu_{V\setminus\+M}(\cdot\mid X_T)$ with probability at least $1-\delta$.
  Thus by \Cref{prop:coupling} once again,
  \begin{align*}
    \DTV{\Xalg}{X_{\mathsf{Glauber}}} \le 2(T+1)\delta = \frac{\epsilon}{2}.
  \end{align*}
  Moreover, consider an optimal algorithm which first obtains a perfect sample $X_\+M$ from $\mu_{\+M}$,
  and then complete it to all $V$ by sampling from $\mu_{V\setminus\+M}(\cdot\mid X_\+M)$.
  Call this sample $X_{\mathsf{ideal}}$,
  and then the law of $X_{\mathsf{ideal}}$ is $\mu$.
  By \Cref{prop:coupling} and \Cref{lemma-mixing},
  \begin{align*}
    \DTV{X_{\mathsf{Glauber}}}{X_{\mathsf{ideal}}}\le\frac{\epsilon}{4}
  \end{align*}
  Combining everything we have that
  \begin{align*}
    \DTV{\Xalg}{\mu}& = \DTV{\Xalg}{X_{\mathsf{ideal}}} \le \DTV{\Xalg}{X_{\mathsf{Glauber}}} + \DTV{X_{\mathsf{Glauber}}}{X_{\mathsf{ideal}}} \le\frac{3\epsilon}{4}. \qedhere
  \end{align*}
\end{proof}

We now have all ingredients to show \Cref{theorem-sampling-main}.

\begin{proof}[Proof of \Cref{theorem-sampling-main}]
We first assume $2^k \geq (2\mathrm{e}dk)^{\frac{6 \ln 2 \cdot (1+\alpha-\beta)}{(1-\alpha-\beta)^2}}$.
Since we use the algorithm in Lemma~\ref{lemma-MT} with $\delta = \frac{\eps}{4}$ to construct the set $\+M$, we have
\begin{align*}
\Pr{\text{the set $\+M$ satisfying Condition~\ref{condition-marked-variables} is constructed successfully}}	 \geq 1 - \frac{\epsilon}{4}.
\end{align*}
%
%
Let $X_{\mathsf{out}} \in \{0,1\}^V$ be the final assignment returned by our algorithm. 
If our algorithm fails to construct the set $\+M$, then $X_{\mathsf{out}}$ is an arbitrary assignment in $\{0,1\}^V$; otherwise $X_{\mathsf{out}} =\Xalg$.
Adding all errors together, \Cref{prop:coupling} implies that
\begin{align*}
\DTV{X_{\mathsf{out}}}{\mu} \leq \epsilon.	
\end{align*}

Finally, we set the parameters $k_{\alpha},k_{\beta}$ in Condition~\ref{condition-marked-variables} and $\eta$ in~\eqref{eq-def-eta}. We list all the constraints together
\begin{align*}
2^k &\geq (2\mathrm{e}dk)^{\frac{6 \ln 2 \cdot (1+\alpha-\beta)}{(1-\alpha-\beta)^2}},\quad	 \text{where } \alpha = \frac{k_{\alpha}}{k}, \beta = \frac{k_{\beta}}{k};\\
2^{k_\alpha} &\geq 4\mathrm{e}^2d^2k^2;\\
2^{k_\beta} &\geq \frac{20}{\eta}\mathrm{e}dk, \quad\text{where } 0<\eta < 1;\\
2^{k_{\beta}} &\geq 2^{16}d^9k^9;\\
k_\alpha &\geq 1;\\
k_\beta &\geq 1;\\
k_\alpha + k_\beta &\leq k.
\end{align*}
We can take 
\begin{equation}
\label{eq-set-parameter}
\begin{split}
  k_{\alpha} &= \lfloor 0.1133k \rfloor,\\
  k_{\beta}  &= \lfloor 0.5097k \rfloor.
\end{split}
\end{equation}
For any $\xi \geq 0$, if
\begin{align}
\label{eq-condition-k}
k \geq 20\log k + 20\log d + 60 + \xi,
\end{align}
then it must hold that $k \geq 60$ and all the constraints are satisfied with $k_{\alpha}$ and $k_{\beta}$ set as in~\eqref{eq-set-parameter}. We can set $\eta$ as
\begin{align}
\label{eq-def-eta-proof}
\eta \triangleq  \left(\frac{1}{2} \right)^{20+\xi/3}\left( \frac{1}{dk} \right)^9.
\end{align}
Note that~\eqref{eq-condition-k} implies $2^k \geq 2^{\xi + 60}d^{20}k^{20}$. We can verify that
\begin{align*}
\frac{20}{\eta}edk = 20\mathrm{e} \cdot {2^{20 + \xi/3}}d^{10}k^{10} \leq {2^{30 + \xi/2 -1}}d^{10}k^{10} \leq 2^{\frac{k}{2}-1} \leq 2^{k_{\beta}}.  	
\end{align*}

We then analyze the time complexity of our algorithm.
Since we run the algorithm in Lemma~\ref{lemma-MT} with $\delta = \frac{\eps}{4}$, then its  time complexity is at most
\begin{align*}
T_{\mathsf{mark}} = O\left(ndk\log \frac{4}{\epsilon}\right).	
\end{align*}
In Algorithm~\ref{alg-mcmc}, 
the first $T\defeq\Tmix$ calls of the subroutine $\sample(\Phi,\delta,X,S)$ satisfy $|S| = 1$ and the last call the  of the subroutine $\sample(\Phi,\delta,X,S)$ satisfies $|S| \leq n$. 
By \Cref{lemma-sample-correctness}, we have
\begin{align*}
T_{\mathsf{alg}} = O\left(T \left(\frac{n}{\delta}\right)^{\frac{\eta}{10}} d^2k^3 \log^2 \frac{n}{\delta} \right) + O\left(n \left(\frac{n}{\delta}\right)^{\frac{\eta}{10}} d^2k^3 \log^2 \frac{n}{\delta} \right), 	
\end{align*}
where $T = \Tmix$, $\delta = \frac{\epsilon}{4(T+1)}$ and $\eta$ is defined in~\eqref{eq-def-eta-proof}. Note that
\begin{align*}
\left(\frac{n}{\delta}\right)^{\frac{\eta}{10}} = O\left(\left( \frac{n}{\epsilon} \right)^{\eta}\right).
\end{align*}
This implies
\begin{align*}
T_{\mathsf{alg}} = 	O\left(n \left(\frac{n}{\epsilon}\right)^{\eta} d^2k^3 \log^3 \frac{n}{\epsilon} \right).
\end{align*}
The total time complexity of our algorithm is
\begin{align*}
  T = T_{\mathsf{mark}}  + 	T_{\mathsf{alg}} 
  & = O\left(n \left(\frac{n}{\epsilon}\right)^{\eta} d^2k^3 \log^3 \frac{n}{\epsilon} \right). \qedhere
\end{align*}
\end{proof}

\section{Approximate counting}
\label{section-counting}

Let  $\Phi = (V, C)$ be a $k$-CNF formula. 
One way to reduce counting to sampling is to start from a CNF formula with $n$ variables and no clause.
Then add clauses one by one and use the self-reducibility~\cite{jerrum1986random} to count the number of solutions for $\Phi$. 
This standard method gives an approximate counting algorithm which requires $\widetilde{O}(n^2d^2)$ calls to the sampling algorithm 
for a constant $\eps$ ($\widetilde{O}$ hides logarithmic factors).

Instead, we give a faster counting algorithm based on the simulated annealing method~\cite{bezakova2008accelerating,vstefankovivc2009adaptive,huber2015approximation,kolmogorov18faster}.
We will show that a non-adaptive annealing schedule with $\widetilde{O}(nd)$ calls to the sampling algorithm suffices (for a constant $\eps$).
The detailed time complexity bound is given in \Cref{theorem-counting-main}.

\begin{theorem}\label{theorem-counting-main}
The followings hold for all $\xi \geq 0$.
There is an algorithm such that given any $\eps > 0$ and $(k,d)$-formula $\Phi$ with $n$ variables where $k \geq 20\log k + 20\log d + 60 + \xi$,
it outputs a number $\widehat{Z}$ that satisfies $\exp(-\eps)Z \leq  \widehat{Z} \leq \exp(\eps) Z$ with probability at least $\frac{3}{4}$, where $Z$ is the number of satisfying assignments of $\Phi$.
The algorithm terminates in time $O\left( \left(\frac{n}{\epsilon}\right)^{2+\eta} d^{3}k^3\log^{4+\eta}\frac{nd}{\epsilon} \right)$, where $\eta = \left(\frac{1}{2} \right)^{19+\xi/3}\left( \frac{1}{dk} \right)^9.$
\end{theorem}
The counting result in Theorem~\ref{theorem-sample-simplified} is a corollary of Theorem~\ref{theorem-counting-main}. We can set the parameter $\zeta$ in Theorem~\ref{theorem-sample-simplified} as $\zeta = \tp{\frac{1}{2}}^{20+\xi/3}$. The running time of the counting algorithm in Theorem~\ref{theorem-counting-main} is
\begin{align*}
O\left( \left(\frac{n}{\epsilon}\right)^{2+2\zeta(dk)^{-9}} d^{3}k^3\log^{4+2\zeta(dk)^{-9}}\frac{nd}{\epsilon} \right) = \widetilde{O}\left( d^{3}k^3 \left(\frac{n}{\epsilon}\right)^{2+\zeta}  \right),
\end{align*}
where the equation holds due to $2(dk)^{-9}\leq 1$.
\subsection{The counting algorithm}

Recall $\Phi = (V, C)$ is a $k$-CNF formula. 
Given any parameter $\theta > 0$, 
for any $X \in \{0,1\}^V$,
define the weight function:
\begin{align*}
  w_\theta(X) \triangleq \exp(-\theta |F(X)| ),
\end{align*}
where $F(X) \subseteq C$ is the set of clauses that are not satisfied by $X$.
Let the partition function $Z(\theta)$ be
\begin{align*}
  Z(\theta) \triangleq \sum_{X \in \{0,1\}^V}w_\theta	(X).
\end{align*}
Then the Gibbs distribution $\mu_{\theta}$ over $\{0,1\}^V$ is given by
\begin{align}
  \label{eq-def-Gibbs}
  \forall X \in \{0,1\}^V: \quad \mu_\theta(X) \triangleq \frac{w_\theta(X) }{Z(\theta)},
\end{align}
Let $Z$ denote the number of satisfying assignments for $\Phi$, then we have
\begin{align*}
Z = \lim_{\theta \rightarrow \infty}Z(\theta).	
\end{align*}
Let $\ell = nd\left\lceil \ln \frac{4nd}{\epsilon}\right\rceil$. Define a sequence of parameters $(\theta_i)_{i \geq 0}$ as
\begin{align}
\label{eq-def-p}
\forall i \in \mathbb{Z}_{\geq 0}:\quad \theta_i = \frac{i}{dn}. 	
\end{align}
The following lemma shows that the partition function $Z(\theta_{\ell})$ is close to $Z$.
\begin{lemma}
\label{lemma-count-target}
If $2^k \geq 2\mathrm{e}dk$, then given any $\epsilon >0$, it holds that
\begin{align*}
Z	\leq Z( \theta_{\ell}) \leq \exp\left( \frac{\epsilon}{2} \right)Z.
\end{align*}
\end{lemma}
\noindent
The proof of Lemma~\ref{lemma-count-target} is deferred to \Cref{section-proof-approx}. Note that the condition for $\Phi$ in Lemma~\ref{lemma-count-target} is weaker than that in Theorem~\ref{theorem-counting-main}.
By Lemma~\ref{lemma-count-target}, we can use $Z( \theta_{\ell})$ to approximate the value of $Z$. We estimate the value of $Z( \theta_{\ell})$ by the following telescoping product
\begin{align}
\label{eq-telescope}
Z(\theta_{\ell}) = \frac{Z(\theta_{\ell})}{Z(\theta_{\ell - 1})} \times \frac{Z(\theta_{\ell-1})}{Z(\theta_{\ell - 2})}\times \ldots \times \frac{Z(\theta_1)}{Z(\theta_0)} \times 2^n,	
\end{align}
where the equation holds because $\theta_0 = 0$ and $Z(\theta_0)= 2^n$.

We now estimate the value of each ratio $\frac{Z(\theta_{i+1})}{Z(\theta_i)}$ in~\eqref{eq-telescope}.
Let $\mu_i = \mu_{\theta_i}$ denote the Gibbs distribution specified by the parameter $\theta_i$. Let $w_i(\cdot)=w_{\theta_i}(\cdot)$ denote the weight function for Gibbs distribution $\mu_i$.
For each $1\leq i \leq \ell$, we define the random variable $W_i$ as 
\begin{align*}
W_i \triangleq \frac{w_{i}(X)}{w_{i-1}(X)}, \quad\text{where } X \sim \mu_{i-1}.	
\end{align*}
We then define $W$ as $2^n$ times the product of all random variables $W_i$:
\begin{align*}
W = 2^n \prod_{i = 1}^{\ell}W_i.	
\end{align*}
We have the following lemma for $W$ and each $W_i$.
\begin{lemma}
\label{lemma-counting-exp}
For each $1\leq i \leq \ell$, the random variable $W_i$ satisfies
\begin{align*}
\E{W_i}	= \frac{Z(\theta_{i})}{Z(\theta_{i-1})},\quad \E{W_i^2} = \frac{Z(\theta_{i+1})}{Z(\theta_{i-1})}.
\end{align*}
Hence, the random variable $W$ satisfies
\begin{align*}
\E{W} = Z(\theta_{\ell}), \quad \E{W^2} = \frac{4^n Z(\theta_{\ell})Z(\theta_{\ell + 1})}{Z(\theta_0)Z(\theta_1)}.	
\end{align*}
\end{lemma}
\begin{proof}
By the definition of $W_i$, we have	
\begin{align*}
\E{W_i} = \sum_{X \in \{0,1\}^V}\frac{w_{i-1}(X)}{Z(\theta_{i-1})}	\times \frac{w_i(X)}{w_{i-1}(X)} = \frac{Z(\theta_i)}{Z(\theta_{i-1})}.
\end{align*}
For each $X \in \{0,1\}^V$, it holds that $w_i(X) = \exp\left( -\frac{i}{dn}|F(X)| \right)$. We have
\begin{align*}
\E{W_i^2} &= \sum_{X \in \{0,1\}^V}\frac{w_{i-1}(X)}{Z(\theta_{i-1})}	\times  \left(\frac{w_i(X)}{w_{i-1}(X)}\right)^2=\sum_{X \in \{0,1\}^V}\frac{w_{i+1}(X)}{Z(\theta_{i-1})}=\frac{Z(\theta_{i+1})}{Z(\theta_{i-1})}.
\end{align*}
Note that all $W_i$ are independent.
By the definition of $W$, we have
\begin{align*}
\E{W} &= 2^n \prod_{i = 1}^{\ell}\E{W_i} =  2^n \times \frac{Z(\theta_{\ell})}{Z(\theta_0)} = Z(\theta(\ell));\\
\E{W^2} &= 4^n \times \prod_{i=1}^{\ell}\E{W_i^2}	 = 4^n \times \frac{Z(\theta_{\ell})Z(\theta_{\ell + 1})}{Z(\theta_0)Z(\theta_1)}. \qedhere
\end{align*}
\end{proof}

By Lemma~\ref{lemma-counting-exp}, the expectation of $W$ is precisely the partition function $Z(\theta_{\ell})$. 
If we can draw random samples from each distribution $\mu_i$, then we can compute all $W_i$ and $W$ using these random samples. 
In Section~\ref{section-alg}, we have given an algorithm that samples CNF solutions uniformly at random. 
With a simple modification, we have the following algorithm that samples assignments from the Gibbs distribution in~\eqref{eq-def-Gibbs}.

\begin{lemma}\label{lemma-sample-weighted} 
Let $\xi \geq 0$ and $\Phi$ be a $(k,d)$-formula with $n$ variables where $k \geq 20\log k + 20\log d + 60 + \xi$.
There is an algorithm  $\mathcal{A}$ such that given any $0<\delta<1$ and any $\theta \geq 0$, the algorithm  $\mathcal{A}(\theta, \delta)$ outputs a random assignment $X$ of $\Phi$ satisfying
  $d_{\-{TV}}(X,\mu_\theta)\le \delta$, where $\mu_\theta$ is the Gibbs distribution defined in~\eqref{eq-def-Gibbs}.
  The algorithm terminates in time $O\left(n \left(\frac{n}{\delta}\right)^{\eta} d^2k^3 \log^3 \frac{n}{\delta} \right)$, where $\eta = \left(\frac{1}{2} \right)^{20+\xi/3}\left( \frac{1}{dk} \right)^9$.
\end{lemma}

Our counting algorithm is described in Algorithm~\ref{alg-count}.
It relies on the Algorithm $\mathcal{A}$ in \Cref{lemma-sample-weighted} as a subroutine.

\begin{algorithm}[ht]
  \SetKwInOut{Input}{Input} \SetKwInOut{Output}{Output} \Input{a CNF formula $\Phi=(V,C)$, a
    parameter $\epsilon > 0$.}  \Output{a number $\widehat{Z}$.} 
     \For{each $j$ from 1 to $m=\lceil 144\epsilon^{-2} \rceil$}{ 
     	\For{each $i = 1$ to $\ell = nd\left\lceil \ln \frac{4nd}{\epsilon}\right\rceil$}{
     		use $\mathcal{A}(\theta_{i-1}, 1/(8\ell m))$ to draw sample $X_i^j \in \{0,1\}^V$	 independently\label{line-sample-count}\;
     		$\widehat{W}_i^j \gets w_i(X_i^j) / w_{i-1}(X_i^j)$\;
     	}
     	$\widehat{W}^j \gets 2^n\cdot\prod_{i=1}^{\ell}\widehat{W}_i^j$\;
     }
     \Return{$\widehat{Z} = \frac{1}{m}\sum_{j=1}^m \widehat{W}^j$;}
  \caption{The counting algorithm}\label{alg-count}
\end{algorithm}

To prove that correctness of Algorithm~\ref{alg-count}, we need the following lemma.
\begin{lemma}
\label{lemma-count-correct}
Let $\mathcal{B}$ be a sampling oracle such that given any parameter $\theta$, $\+B(\theta)$ returns a perfect sample from the distribution $\mu_\theta$.
Suppose we replace $\mathcal{A}(\theta_{i-1}, 1/(8\ell m))$ in Line~\ref{line-sample-count} of Algorithm~\ref{alg-count} with $\+B(\theta_{i-1})$. 
Denote the output of the modified algorithm by $\widehat{Z}_{\+B}$. 
Then, it holds that
\begin{align*}
  \Pr{\exp(-\epsilon / 2) Z(\theta_\ell)\leq \widehat{Z}_{\+B}  \leq \exp(\epsilon / 2) Z(\theta_\ell) }	\geq 7/8.
\end{align*}
\end{lemma}
\begin{proof}
By the assumption in Lemma~\ref{lemma-count-correct}, 
we know that each $\widehat{W}^j$ is a perfect sample from the distribution of the random variable $W$. 
Note that $\widehat{Z}_{\+B} = \frac{1}{m}\sum_{i=1}^m \widehat{W}^i$.
Hence $\E{\widehat{Z}_{\+B}} = \E{W} = Z(\theta_{\ell})$.
By Chebyshev's inequality, we have
\begin{align}
\label{eq-concentration}
\Pr{ \left|\widehat{Z}_{\+B}- \E{\widehat{Z}_{\+B}}\right| \geq (\epsilon / 3 )\E{\widehat{Z}_{\+B}}} \leq \frac{9\Var{\widehat{Z}_{\+B}}}{\epsilon^2\E{\widehat{Z}_{\+B}}^2} = \frac{9\Var{W}}{m\epsilon^2\E{W}^2}
\end{align}
By Lemma~\ref{lemma-counting-exp}, we have
\begin{align*}
\frac{\Var{W}}{\E{W}^2} = \frac{\E{W^2}}{\E{W}^2}-1 = \frac{Z(\theta_{\ell + 1})Z(\theta_0)}{Z(\theta_{\ell})Z(\theta_1)}-1,
\end{align*}
where the last equation holds because $\E{W} = Z(\theta_{\ell})$, $\E{W^2} = \frac{4^n Z(\theta_{\ell})Z(\theta_{\ell + 1})}{Z(\theta_0)Z(\theta_1)}$ and $Z(\theta_0) = 2^n$. Note that $Z(\theta_{\ell + 1}) \leq Z(\theta_\ell)$, we have
\begin{align*}
\frac{Z(\theta_{\ell + 1})}{Z(\theta_{\ell})} \leq 1.	
\end{align*}
By the definition of the partition function $Z(\cdot)$, we have
\begin{align*}
\frac{Z(\theta_0)}{Z(\theta_1)} =  \frac{\sum_{X \in \{0,1\}^V}w_0(X)
  }{\sum_{X \in \{0,1\}^V}w_1(X) } =  \frac{\sum_{X \in \{0,1\}^V}1
  }{\sum_{X \in \{0,1\}^V}\exp(-\theta_1|F(X)|)}
  \le \max_{X\in\set{0,1}^V}\exp(\theta_1\abs{F(X)})\le \-e.	
\end{align*}
The last inequality is due to the fact that $\theta_1=\frac{1}{nd}$
and $\abs{F(X)}\le nd$.

Hence, we can bound~\eqref{eq-concentration} as follows
\begin{align*}
\Pr{ \left|\widehat{Z}_{\+B}- \E{\widehat{Z}_{\+B}}\right| \leq (\epsilon / 3 )\E{\widehat{Z}_{\+B}}}  \leq \frac{9(\mathrm{e} - 1)}{m\epsilon^2} \leq \frac{1}{8},
\end{align*}
where the last inequality holds because $m=\lceil 144\epsilon^{-2} \rceil$.
Note that $\E{\widehat{Z}_{\+B}} = Z(\theta_{\ell})$ due to the assumption in Lemma~\ref{lemma-count-correct}. This proves that
\begin{align*}
  &\phantom{{}={}}\Pr{\exp(-\epsilon / 2) Z(\theta_\ell)\leq \widehat{Z}_{\+B}  \leq \exp(\epsilon / 2) Z(\theta_\ell) }\\
  &\geq \Pr{(1-\epsilon / 3) Z(\theta_\ell)\leq \widehat{Z}_{\+B}  \leq (1+\epsilon / 3) Z(\theta_\ell) } \geq \frac{7}{8}. \qedhere
\end{align*}
\end{proof}

 We then construct a coupling $\+C$ between Algorithm~\ref{alg-count} and the algorithm in Lemma~\ref{lemma-count-correct}. 
 For each execution of Line~\ref{line-sample-count}, we use the optimal coupling to couple the random sample returned by algorithm $\mathcal{A}(\theta_{i-1}, 1/(8\ell m))$ 
 and the random sample returned by $\+B(\theta_{i-1})$. 
 By \ref{prop:coupling}, with probability at least $1-1/(8\ell m)$, two samples are perfectly coupled. 
 Since Line~\ref{line-sample-count} is executed for $\ell m$ times, then by a union bound, with probability at least $\frac{7}{8}$, two algorithms obtain the same output, i.e. $\widehat{Z} = \widehat{Z}_{\+B}$. 
 Combining with Lemma~\ref{lemma-count-correct}, we have
\begin{align*}
&\phantom{{}={}}\Pr{\widehat{Z} < \exp(-\epsilon / 2) Z(\theta_\ell) \lor \widehat{Z} > \exp(\epsilon / 2) Z(\theta_\ell) }\\
&\leq\Pr{\widehat{Z}_{\+B} < \exp(-\epsilon / 2) Z(\theta_\ell) \lor \widehat{Z}_{\+B} > \exp(\epsilon / 2) Z(\theta_\ell) } + \Pr[\+C]{\widehat{Z} \neq \widehat{Z}_{\+B}}\\
&\leq \frac{1}{4}.
\end{align*}
This proves that
\begin{align*}
\Pr{\exp(-\epsilon / 2) Z(\theta_\ell)\leq \widehat{Z}  \leq \exp(\epsilon / 2) Z(\theta_\ell) }	\geq 3/4.
\end{align*}
By Lemma~\ref{lemma-count-target}, we know that $Z(\theta_\ell)$ approximates the value of $Z$, where $Z$ is the number of solutions for $\Phi$.
We have
\begin{align*}
\Pr{\exp(-\epsilon) Z\leq \widehat{Z}  \leq \exp(\epsilon) Z }	\geq 3/4.
\end{align*}
This proves the correctness of Algorithm~\ref{alg-count}.

The time complexity of Algorithm~\ref{alg-count} is dominated by the time complexity of generating random samples. In Algorithm~\ref{alg-count}, the sampling algorithm $\mathcal{A}$ in Lemma~\ref{lemma-sample-weighted} is called for $m\ell$ times. 
Note that we only call algorithm $\mathcal{A}$ with $\delta = 1/(8m\ell)$.
The time complexity of each call of $\mathcal{A}$ is
\begin{align*}
T_{\mathcal{A}} =O\left( n\left(\frac{n}{\epsilon} \right)^{2\eta}d^{2+\eta}k^3 \log^{3+\eta}\frac{nd}{\epsilon}\right),
\end{align*}
where $\eta = \left(\frac{1}{2} \right)^{20+\xi/3}\left( \frac{1}{dk} \right)^9$. Note that $m\ell = O\left(\frac{nd}{\epsilon^2}\log \frac{nd}{\epsilon} \right)$. Then, the total time complexity of Algorithm~\ref{alg-count} is at most 
\begin{align*}
T_{\mathsf{count}} = O\left( \left( \frac{n}{\epsilon} \right)^{2} \left( \frac{n}{\epsilon} \right)^{2\eta} d^{3+\eta}k^3\log^{4+\eta}\frac{nd}{\epsilon} \right).		
\end{align*}
Let $\eta' = 2\eta = \left(\frac{1}{2} \right)^{19+\xi/3}\left( \frac{1}{dk} \right)^9$, we have
\begin{align*}
T_{\mathsf{count}} = O\left( \left( \frac{n}{\epsilon} \right)^{2} \left( \frac{nd}{\epsilon} \right)^{\eta'} d^{3}k^3\log^{4+\eta'}\frac{nd}{\epsilon} \right) = 	 O\left( \left( \frac{n}{\epsilon} \right)^{2+\eta'} d^{3}k^3\log^{4+\eta'}\frac{nd}{\epsilon} \right),
\end{align*}
where the last equation holds due to $d^{d^{-9}} = O(1)$.
This proves the time complexity of Algorithm~\ref{alg-count}.

\subsection{Comparing \texorpdfstring{$Z$}{Z} and \texorpdfstring{$Z(\theta_\ell)$}{Z(theta\_l)} (proof of \texorpdfstring{\Cref{lemma-count-target}}{Lemma 6.2})}
\label{section-proof-approx}

We first prove a lemma stating that adding a new clause to a CNF formula decrease the number of solutions by at most half if the parameters are in the local lemma regime.

\begin{lemma}
\label{lemma-LLL-decay}
Let $\Phi = (V, C)$ be a $k$-CNF formula. Let $\Phi' = (V, C')$ be a new $k$-CNF formula obtained from $\Phi$ by adding a new clause $f$, i.e. $C' = C \cup \{f\}$. Suppose each variable belongs to at most $d$ clauses in both $\Phi$ and $\Phi'$. If $2^k \geq 2\mathrm{e}dk$, then it holds that
\begin{align*}
\frac{Z_{\Phi'}}{Z_\Phi}\geq \frac{1}{2},
\end{align*}
where $Z_\Phi$ is the number of solution for $\Phi$ and $Z_{\Phi'}$ is the number of solutions for $\Phi'$.
\end{lemma}
\begin{proof}
Let $\mu$ and $\mu'$ denote the uniform distributions of all solutions for $\Phi$ and $\Phi'$
respectively. Note that if $X \in \{0,1\}^V$ is a solution for $\Phi'$, then it is a solution for
$\Phi$ as well. Therefore, we have
\begin{align}
\label{eq-ratio-Z}
\frac{Z_{\Phi'}}{Z_{\Phi}} = \Pr[X\sim \mu]{X \text{ is a solution for } \Phi'} =\Pr[X \sim \mu]{f \text{ is satisfied by } X}.
\end{align}
Recall that we use $\Pr[\+P]{\cdot}$ to denote the product distribution such that each variable $v \in V$ takes a value from $\{0,1\}$ uniformly and independently. Let $\+B_{c}$ denote the bad event that the clause $c \in C$ is not satisfied. Note that, in $\Phi$, each clause contains $k$ variables and each variable belongs to at most $d$ clauses. 
By Theorem~\ref{theorem-LLL}, if we take $x(\+B_c) = \frac{1}{2dk}$ for each $\+B_c$, it holds that for any $c \in C$,
\begin{align*}
\Pr[\+P]{\+B_{c}} = \left(\frac{1}{2} \right)^k	\leq \frac{1}{2\mathrm{e}dk} \leq x(\+B_c)\prod_{\+B_{c'} \in \Gamma(\+B_c)}\left(1 - x(\+B_{c'}) \right),
\end{align*}
where $\Gamma(\+B_c)$ contains all $\+B_{c'}$ satisfying $c' \in C$, $c' \neq c$ and $\vbl{c} \cap \vbl{c'} \neq \emptyset$. We use $\+F$ to denote the event that $f$ is not satisfied. Since each variable belongs to at most $d$ clauses in $\Phi'$, we have
\begin{align}
\label{eq-ratio-2}
\Pr[\+P]{\+F\mid \bigwedge_{c\in C}\overline{\+B_c}} \leq \Pr[\+P]{\+F}\left(1 - \frac{1}{2dk}  \right)^{-dk} \leq 2 \left(\frac{1}{2}\right)^k \leq \frac{1}{2},
\end{align}
where the last inequality holds because $k \geq 2$ if $2^k \geq 2\mathrm{e}dk$. 
Note that the product distribution $\+P$ conditioned on $ \bigwedge_{c\in C} \overline{\+B_c}$ is precisely the distribution $\mu$. 
Combining~\eqref{eq-ratio-Z} and~\eqref{eq-ratio-2}, we have 
\begin{align*}
  \frac{Z_{\Phi'}}{Z_{\Phi}} & = 1 - \Pr[X \sim \mu]{f \text{ is not satisfied by } X} = 1 - \Pr[\+P]{\+F \mid \bigwedge_{c \in C} \overline{\+B_c}} \geq \frac{1}{2}. \qedhere
\end{align*}
\end{proof}

For a $k$-CNF formula $\Phi=(V, C)$ and any subset $S \subseteq C$, we define the value $Z_S$ as
\begin{align}
\label{eq-def-Z_S}
Z_S \triangleq \#\left\{X \in \{0,1\}^V \,\big |  \,\substack{\text{all clauses in $S$ are not satisfied by $X$,}\\ \text{and all clauses in $C \setminus S$ are satisfied by $X$} } \right\}.	
\end{align}
The value $Z_S$ counts the number of those assignments $X \in \{0,1\}^V$ satisfying \emph{exactly}
the clauses in $C\setminus S$. The next lemma bounds the size of $Z_S$.

\begin{lemma}
\label{lemma-ZS}
Suppose each variable belongs to at most $d$ clauses. If $2^k \geq 2\mathrm{e}dk$, then for any $S \subseteq C$, it holds that
$Z_S \leq 2^{|S|}Z_\emptyset$.	
\end{lemma}

\begin{proof}
Let $S \subseteq C$ be a set of clauses with $|S| = k$. Suppose $S = \{c_1,c_2,\ldots,c_k\}$. We define a sequence of CNF formulas $\Phi_0,\Phi_1,\ldots,\Phi_k$. For each $\Phi_i = (V, C_i)$, the set of clauses $C_i$ is defined as
\begin{align*}
C_i \triangleq (C \setminus S) \cup \{c_j \mid 1\leq j \leq i \}.
\end{align*}
This is equivalent to let $C_0=C\setminus S$ and $C_i=C_{i-1}\cup \set{c_i}$ for every $1\le i\le
k$. For every $0\le i\le k-1$, since each $\Phi_i$ is a subformula of $\Phi$, the condition $2^k\ge 2edk$ still holds and we
can apply \Cref{lemma-LLL-decay} for each $\Phi_i$ (with $\Phi=\Phi_i$ and $\Phi'=\Phi_{i+1}$
in the statement of \Cref{lemma-LLL-decay}).  This yields


\begin{equation}\label{eqn-ratio}
\frac{Z_{\Phi_k}}{Z_{\Phi_0}} = \prod_{j=1}^k \frac{Z_{\Phi_j}}{Z_{\Phi_{j-1}}}\geq 2^{-k},	
\end{equation}
where $Z_{{\Phi}_i}$ is the number of solutions for $\Phi_i$.

On the other hand, by the definition of $Z_S$, we have
\begin{align*}
  Z_{\Phi_k}=Z_\Phi=Z_\varnothing,\;\mbox{and}\; Z_{\Phi_0}=\sum_{S'\subseteq S}Z_{S'}.  
\end{align*}
Combining with \Cref{eqn-ratio}, we obtain
\begin{align*}
  \frac{Z_\varnothing}{\sum_{S'\subseteq S}Z_{S'}}\ge 2^{-k}.  
\end{align*}
Hence, we have
\begin{align*}
  Z_S &\leq \sum_{S' \subseteq S} Z_{S'} \leq 2^k Z_{\emptyset}.\qedhere
\end{align*}
\end{proof}

We now prove \Cref{lemma-count-target}. By the definition of $Z(\theta_{\ell})$, the lower bound
$Z(\theta_{\ell}) \geq Z$ clearly holds. For the upper bonud, noting that $\theta_{\ell} = \left\lceil \ln \frac{4nd}{\epsilon}\right\rceil$, we have
\begin{align*}
Z(\theta_{\ell}) = \sum_{X \in \{0,1\}^V}\exp(-\theta_{\ell}F(X)) \leq \sum_{X \in \{0,1\}^V}\left( \frac{\epsilon}{4nd} \right)^{|F(X)|},
\end{align*}
where $F(X) \subseteq C$ is the set of clauses that are not satisfied by $X$. By the definition of $Z_{S}$, we have
\begin{align*}
  Z(\theta_{\ell})  &\leq \sum_{X \in \{0,1\}^V}\left( \frac{\epsilon}{4nd} \right)^{|F(X)|}
  =\sum_{S \subseteq C}\left( \frac{\epsilon}{4nd} \right)^{|S|}Z_S\\
  \tag{by \Cref{lemma-ZS}}\qquad &\leq\sum_{S \subseteq C}\left( \frac{\epsilon}{4nd} \right)^{|S|}2^{|S|}Z_{\emptyset} 
  = Z_{\emptyset}\sum_{k = 0}^{|C|}\binom{|C|}{k}	\left( \frac{\epsilon}{4nd} \right)^{k}2^{k}\\
  \tag{as $|C| \leq nd$}\qquad& = Z_{\emptyset}\left(1+ \frac{\epsilon}{2nd} \right)^{\abs{C}}
  \leq Z_{\emptyset}\left(1+ \frac{\epsilon}{2nd} \right)^{nd}\\
  \tag{as $Z = Z_{\emptyset}$}\qquad  &\leq Z\exp\left(\frac{\epsilon}{2}\right).
\end{align*}
This finishes the proof of \Cref{lemma-count-target}.




\subsection{The modified sampling algorithm (proof of Lemma~\ref{lemma-sample-weighted})}
In this section, we give a modified sampling algorithm to sample from the Gibbs distribution $\mu_\theta$ defined in~\eqref{eq-def-Gibbs}. Given a CNF formula $\Phi = (V, E)$ and a parameter $\theta \geq 0$, we introduce $|C|$ extra variables 
\begin{align*}
U \triangleq \{ u_c \in \{0,1\}  \mid c \in C \}.	
\end{align*}
We now define a new CNF formula $\Phi'=(V \cup U, C')$. The set of clauses $C'$ is defined as
\begin{align*}
C' \triangleq \{ u_c \lor c \mid c \in C \}.	
\end{align*}
Hence, given any assignment $X \in \{0,1\}^{V \cup U}$, a clause $c' = u_c \lor c$ is satisfied by $X$ if $X(u_c) = 1$ or the clause $c$ is satisfied by $X$.
\begin{observation}
\label{observation-uc}
The CNF formula $\Phi'=(V \cup U, C')$ is $k+1$ uniform  and each variable $u \in U$ belongs to only one clause.	
\end{observation}

Let $\+P$ denote the product distribution over $\{0,1\}^{V \cup U}$ such that each variable $v \in V$ takes value from $\{0,1\}$ uniformly, and each variable $u \in U$ takes value 1 with probability $\exp(-\theta)$ and takes value 0 with probability  $1 - \exp(-\theta)$. For each clause $c' \in C'$, we define a bad event $\+B_{c'}$ as the clause $c'$ is not satisfied. Recall that $\mu_{\theta}$ is the Gibbs distribution defined in~\eqref{eq-def-Gibbs}. We have the following proposition. 
\begin{proposition}
\label{proposition-weighted-CNF}
For any assignment $X \in \{0,1\}^V$, it holds that
\begin{align*}
\Pr[\+P]{\text{each variable $v \in V$ takes the value $X(v)$ } \,\big |\, \bigwedge_{c' \in C'}\overline{\+B_{c'}}} = \mu_{\theta}(X).
\end{align*}
\end{proposition}
\begin{proof}
We use $F(X(V))$ to denote the set of clauses $c \in C$ ($C$ is the set of clauses in original CNF formula $\Phi$) such that $c$ is not satisfied by $X(V)$. 
Consider a clause $c' \in C'$ where $c' = u_c \lor c$.
Suppose the bad event $\+B_{c'}$ does not occur.
If $c \in F(X(V))$, then $u_c$ must take the value 1.
If $c \notin F(X(V))$, then $u_c$ can take an arbitrary value from $\{0,1\}$.
We have
\begin{align*}
&\Pr[\+P]{\text{(each variable $v \in V$ takes the value $X(v)$) } \land   \left(\bigwedge_{c' \in C'}\overline{\+B_{c'}}\right) } \\
=\ & \left(\frac{1}{2}\right)^{|V|} \exp(-\theta |F(X(V))|) =\left(\frac{1}{2}\right)^{|V|} w_{\theta}(X(V)),
\end{align*}
where $w_{\theta}(\cdot)$ is the weight function for
$\mu_{\theta}$. Therefore,
\begin{align*}
  &\phantom{{}={}}\Pr[\+P]{\text{each variable $v \in V$ takes the value $X(v)$ } \,\big |\,
    \bigwedge_{c' \in C'}\overline{\+B_{c'}}}\\
  &=\frac{w_\theta(X(V))}{2^{\abs{V}}\Pr[\+P]{\bigwedge_{c'\in C'}\overline{\+B_{c'}}}}\\
  &=\frac{w_\theta(X(V))}{2^{\abs{V}}\sum_{X'\in\set{0,1}^V}\Pr[\+P]{\text{(each variable $v \in V$
    takes the value $X'(v)$) } \land \bigwedge_{c'\in C'}\overline{\+B_{c'}}}}\\
  &=\frac{w_\theta(X(V))}{\sum_{X'\in\set{0,1}^V}w_\theta(X'(V))} = \mu_\theta(X). \qedhere
\end{align*}
\end{proof}

Let $\mu'$ denote the product distribution $\+P$ over $\{0,1\}^{V \cup U}$ conditioned on none of the bad event $\+B_{c'}$ for $c' \in C'$ occurs. 
Our aim is to sample $X \in \{0,1\}^{V \cup U}$ such that
\begin{align*}
  \DTV{X}{\mu'} \leq \delta,
\end{align*}
which, by Proposition~\ref{proposition-weighted-CNF}, implies that
\begin{align*}
  \DTV{X(V)}{\mu_{\theta}} \leq \delta.		
\end{align*}

Recall that $\Phi'=(V \cup U, C')$ is a $(k+1)$-uniform CNF formula.
We describe how to modify our algorithm in Section~\ref{section-alg} to sample from the Gibbs distribution $\mu'$.

The first step is to mark variables.
We construct a set of marked variables $\+M \subseteq V$ such that each clause $c' \in C'$ contains at least $k_{\alpha}$ marked variables and at least $k_{\beta}$ unmarked variables. 
Note that we do not mark variables in the set $U$, i.e. $U \cap \+M = \emptyset$. 
This step can be accomplished by the Moser-Tardos algorithm in Section~\ref{section-mark}.

We define the Glauber dynamics chain $(X_t)_{t \geq 0}$ for marked variables, whose stationary distribution is the marginal distribution $\mu'_{\+M}$ on $\+M$ projected from $\mu'$. We start with an initial assignment $X_0\in\{0,1\}^{\+M}$ where $X_0(v)$ is uniformly at random for all $v\in\+M$.
In the $t$-th step, the chain evolves as follows:
\begin{itemize}
\item pick $v \in \+M$ uniformly at random and set $X_t(u) \gets X_{t-1}(u)$ for all
  $u \in \+M \setminus \{v\}$;
\item sample $X_t(v) \in \{0, 1\}$ from the distribution
  $\mu_{v}'(\cdot \mid X_{t-1}(\+M \setminus \{v\} ))$.
\end{itemize}

We use Algorithm~\ref{alg-mcmc} to simulate the Glauber dynamics chain defined above. 
There are two modifications.
First, in Line~\ref{line-sample-1}, we use the subroutine $ \sample(\Phi', \frac{\epsilon}{4(T+1)}, X_{t-1}({\+M} \setminus \{v\}), \{v\})$ to draw random sample $X_t(v) \in \{0,1\}$.
Second, in Line~\ref{line-sample-2}, we use the subroutine $\sample(\Phi', \frac{\epsilon}{4(T+1)}, X_{T},  V \cup U \setminus\+M)$ to draw the random sample $X_{ V \cup U \setminus\+M}$.  
  
We also need to adjust the $\sample$ subroutine in the rejection sampling step 
$$Y^X_i \gets \Rejection\left(\Phi^X_i, R \right),$$
namely Line~\ref{line-sample-rejection-sampling} of $\sample(\Phi, \delta, X, S)$. 
Recall that $\Phi^X_i =(V_i^X, C_i^X)$. 
In the rejection sampling, for each variable $v \in V_i^X \cap V$, 
we sample its value from $\{0,1\}$ uniformly and independently; 
for each variable $u \in V_i^X \cap U$, we sample its value from $\{0,1\}$ independently such that $\Pr{u = 1} = \exp(-\theta)$.

We need to verify Lemma~\ref{lemma-mixing} and Lemma~\ref{lemma-sample} for the algorithm above. 
Due to the definition of $\Phi'=(V \cup U, C')$ and Observation~\ref{observation-uc}, the following two facts hold for $\Phi'$:
\begin{itemize}
\item each variable in $\Phi'$ belongs to at most $d$ clauses;
\item for any $c' \in C'$, it holds that $\Pr[\+P]{ \+B_{c'}} = \exp(-\theta)\left(\frac{1}{2} \right)^k \leq \left(\frac{1}{2} \right)^k$.
\end{itemize}
With these two facts, we can verify that all results based on the local lemma still hold for $\Phi'$ with the product distribution $\+P$. 
An analogue to \Cref{lemma-sample} holds by the identical proof in \Cref{section-proof-main}.

To prove the rapid mixing analogue to Lemma~\ref{lemma-mixing}, 
we need to sightly modify the two coupling procedures $\+C_v$ and $\+C$ in Algorithm~\ref{alg-coupling-v} and Algorithm~\ref{alg-coupling}, respectively.
Let $X,Y \in \{0,1\}^{V \cup U}$ be two assignments for path coupling that disagree only on a variable $v_0 \in \+M$.
Recall that $\mu'$ is the Gibbs distribution specified by $\Phi'$. 
Fix a variable $v \in \+M \setminus \{v_0\}$.
We use $\nu'$ to denote the distribution $\mu'$ conditional on the assignment of the set $\Lambda = \+M \setminus \{v_0,v\} $ is specified by $X(\Lambda) = Y(\Lambda)$, 
where $X, Y \in \{0,1\}^{\+M}$ differ at only $v_0$. 
Formally,
\begin{align}
  \label{eq-def-nu'}
  \forall \sigma \in \{0,1\}^{V \cup U}:\quad 
  \nu'(\sigma) =\frac{\one{\sigma(\Lambda)= X(\Lambda)} \cdot \mu'(\sigma) }{ \sum_{\tau \in \{0,1\}^{V \cup U}}\one{\tau(\Lambda)= X(\Lambda)}\cdot\mu'(\tau)}.
\end{align}
We define a hypergraph $H'\triangleq(V, \+E')$ for $\Phi'=(V\cup U, C')$,
which is obtained from $H_{\Phi'}$ (defined in~\eqref{eq-def-Hphi}) by removing all variables in $U$.
Namely, the variable set in $H$ is $V$ rather than $V \cup U$,
and 
the hyperedge set $\+E$ is defined by
\begin{align*}
\+E' \triangleq \{ V \cap \vbl{c'} \mid c' \in C' \}	 = \{ \vbl{c} \mid c \in C \}.
\end{align*}
The two coupling procedures are modified as follows.
\begin{itemize}
  \item \Cref{alg-coupling-v}, $\+C_v$: the input hypergraph is $H'$. 
    In Line~\ref{line-pX-pY-Cv}, 
    we set $p^{X}_u = \nu'_u(0\mid X^{\+C_v})$ and $p^{Y}_u = \nu'_u(0\mid Y^{\+C_v})$, 
    where $\nu'$ is defined in~\eqref{eq-def-nu'}; 
    in Line~\ref{line-extend-Cv}, 
    we use the optimal coupling between $\nu'_{U \cup V_1 \setminus \Vcol}(\cdot\mid X^{\+C_v}(\Vcol\cup V_2))$ and $\nu'_{U \cup V_1 \setminus \Vcol}(\cdot\mid Y^{\+C_v}(\Vcol \cup V_2))$ to extend $X^{\+C_v}$ and $Y^{\+C_v}$ further on the set $U \cup V_1 \setminus \Vcol$.

  \item \Cref{alg-coupling}, $\+C$: the input hypergraph is $H'$.
    In Line~\ref{line-pX-pY-Cv}, we set $p^{X}_u = \mu'_u(0\mid X^{\+C})$ and $p^{Y}_u = \mu'_u(0\mid Y^{\+C})$, 
    where $\mu'$ is the Gibbs distribution defined in~\eqref{eq-def-Gibbs}.
\end{itemize}
In other words, in the while-loop (namely, Line~\ref{line-find-u-Cv}) of $\+C_v$ and $\+C$, we do not choose any variable in $U$.
However, the effect of $U$ needs to be taken into consideration in the calculation of the probabilities $p^{X}_u$ and $p^{Y}_u$ in Line~\ref{line-pX-pY-Cv}.

Consider the modified coupling procedure $\+C_v$. 
Let  $V_1, V_2 , \mathcal{S},\Vcol$ and $X^{\+C_v}, Y^{\+C_v}$ be the sets and assignments after the execution of $\+C_v$. 
Due to Observation~\ref{observation-uc}, each variable $u \in U$ belongs to only one clauses.
Then we can verify that two distributions $\nu'_{V_2\setminus \Vcol}(\cdot \mid X^{\+C_v}(\Vcol)) $ and $\nu'_{V_2\setminus \Vcol}(\cdot \mid X^{\+C_v}(\Vcol \cap V_2))$ are identical, 
and two distributions $\nu'_{V_2\setminus \Vcol}(\cdot \mid Y^{\+C_v}(\Vcol)) $ and $\nu'_{V_2\setminus \Vcol}(\cdot \mid Y^{\+C_v}(\Vcol \cap V_2))$ are identical. 
With these two facts, 
we can prove that $X^{\+C_v}(u) = Y^{\+C_v}(u)$ for all $u \in V_2$. 
Therefore the rapid mixing result, \Cref{lemma-mixing}, follows from the identical proof in Section~\ref{section-proof-mixing}.

\bibliographystyle{alpha} \bibliography{refs.bib}

\newcommand{\etalchar}[1]{$^{#1}$}
\begin{thebibliography}{CGG{\etalchar{+}}19}

\bibitem[BCKL13]{borgs2013left}
Christian Borgs, Jennifer Chayes, Jeff Kahn, and L{\'a}szl{\'o} Lov{\'a}sz.
\newblock Left and right convergence of graphs with bounded degree.
\newblock {\em Random Struct. Algorithms}, 42(1):1--28, 2013.

\bibitem[BD97]{bubley1997path}
Russ Bubley and Martin Dyer.
\newblock Path coupling: A technique for proving rapid mixing in {Markov}
  chains.
\newblock In {\em FOCS}, pages 223--231, 1997.

\bibitem[BGG{\etalchar{+}}19]{BGGGS19}
Ivona Bez{\'{a}}kov{\'{a}}, Andreas Galanis, Leslie~Ann Goldberg, Heng Guo, and
  Daniel {\v{S}}tefankovi{\v{c}}.
\newblock Approximation via correlation decay when strong spatial mixing fails.
\newblock {\em {SIAM} J. Comput.}, 48(2):279--349, 2019.

\bibitem[B{\v{S}}VV08]{bezakova2008accelerating}
Ivona Bez{\'a}kov{\'a}, Daniel {\v{S}}tefankovi{\v{c}}, Vijay~V Vazirani, and
  Eric Vigoda.
\newblock Accelerating simulated annealing for the permanent and combinatorial
  counting problems.
\newblock {\em {SIAM} J. Comput.}, 37(5):1429--1454, 2008.

\bibitem[CGG{\etalchar{+}}19]{CGGPSV19}
Zongchen Chen, Andreas Galanis, Leslie~Ann Goldberg, Will Perkins, James
  Stewart, and Eric Vigoda.
\newblock Fast algorithms at low temperatures via {M}arkov chains.
\newblock In {\em {APPROX-RANDOM}}, volume 145 of {\em LIPIcs}, pages
  41:1--41:14. Schloss Dagstuhl - Leibniz-Zentrum f{\"{u}}r Informatik, 2019.

\bibitem[DFK91]{DFK91}
Martin~E. Dyer, Alan~M. Frieze, and Ravi Kannan.
\newblock A random polynomial time algorithm for approximating the volume of
  convex bodies.
\newblock {\em J. {ACM}}, 38(1):1--17, 1991.

\bibitem[EHS{\etalchar{+}}19]{EHSVY19}
Charilaos Efthymiou, Thomas~P. Hayes, Daniel Stefankovic, Eric Vigoda, and
  Yitong Yin.
\newblock Convergence of {MCMC} and loopy {BP} in the tree uniqueness region
  for the hard-core model.
\newblock {\em {SIAM} J. Comput.}, 48(2):581--643, 2019.

\bibitem[FA17]{FA17}
Alan~M. Frieze and Michael Anastos.
\newblock Randomly coloring simple hypergraphs with fewer colors.
\newblock {\em Inf. Process. Lett.}, 126:39--42, 2017.

\bibitem[GJL19]{GJL19}
Heng Guo, Mark Jerrum, and Jingcheng Liu.
\newblock Uniform sampling through the {L}ov{\'{a}}sz local lemma.
\newblock {\em J. {ACM}}, 66(3):18:1--18:31, 2019.

\bibitem[GLLZ19]{guo2019counting}
Heng Guo, Chao Liao, Pinyan Lu, and Chihao Zhang.
\newblock Counting hypergraph colorings in the local lemma regime.
\newblock {\em {SIAM} J. Comput.}, 48(4):1397--1424, 2019.

\bibitem[HH17]{HH17}
Bernhard Haeupler and David~G. Harris.
\newblock Parallel algorithms and concentration bounds for the {L}ov{\'{a}}sz
  local lemma via witness-{DAG}s.
\newblock In {\em {SODA}}, pages 1170--1187. {SIAM}, 2017.

\bibitem[HPR19]{HPR19}
Tyler Helmuth, Will Perkins, and Guus Regts.
\newblock Algorithmic {P}irogov-{S}inai theory.
\newblock In {\em {STOC}}, pages 1009--1020. {ACM}, 2019.

\bibitem[HSS11]{haeupler2011new}
Bernhard Haeupler, Barna Saha, and Aravind Srinivasan.
\newblock New constructive aspects of the {L}ov{\'a}sz local lemma.
\newblock {\em J. {ACM}}, 58(6):28, 2011.

\bibitem[HSZ19]{HSZ19}
Jonathan Hermon, Allan Sly, and Yumeng Zhang.
\newblock Rapid mixing of hypergraph independent sets.
\newblock {\em Random Struct. Algorithms}, 54(4):730--767, 2019.

\bibitem[Hub15]{huber2015approximation}
Mark Huber.
\newblock Approximation algorithms for the normalizing constant of {G}ibbs
  distributions.
\newblock {\em Ann. Appl. Probab.}, 25(2):974--985, 2015.

\bibitem[JKP19]{JKP19}
Matthew Jenssen, Peter Keevash, and Will Perkins.
\newblock Algorithms for {\#}{BIS}-hard problems on expander graphs.
\newblock In {\em {SODA}}, pages 2235--2247. {SIAM}, 2019.

\bibitem[JSV04]{jerrum2004polynomial}
Mark Jerrum, Alistair Sinclair, and Eric Vigoda.
\newblock A polynomial-time approximation algorithm for the permanent of a
  matrix with nonnegative entries.
\newblock {\em J. {ACM}}, 51(4):671--697, 2004.

\bibitem[JVV86]{jerrum1986random}
Mark~R. Jerrum, Leslie~G. Valiant, and Vijay~V. Vazirani.
\newblock Random generation of combinatorial structures from a uniform
  distribution.
\newblock {\em Theoret. Comput. Sci.}, 43:169--188, 1986.

\bibitem[Kol18]{kolmogorov18faster}
Vladimir Kolmogorov.
\newblock A faster approximation algorithm for the {G}ibbs partition function.
\newblock In {\em COLT}, pages 228--249, 2018.

\bibitem[LP17]{levin2017markov}
David~A Levin and Yuval Peres.
\newblock {\em Markov chains and mixing times}.
\newblock American Mathematical Soc., 2017.

\bibitem[Moi19]{Moi19}
Ankur Moitra.
\newblock Approximate counting, the {L}ov{\'{a}}sz local lemma, and inference
  in graphical models.
\newblock {\em J. {ACM}}, 66(2):10:1--10:25, 2019.

\bibitem[MT10]{moser2010constructive}
Robin~A Moser and G{\'a}bor Tardos.
\newblock A constructive proof of the general {L}ov{\'a}sz local lemma.
\newblock {\em J. {ACM}}, 57(2):11, 2010.

\bibitem[MU17]{mitzenmacher2017probability}
Michael Mitzenmacher and Eli Upfal.
\newblock {\em Probability and computing: randomization and probabilistic
  techniques in algorithms and data analysis}.
\newblock Cambridge university press, 2017.

\bibitem[{\v{S}}VV09]{vstefankovivc2009adaptive}
Daniel {\v{S}}tefankovi{\v{c}}, Santosh Vempala, and Eric Vigoda.
\newblock Adaptive simulated annealing: A near-optimal connection between
  sampling and counting.
\newblock {\em J. {ACM}}, 56(3):18, 2009.

\bibitem[Wei06]{Wei06}
Dror Weitz.
\newblock Counting independent sets up to the tree threshold.
\newblock In {\em {STOC}}, pages 140--149. {ACM}, 2006.

\bibitem[Wig19]{wigderson2019book}
Avi Wigderson.
\newblock {\em Mathematics and Computation: A Theory Revolutionizing Technology
  and Science}.
\newblock Princeton University Press, 2019.

\end{thebibliography}

\end{document}